\documentclass[draftcls,11pt]{IEEEtran}
\onecolumn 
\usepackage[cmex10]{amsmath}
\usepackage{xcolor}
\usepackage[multiple]{footmisc}

\usepackage{mdwtab}
\usepackage{eqparbox}
\usepackage{tabularx}
\usepackage{graphicx,amssymb,rangecite,upgreek,graphicx,dsfont,mathrsfs}

\usepackage{algorithm}
\usepackage{algorithmic} 
\usepackage[usenames,dvipsnames]{pstricks}
 \usepackage{epsfig}

 \usepackage{pst-grad} 
 \usepackage{pst-plot} 
\newcommand {\exe} {\stackrel{\cdot} {=}}

\newcommand {\bze} {\mbox{\boldmath $0$}}

\newcommand {\ba} {\mbox{\boldmath $a$}}
\newcommand {\bb} {\mbox{\boldmath $b$}}

\newcommand {\bh} {\mbox{\boldmath $h$}}

\newcommand {\bet} {\mbox{\boldmath $e$}}

\newcommand {\br} {\mbox{\boldmath $r$}}
\newcommand {\bs} {\mbox{\boldmath $s$}}
\newcommand {\bst} {\mbox{\footnotesize\boldmath $s$}}
\newcommand {\but} {\mbox{\footnotesize\boldmath $u$}}
\newcommand {\brt} {\mbox{\footnotesize\boldmath $r$}}

\newcommand {\bu} {\mbox{\boldmath $u$}}
\newcommand {\bv} {\mbox{\boldmath $v$}}

\newcommand {\bx} {\mbox{\boldmath $x$}}
\newcommand {\by} {\mbox{\boldmath $y$}}
\newcommand {\bz} {\mbox{\boldmath $z$}}
\newcommand {\bzti} {\tilde{\bz}}
\newcommand {\bA} {\mbox{\boldmath $A$}}
\newcommand {\bB} {\mbox{\boldmath $B$}}

\newcommand {\bD} {\mbox{\boldmath $D$}}
\newcommand {\bE} {\mathbb{E}}
\newcommand {\pr} {\mathbb{P}}
\newcommand {\bThet} {\mbox{\boldmath $\Theta$}}

\newcommand {\bG} {\mbox{\boldmath $G$}}
\newcommand {\bH} {\mbox{\boldmath $H$}}
\newcommand {\bI} {\mbox{\boldmath $I$}}

\newcommand {\bN} {\mbox{\boldmath $N$}}

\newcommand {\bQ} {\mbox{\boldmath $Q$}}
\newcommand {\bR} {\mbox{\boldmath $R$}}
\newcommand {\bXt} {\mbox{\boldmath \footnotesize $X$}}
\newcommand {\bYt} {\mbox{\boldmath \footnotesize $Y$}}
\newcommand {\bS} {\mbox{\boldmath $S$}}

\newcommand {\bU} {\mbox{\boldmath $U$}}
\newcommand {\calHs} {\mbox{\boldmath $\calH^{\bst}$}}
\newcommand {\calHr} {\mbox{\boldmath $\calH^{\brt}$}}
\newcommand {\calHsi} {\mbox{\boldmath $\calH_i^{\bst}$}}
\newcommand {\calHri} {\mbox{\boldmath $\calH_i^{\brt}$}}

\newcommand {\calHsii} {\mbox{\boldmath $\calH_{i,j}^{\bst}$}}
\newcommand {\calHrii} {\mbox{\boldmath $\calH_{i,j}^{\brt}$}}
\newcommand {\bV} {\mbox{\boldmath $V$}}
\newcommand {\bW} {\mbox{\boldmath $W$}}
\newcommand {\bX} {\mbox{\boldmath $X$}}
\newcommand {\bY} {\mbox{\boldmath $Y$}}

\newcommand {\blam} {\mbox{\boldmath $\lambda$}}
\newcommand {\bHt} {\mbox{\boldmath \footnotesize $H$}}

\newcommand {\bAt} {\mbox{\boldmath \footnotesize $A$}}
\newcommand {\bSt} {\mbox{\boldmath \footnotesize $S$}}
\newcommand {\bRt} {\mbox{\boldmath \footnotesize $R$}}
\newcommand {\bSti} {\mbox{\boldmath \tiny $S$}}
\newcommand {\bRti} {\mbox{\boldmath \tiny $R$}}

\newcommand {\bThett} {\mbox{\boldmath \tiny $\Theta$}}

\newcommand {\blamt} {\mbox{\boldmath \footnotesize $\lambda$}}

\newcommand {\bxt} {\mbox{\footnotesize\boldmath $x$}}
\newcommand{\bxtt}{\bx_{\bst}}

\newcommand{\bHtt}{\bH_{\bst}}
\newcommand{\bHttt}{\bH_{\tilde\bSt}}
\newcommand{\bHtrt}{\bH_{\tilde\bRt}}
\newcommand{\bHtu}{\bH_{\but}}
\newcommand{\bHtr}{\bH_{\brt}}

\newcommand{\blamtt}{\blam_{\bst}}
\newcommand{\bItt}{\bI_{\bst}}
\newcommand{\bItr}{\bI_{\brt}}
\newcommand{\bQsr}{\bQ_{\bst\cap\brt}}
\newcommand{\bQs}{\bQ_{\bst}}
\newcommand{\bQr}{\bQ_{\brt}}
\newcommand {\calHst} {\mbox{\boldmath $\calH^{\tilde\bSti}$}}
\newcommand {\calHrt} {\mbox{\boldmath $\calH^{\tilde\bRti}$}}
\newcommand{\calA}{{\cal A}}
\newcommand{\calB}{{\cal B}}

\newcommand{\calD}{{\cal D}}

\newcommand{\calF}{{\cal F}}
\newcommand{\calG}{{\cal G}}
\newcommand{\calH}{{\cal H}}

\newcommand{\calN}{{\cal N}}
\newcommand{\calO}{{\cal O}}

\newcommand{\calR}{{\cal R}}
\newcommand{\calS}{{\cal S}}
\newcommand{\calT}{{\cal T}}

\newcommand{\calX}{{\cal X}}

\newcommand{\define}{\stackrel{\triangle}{=}}

\newcommand{\be}{\begin{equation}}
\newcommand{\ee}{\end{equation}}
\newcommand{\beqna}{\begin{eqnarray}}
\newcommand{\eeqna}{\end{eqnarray}}

\usepackage{theorem}
\DeclareFontFamily{U}{mathx}{\hyphenchar\font45}
\DeclareFontShape{U}{mathx}{m}{n}{
      <5> <6> <7> <8> <9> <10>
      <10.95> <12> <14.4> <17.28> <20.74> <24.88>
      mathx10
      }{}
\DeclareSymbolFont{mathx}{U}{mathx}{m}{n}
\DeclareMathSymbol{\bigtimes}{1}{mathx}{"91}
\theorembodyfont{\rmfamily}

\newcommand{\Ind}{{\mathds{1}}}

\newcommand{\abs}[1]{\left|#1\right|}

\newcommand{\diag}{\mathop{\mathrm{diag}}}

\DeclareMathOperator{\tr}{tr}
\theoremheaderfont{\itshape}
\newtheorem{definition}{Definition}
\newtheorem{theorem}{Theorem}
\newtheorem{proof}{Proof}
\newtheorem{example}{Example}

\newtheorem{lemma}{Lemma} 
\newtheorem{corollary}{Corollary}
\newtheorem{prop}{Proposition}
\newtheorem{remark}{Remark}
\newcommand{\p}[1]{\left(#1\right)}
\newcommand{\pp}[1]{\left[#1\right]}
\newcommand{\ppp}[1]{\left\{#1\right\}}
\newcommand{\norm}[1]{\left\|#1\right\|}

\begin{document}

\title{Asymptotic MMSE Analysis Under Sparse Representation Modeling$^\ast$}
\author{Wasim~Huleihel
        and~Neri~Merhav
				\\
        Department of Electrical Engineering \\
Technion - Israel Institute of Technology \\
Haifa 32000, ISRAEL\\
E-mail: \{wh@tx, merhav@ee\}.technion.ac.il
\thanks{$^\ast$This research was partially supported by The Israel Science Foundation (ISF), grant no. 412/12.}
}
\maketitle

\begin{abstract}
Compressed sensing is a signal processing technique in which data is acquired directly in a compressed form. There are two modeling approaches that can be considered: the worst-case (Hamming) approach and a statistical mechanism, in which the signals are modeled as random processes rather than as individual sequences. In this paper, the second approach is studied. In particular, we consider a model of the form $\bY = \bH\bX+\bW$, where each comportment of $\bX$ is given by $X_i = S_iU_i$, where $\ppp{U_i}$ are i.i.d. Gaussian random variables, and $\ppp{S_i}$ are binary random variables independent of $\ppp{U_i}$, and not necessarily independent and identically distributed (i.i.d.), $\bH\in\mathbb{R}^{k\times n}$ is a random matrix with i.i.d. entries, and $\bW$ is white Gaussian noise. Using a direct relationship between optimum estimation and certain partition functions, and by invoking methods from statistical mechanics and from random matrix theory (RMT), we derive an asymptotic formula for the minimum mean-square error (MMSE) of estimating the input vector $\bX$ given $\bY$ and $\bH$, as $k,n\to\infty$, keeping the measurement rate, $R = k/n$, fixed. In contrast to previous derivations, which are based on the replica method, the analysis carried out in this paper is rigorous. 
\end{abstract}


\IEEEpeerreviewmaketitle

\section{Introduction}

\IEEEPARstart{C}{ompressed} sensing \cite{Tao,Donho1} is a signal processing technique that compresses analog vectors by means of a linear transformation. Using some prior knowledge on the signal \emph{sparsity}, and by designing efficient ``encoders" and ``decoders", the goal is to achieve effective compression in the sense of taking a number of measurements much smaller than the dimension of the original signal.

A general setup of compressed sensing is shown in Fig. \ref{fig:Moisycompressed}. The mechanism is as follows: A real vector $\bX\in\mathbb{R}^n$ is mapped into $\bV\in\mathbb{R}^k$ by an encoder (or compressor) $f:\mathbb{R}^n\to\mathbb{R}^k$. The decoder (decompressor) $g:\mathbb{R}^k\to\mathbb{R}^n$ receives $\bY$, which is a noisy version of $\bV$, and outputs $\hat{\bX}$ as the estimation of $\bX$. The measurement rate, or compression ratio, $R$, satisfies $k = \left\lfloor Rn\right\rfloor$. Generally, there are two approaches to the choice of the encoder. The first approach is to constrain the encoder to be a \emph{linear} mapping, denoted by a matrix $\bH\in\mathbb{R}^{k\times n}$, usually called the \emph{sensing matrix} or \emph{measurement matrix}. Under this encoding linearity constraint, it is reasonable to consider optimal deterministic and random sensing matrices. The other approach is to consider \emph{non-linear} encoders. In this paper, we will focus on random linear encoders; $\bH$ is assumed to be a random matrix with i.i.d. entries of zero mean and variance $1/n$. At the decoder side, most of the compressed sensing literature focuses on low-complexity decoding algorithms, which are robust with respect to observation noise, for example, decoders based on convex optimization, greedy algorithms, etc. (see, for example \cite{cc6,cc7,cc8,Gastpar1}). In this paper, on the other hand, the decoder is assumed optimal, namely, it is given by the minimum mean-square error (MMSE) estimator. The input vector $\bX$ is assumed random, distributed according to some measure that is modeling the sparsity. Note that this  Bayesian formulation differs from the ``usual" compressive sensing models, in which the underlying signal is assumed deterministic and the performance is measured on a worst-case basis with respect to $\bX$ (Hamming theory). This statistical approach has been previously adopted in the literature (see, for example, \cite{cc8,Gastpar1,Tanner2,Wu2,Dono,Tulino,WuVerdu,GuoShamaiBaron,KabashimaTanWa}). Finally, the noise is assumed additive, white, and Gaussian. 

The main goal of this paper is to analyze rigorously the asymptotic behavior of the MMSE, namely, to find the MMSE for $k,n\to\infty$ with a fixed ratio $R$. Using the asymptotic MMSE, one can investigate the fundamental tradeoff between optimal reconstruction errors and measurement rates, as a function of the signal and noise statistics. For example, it will be seen that there exists a phase transition threshold of the measurement rate (which depends only on the input statistics). Above the threshold, the noise sensitivity (defined as the ratio between that MMSE and the noise variance) is bounded for all noise variances. Below the threshold, the noise sensitivity goes to infinity as the noise variance tends to zero.

\begin{figure}
\centering
\begin{pspicture}(0,-1.2592187)(11.682813,1.2992188)
\psframe[linewidth=0.04,dimen=outer](4.1209373,0.22078125)(1.5009375,-1.2592187)
\psframe[linewidth=0.04,dimen=outer](10.020938,0.22078125)(7.4209375,-1.2592187)
\psellipse[linewidth=0.04,dimen=outer](5.8509374,-0.43921876)(0.53,0.52)
\psline[linewidth=0.04cm,arrowsize=0.05291667cm 2.0,arrowlength=1.4,arrowinset=0.4]{->}(0.3009375,-0.47921875)(1.5609375,-0.49921876)
\psline[linewidth=0.04cm,arrowsize=0.05291667cm 2.0,arrowlength=1.4,arrowinset=0.4]{->}(6.4009376,-0.43921876)(7.3809376,-0.45921874)
\psline[linewidth=0.04cm,arrowsize=0.05291667cm 2.0,arrowlength=1.4,arrowinset=0.4]{->}(4.1209373,-0.45921874)(5.3409376,-0.47921875)
\psline[linewidth=0.04cm,arrowsize=0.05291667cm 2.0,arrowlength=1.4,arrowinset=0.4]{->}(10.060938,-0.43921876)(11.320937,-0.45921874)
\psline[linewidth=0.04cm,arrowsize=0.05291667cm 2.0,arrowlength=1.4,arrowinset=0.4]{->}(5.9009376,0.86078125)(5.9209375,0.08078125)
\usefont{T1}{ptm}{m}{n}
\rput(2.7982812,-0.34921876){Encoder}
\usefont{T1}{ptm}{m}{n}
\rput(2.8323438,-0.86921877){$f_n:\;\mathbb{R}^n\to\mathbb{R}^k$}
\usefont{T1}{ptm}{m}{n}
\rput(8.668906,-0.38921875){Decoder}
\usefont{T1}{ptm}{m}{n}
\rput(8.752344,-0.86921877){$g_n:\;\mathbb{R}^k\to\mathbb{R}^n$}
\usefont{T1}{ptm}{m}{n}
\rput(10.712344,-0.16921875){$\hat{\bX}$}
\usefont{T1}{ptm}{m}{n}
\rput(6.8723435,-0.18921874){$\bY$}
\usefont{T1}{ptm}{m}{n}
\rput(5.8523436,1.1107812){$\bW$}
\usefont{T1}{ptm}{m}{n}
\rput(4.742344,-0.18921874){$\bV$}
\usefont{T1}{ptm}{m}{n}
\rput(1.0023438,-0.24921875){$\bX$}
\psline[linewidth=0.04cm](5.8409376,-0.25921875)(5.8409376,-0.6592187)
\psline[linewidth=0.04cm](5.6409373,-0.45921874)(6.0409374,-0.45921874)
\end{pspicture} 
\centering
\caption{Noisy compressed sensing setup.}
\label{fig:Moisycompressed}
\end{figure}
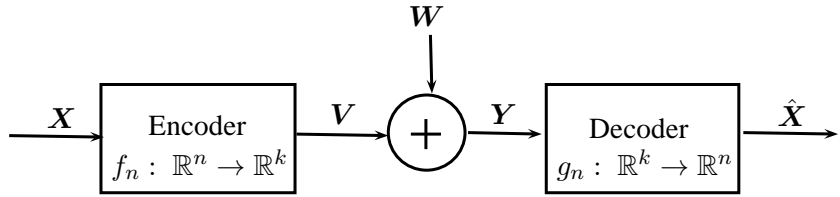

There are several previously reported results that are related to this work. Some of these results were derived rigorously and some of them were not, since they were based on the powerful, but non-rigorous, \emph{replica} method. In the following, we briefly state some of these results. In \cite{GuoShamaiBaron}, using the replica method, a decoupling principle of the posterior distribution was claimed, namely, the outcome of inferring about any fixed collection of signal elements becomes independent conditioned on the measurements. Also, it was shown that each signal-element-posterior becomes asymptotically identical to the posterior resulting from inferring the same element in scalar Gaussian noise. Accordingly, this principle allows to calculate the MMSE of estimating the signal input given the observations. In \cite{WuVerdu}, among other results, it was shown rigorously that for i.i.d. input processes, distributed according to any discrete-continuous mixture measure (where the discrete part has finite Shannon entropy), the phase transition threshold for optimal encoding is given by the input information dimension. This result serves as a rigorous verification of the replica calculations in \cite{GuoShamaiBaron}. In \cite{Dono,Krzakala1,Krzakala2}, the authors designed structured sensing matrices (not necessarily i.i.d.), and a corresponding reconstruction procedure, that allows compressed sensing to be performed at acquisition rates approaching to the theoretically optimal limits. A wide variety of previous works are concerned with low-complexity decoders, which are robust with respect to the noise, e.g., decoders based on convex optimizations (such as $\ell_1$-minimization and $\ell_1$-penalized least-squares) \cite{cc6,cc7}, graph-based iterative decoders such as linear MMSE estimation and approximate message passing (AMP) \cite{cc8}, etc. For example, in \cite{Gastpar1,Gastpar2,Bayati}, the linear MMSE and LASSO estimators were studied for the case of i.i.d. sensing matrices as special cases of the AMP algorithm, the performance of which was rigorously characterized for Gaussian sensing matrices \cite{MohsenBayati}, and generalized for a broad class of sensing matrices in \cite{Dono,Bayati2,Bayati3}. Another, somewhat related, subject, is the recovery of the sparsity pattern with vanishing and non-vanishing error probability, was studied in a number of recent works, e.g., \cite{Gastpar1,Tulino,Gastpar2,recovery1,recovery2,recovery3,recovery4,recovery5,recovery6}. For example, in \cite{Tulino}, using the replica method and the decoupling principle, the authors extend the scope of conventional noisy compressive sampling where the sensing matrix is assumed to have i.i.d. entries to allow it to satisfy a certain freeness condition (encompassing Haar matrices and other unitary invariant matrices).

In this paper, under the previously mentioned model assumptions, we rigorously derive the asymptotic MMSE in a single-letter form. The key idea in our analysis is the fact that by using some direct relationship between optimum estimation and certain partition functions \cite{Neri1}, the MMSE can be represented in some mathematically convenient form which (due to the previously mentioned input and noise Gaussian statistics assumptions) consists of functions of the \emph{Stieltjes} and \emph{Shannon} transforms. This observation allows us to use some powerful results from random matrix theory (RMT), concerning the asymptotic behavior (a.k.a. deterministic equivalents) of the Stieltjes and Shannon transforms (see e.g., \cite{baisilbook,coulbook} and many references therein). Our asymptotic MMSE formula seems to appear different than the one that is obtained from the replica method \cite{GuoShamaiBaron}. Nevertheless, numerical calculations indicate matching results with high accuracy and therefore suggest that the results are equivalent. Thus, similarly to other known cases in statistical mechanics, for which the replica predictions were proved to be correct, our results support the replica method predictions. Notwithstanding the apparent equivalence, we believe that our formula is more insightful compared to the replica method results. Also, in contrast to previous works, in which only memoryless sources were considered (an indispensable assumption in the analysis), we allow a certain structured dependency among the various components of the source. Finally, we mention that in a previous related paper \cite{Wasim}, the authors have used similar methodologies to obtain the asymptotic mismatched MSE of a codeword (from a randomly selected code), corrupted by a Gaussian vector channel. 

The remaining part of this paper is organized as follows. In Section \ref{sec:model}, the model is presented and the problem is formulated. In Section \ref{sec:body}, the main results are stated and discussed along with a numerical example that demonstrates the theoretical result. In Section \ref{sec:proof}, the main result is proved, and finally, in Section \ref{sec:Conclusion} our conclusions are drawn and summarized.

\section{Notation Conventions and Problem Formulation}\label{sec:model}
\allowdisplaybreaks

\subsection{Notation Conventions}\label{sec:notation}
Throughout this paper, scalar random variables (RV's) will be denoted by capital letters, their sample values will be denoted by the respective lower case letters and their alphabets will be denoted by the respective calligraphic letters. A similar convention will apply to random vectors and matrices and their sample values, which will be denoted with same symbols in the bold face font. We let $P_{\bSt}$ and $p_{\bXt,\bYt}$ be the probability mass function and the joint density function of the discrete random vector $\bS$ and the continuous random vectors $\bX$ and $\bY$, respectively. Accordingly, $p_{\bXt}$ will denote the marginal of $\bX$, $p_{\bYt\vert\bXt}$ will denote the conditional density of $\bY$ given $\bX$, and so on. Probability measures will be denoted generically by the letter $\pr$.

The expectation operator of a measurable function $f\p{\bX,\bY}$ with respect to (w.r.t.) $p_{\bXt,\bYt}$ will be denoted by $\bE\ppp{f\p{\bX,\bY}}$. The conditional expectation of the same function given a realization $\by$ of $\bY$, will be denoted by $\bE\ppp{f\p{\bX,\bY}\vert \bY = \by}$. When using vectors and matrices in a linear-algebraic format, $n$-dimensional vectors, like $\bx$, will be understood as column vectors, the operators $\p{\cdot}^T$ and $\p{\cdot}^H$ will denote vector or matrix transposition and vector or matrix conjugate transposition, respectively, and so, $\bX^T$ would be a row vector. For two positive sequences $\ppp{a_n}$ and $\ppp{b_n}$, the notation $a_n\exe b_n$ means equivalence in the exponential order, i.e., $\lim_{n\to\infty}\frac{1}{n}\log\p{a_n/b_n} = 0$, where in this paper, logarithms are defined w.r.t. the natural basis, that is, $\log(\cdot) = \ln(\cdot)$. For two sequences of random variables $\ppp{a_n}$ and $\ppp{b_n}$, we denote by $a_n\asymp b_n$ and $a_n\sim b_n$ the equivalence relations $a_n-b_n\stackrel{\text{a.s.}}{\rightarrow}0$ and $a_n/b_n\stackrel{\text{a.s.}}{\rightarrow}1$ almost surely (a.s.) for $n\to\infty$, respectively. Finally, the indicator function of an event $\calA$ will be denoted by $\Ind_{\calA}$.

\subsection{Model and Problem Formulation}\label{subsec:model}
As was mentioned earlier, we consider sparse signals, supported on a subspace with dimension smaller than $n$. In the literature, it is often assumed that the input process $\bX$ has i.i.d. components. In this work, however, we generalize this assumption by considering the following stochastic model: Each component, $X_i$, $1\leq i\leq n$, of $\bX$, is given by $X_i = S_iU_i$ where $\ppp{U_i}$ are i.i.d. Gaussian random variables with zero mean and variance $\sigma^2$, and $\ppp{S_i}$ are binary random variables taking values in $\ppp{0,1}$, independently of $\ppp{U_i}$. Now, instead of assuming that the \emph{pattern} sequence $\bS = \p{S_1,\ldots,S_n}$ is i.i.d., we will assume a more general distribution. In particular, defining the
%
``\emph{magnetization}"\footnote{The term ``magnetization" is borrowed from the field of statistical mechanics of spin array systems, in which $S_i$ is taking values in $\ppp{-1,1}$. Nevertheless, for the sake of convince, we will use this term also in our problem.}
\begin{align}
m_{\bst} \triangleq \frac{1}{n}\sum_{i=1}^ns_i,
\end{align}
we assume a distribution of the form
\begin{align}
P_{\bSt}\p{\bs} = C_n\cdot\exp\ppp{nf\p{m_{\bst}}},
\label{inputassmeas}
\end{align}
where $f\p{\cdot}$ is a certain function, independent of $n$, and $C_n$ is a normalization constant. Note that for the popular i.i.d. assumption, $f$ is a linear function. Let us assume that $f$ twice differentiable with a finite first derivative on $\pp{0,1}$. Then, by using the method of types \cite{Cizer}, we obtain
\begin{align}
C_n &= \p{\sum_{\bst\in\ppp{0,1}^n}\exp\ppp{nf\p{m_{\bst}}}}^{-1}\nonumber\\
&= \p{\sum_{m\in\ppp{0,1/n,\ldots,1}}\Omega\p{m}\exp\ppp{nf\p{m}}}^{-1}\nonumber\\
&\exe \exp\ppp{-n\sup_m\ppp{h_2\p{m}+f\p{m}}}\nonumber\\
& = \exp\ppp{-n\pp{h_2\p{m_a}+f\p{m_a}}}
\label{apriorimag}
\end{align} 
where $\Omega\p{m}$ designates the number of binary $n$-vectors with magnetization $m$, $h_2\p{\cdot}$ designates the binary entropy function, and $m_a$ is the maximizer of $h_2\p{m}+f\p{m}$ over $\pp{0,1}$. In other words, $m_a$ is the \emph{a-priori} magnetization, namely, the magnetization that \emph{dominates} $P_{\bSt}(\cdot)$. Note that the maximum of $h_2\p{m}+f\p{m}$ is achieved by an internal point in $\pp{0,1}$. This is because $h_2(\cdot)$ is concave with infinite derivatives at the boundaries, whereas the derivative of $f$ is finite. In case of multiple maximizers, the global supremum (assumed to be unique) is identified by comparing the corresponding values of $h_2\p{m}+f\p{m}$. 

To conclude, we summarize the structure of our sparsity model. The input $\bX$ is generated as follows: first, the support size of $\bX$ is drawn according to the distribution in \eqref{inputassmeas}. Then, the support set is drawn uniformly at random from all subsets of that cardinality. Finally, the non-zero elements are filled with i.i.d. standard Gaussian random variables.   
\begin{remark}
The structure of $P_{\bSt}(\cdot)$ in \eqref{inputassmeas} can be relaxed by allowing $f = f_n$, $f_n$ converge to some limit $f$ uniformly on $\pp{0,1}$. This relaxation allows our model to include, for example, the basic case of exact sparsity in which $P_{\bSt}(\bs) = 1/\binom {n} {nm_{s}}$. Due to fact that our analysis is not affected by this relaxation (attributed to the assumption that $\ppp{f_n}$ and $f$ depend only on $m_{\bst}$), and accordingly, the main result of this paper remains the same, we will assume that $f$ is fixed, as described in \eqref{inputassmeas}. 
\end{remark}
\begin{remark}
In the i.i.d. case, each $X_i$ is distributed according to following mixture distribution (a.k.a. Bernoulli-Gaussian measure)
\begin{align}
p_X(x) = \p{1-p}\cdot\delta\p{x} + p\cdot p_G\p{x}
\label{mes}
\end{align}
where $\delta\p{x}$ is the Dirac function, $p_G\p{x}$ is a Gaussian density function corresponding to a Gaussian random variable with zero mean and variance $\sigma^2$, and $0\leq p\leq1$. Consider a random vector $\bX$ in which each component is \emph{independently} drawn from $p_{X}$. Then, by the law of large numbers (LLN), $\frac{1}{n}\norm{\bX}_0\stackrel{\mathbb{P}}{\rightarrow}p$, where $\norm{\bX}_0$ designates the number of non-zero elements of the vector $\bX$. In other words, in this case, $m_a=p$. Thus, it is clear that the weight $p$ parametrizes the signal sparsity and $p_G$ is the prior distribution of the non-zero entries. Note that the fact that, $\frac{1}{n}\norm{\bX}_0\stackrel{\mathbb{P}}{\rightarrow}m_a$, is true regardless the i.i.d. assumption. Indeed, this follows from Chebyshev's inequality, and the fact that (using the saddle-point method \cite[Section 4.2]{NeriMono})
\begin{align}
\lim_{n\to\infty}\bE\pp{\abs{\frac{1}{n}\sum_{i=1}^nS_i-m_a}} &= \lim_{n\to\infty}\sum_{\bst\in\ppp{0,1}^n}\abs{m_{\bst}-m_a}P_{\bSt}(\bs) \\&
= \lim_{n\to\infty}\frac{\sum_{\bst\in\ppp{0,1}^n}\abs{m_{\bst}-m_a}\exp\pp{nf(m_{\bst})}}{\sum_{\bst\in\ppp{0,1}^n}\exp\pp{nf(m_{\bst})}} = 0.\label{expecmagaprior}
\end{align}
\end{remark}

Finally, we consider the following model
\begin{align}
\bY = \bH\bX+\bW,
\end{align}
where $\bH$ is a $k\times n$ random matrix, a.k.a. the \emph{sensing matrix}, with i.i.d. entries of zero mean and variance $1/n$. We assume that the entries of $\bH$, denoted by $\ppp{H_{i,j}}_{i,j}$, have bounded normalized moments, i.e., $\bE(\sqrt{n}H_{i,j})^l\leq\upsilon_l<\infty$, for $l\in\ppp{1,2,\ldots,8}$. The components of the noise $\bW$ are i.i.d. Gaussian random variables with zero mean and variance $1/\beta$. The MMSE of $\bX$ given $\bY$ and $\bH$ is defined as follows
\begin{align}
\text{mmse}\p{\bX\vert\bY,\bH} &\triangleq \bE\norm{\bX-\bE\ppp{\bX\vert\bY,\bH}}^2
\end{align}
where $\bE\ppp{\bX\vert\bY,\bH}$ is the conditional expectation w.r.t. $p_{\bXt\vert\bYt,\bHt}$. As was mentioned earlier, we are interested in the asymptotic regime, where $k,n\to\infty$ with a fixed ratio $R$, which we shall refer to as the \emph{measurement rate}. Accordingly, we define the \emph{asymptotic MMSE} as
\begin{align}
D\p{R,\beta}\triangleq \limsup_{n\to\infty}\frac{1}{n}\text{mmse}\p{\bX\vert\bY,\bH}.
\label{asympMMSEdef}
\end{align}
Our main goal is to rigorously derive a computable, single-letter expression for $D\p{R,\beta}$. 

\section{Main Result}\label{sec:body}
In this section, our main result is first presented and discussed. Then, we provide a numerical example in order to illustrate the theoretical results. The proof of the main theorem is provided in Section \ref{sec:proof}.

Before we state our main result, we define some auxiliary functions of a generic variable $x\in\pp{0,1}$:
\begin{align}
&b\p{x} \triangleq \frac{-\p{1+\beta\sigma^2\p{R-x}}+\sqrt{\pp{1+\beta\sigma^2\p{R-x}}^2+4\beta\sigma^2x}}{2\beta\sigma^2x},\label{firstt}\\
&g\p{x} \triangleq 1+\beta\sigma^2xb\p{x},\\
&\bar{I}\p{x} \triangleq \frac{R}{x}\log{g\p{x}}-\log{b\p{x}}-\frac{\beta\sigma^2Rb\p{x}}{g\p{x}},\label{firs2tt}\\
&V\p{x}\triangleq\frac{\beta^3\sigma^4b^2\p{x}x^2}{2g^2\p{x}},\\
&L\p{x}\triangleq\frac{\beta^2\sigma^2b\p{x}}{2g^2\p{x}},\\
&\nu_1\p{x} \triangleq \frac{\beta R}{g\p{x}}+\frac{1}{\sigma^2},
\end{align}
and
\begin{align}
&t\p{x}\triangleq f\p{x}-\frac{x}{2}\bar{I}\p{x}+V\p{x}\p{m_aR\sigma^2+\frac{R}{\beta}}.
\end{align}
Next, for $x,y\in\pp{0,1}$ define the functions
\begin{align}
&\nu_2\p{x,y} \triangleq \frac{\beta R}{g\p{x}}-\frac{\beta^2 R\sigma^2b\p{x}y}{g^2\p{x}}+\frac{1}{\sigma^2},
\end{align}
and
\begin{align}
&\alpha\p{x,y}\triangleq \frac{1}{\nu_1\p{x}\nu_2\p{x,y}}.\label{lastt}
\end{align}
The asymptotic MMSE is given in the following theorem.
\begin{theorem}[Asymptotic MMSE]\label{th:1}
Let $Q$ be a random variable distributed according to
\begin{align}
p_Q\p{q} = \frac{1-m_a}{\sqrt{2\pi P_y}}\exp\p{-\frac{q^2}{2P_y}} +\frac{m_a}{\sqrt{2\pi \p{P_y+R^2\sigma^2}}}\exp\p{-\frac{q^2}{2\p{P_y+R^2\sigma^2}}} 
\end{align}
where $m_a$ is defined as in \eqref{apriorimag} and $P_y \triangleq m_a\sigma^2R+R/\beta$. Let us define 
\begin{align}
K\p{Q,\alpha_1,\alpha_2} \triangleq \frac{1}{2}\pp{1+\tanh\p{\frac{L\p{\alpha_1}Q^2-\alpha_2}{2}}}
\label{KtermFluc}
\end{align}
where $\alpha_1\in\pp{0,1}$ and $\alpha_2\in\mathbb{R}$. Let $m^{\circ}$ and $\gamma^\circ$ be solutions of the system of equations
\begin{subequations}
\begin{align}
&\gamma^\circ \triangleq-\bE\ppp{K\p{Q,m^\circ,\gamma^\circ}Q^2L'(m^\circ)}-t'(m^\circ),\label{magnetDet1}\\
&m^\circ \triangleq \bE\ppp{K\p{Q,m^\circ,\gamma^\circ}}
\label{magnetDet}
\end{align}\label{magnetddd}%
\end{subequations}
where $L'(\cdot)$ and $t'(\cdot)$ are the derivatives of $L(\cdot)$ and $t(\cdot)$, respectively, and in case of more than one solution, $\p{m^\circ,\gamma^\circ}$ is the pair with the largest value of
\begin{align}
t\p{m^\circ}+\p{m^\circ-\frac{1}{2}}\gamma^\circ+\bE\ppp{\frac{1}{2}L\p{m^\circ}Q^2+\log \pp{2\cosh\p{\frac{L\p{m^\circ}Q^2-\gamma^\circ}{2}}}}.
\label{criticalPoint}
\end{align}
Finally, define 
\begin{align}
&\rho_1^\circ \triangleq\bE\ppp{K\p{Q,m^\circ,\gamma^\circ}Q^2},\label{rhoa1}\\
&\rho_2^\circ \triangleq \bE\ppp{K^2\p{Q,m^\circ,\gamma^\circ}},\label{rhoa2}\\
&\rho_3^\circ \triangleq \bE\ppp{K^2\p{Q,m^\circ,\gamma^\circ}Q^2}\label{rhoa3}.
\end{align}
Then, the limit supremum in \eqref{asympMMSEdef} is, in fact, an ordinary limit, and the asymptotic MMSE is given by
\begin{align}
D\p{R,\beta}&= \sigma^2m_a-\beta^2\frac{\alpha\p{m^\circ,\rho_2^\circ}}{g^2\p{m^\circ}}\rho_3^\circ+2\frac{\alpha\p{m^\circ,\rho_2^\circ}b\p{m^\circ}}{g^3\p{m^\circ}}\beta^3\sigma^2\rho_2^\circ\pp{\rho_1^\circ-m^\circ P_y}.
\label{DasymMMSE}
\end{align}
\end{theorem}

In the following, we explain the above result qualitatively, and in particular, the various quantities that have been defined in Theorem \ref{th:1}. The first important quantity is $m^\circ$, which is obtained as the solution of the system of equations in \eqref{magnetddd}, and which we will refer to as the \emph{posterior} magnetization. We use the term ``posterior" in order to distinguish it from the a-priori magnetization $m_a$; while $m_a$ is the magnetization that dominates the probability distribution function of the source, before observing $\bY$, the posterior magnetization is the one that dominates the posterior distribution, namely, after observing the measurements. The solution of the equation
\begin{align}
t\p{m^\circ}+\p{m^\circ-\frac{1}{2}}\gamma^\circ+\bE\ppp{\frac{1}{2}L\p{m^\circ}Q^2+\log \pp{2\cosh\p{\frac{L\p{m^\circ}Q^2-\gamma^\circ}{2}}}}=0,
\end{align}
is known as a \emph{critical point}, beyond which the solution to \eqref{magnetddd} ceases to be the dominant posterior magnetization, and accordingly, it must jump elsewhere. Furthermore, as we vary one of the other parameters of our model (including the source model), it might happen that the dominant magnetization jumps from one value to another. 

It is interesting to note that there are essentially two origins for possible phase transitions in our model: The first one is the channel $\bH$ that induces ``long-range interactions"\footnote{In the settings considered, the posterior is proportional to $\exp\ppp{-\beta\norm{\by-\bH\bX}^2/2}$, and after expansion of the norm, the exponent includes an ``external-field term", proportional to $\by^T\bH\bx$, and a ``pairwise spin-spin interaction term", proportional to $\norm{\bH\bX}^2$. These terms contain a linear subset of components (or ``particles") of $\bX$, which are known as long-range interactions.}. The second is the source, which may have possible dependency (or interaction) between its various components (see \eqref{inputassmeas}). Accordingly, in \cite[Example E]{Neri2} the problem of estimation of sparse signals, assuming that $\bH=\bI$, was considered. It was shown that, despite the fact that there are no long-range interactions induced by the channel, still there are phase transitions if the source is not i.i.d. Indeed, in the i.i.d. case, the problem is analogous to a system of non-interacting particles, where of course, no phase transitions can exist. Specifically, assume that $\bH=\bI$, and consider the special case where $f\p{m}$ is quadratic\footnote{As was noted in \cite{Neri2}, quadratic model (similar to the \emph{random-field Curie-Weiss model} of spin systems (see e.g., \cite[Sect. 4.2]{curriwiess})) can be thought of as consisting of the first two terms of the Taylor series expansion of a smooth function.}, i.e., $f\p{m} = am+bm^2/2$. We demonstrate that the dominant posterior magnetization might jump from one value to another. Note that this example was also considered in \cite[Example E]{Neri2}. For simplicity of the demonstration, we assume that $\sigma^2$ and $\beta$ are small, and then it can be shown that $m^\circ$ behaves as \cite[Example E]{Neri2}:
\begin{align}
m^\circ &\approx \frac{1}{2}\pp{1+\tanh\p{\frac{1}{2}t'(m^\circ)}}\\
& \approx \frac{1}{2}\pp{1+\tanh\p{\frac{bm^\circ+a}{2}}},
\end{align}
which can be regarded as the same equation of the \emph{spin}-magnetization (namely, after transforming $S_i$'s into spins, $\mu_i\in\ppp{-1,1}$, using the transformation $\mu_i = 1-2S_i$) as in the Curie-Weiss model of spin arrays (see e.g., \cite[Sect. 4.2]{curriwiess}). For example, for $a=0$ and $b>1$, this equation has two symmetric non-zero solutions $\pm m_0$, which both dominate the partition function. If $0<a\ll1$, it is evident that the symmetry is broken, and there is only one dominant solution which is about $\abs{m_0}$. Further discussion on the behavior of the above saddle point equation and various interesting approximations of the dominant magnetization can be found in \cite{Neri2,curriwiess,Mezard}. 

It is now tempting to compare Theorem \ref{th:1} with the prediction of the replica method \cite{GuoShamaiBaron}. Unfortunately, we were unable to show analytically that the two results are in agreement, despite the fact that there are some similarities. Nevertheless, numerical calculations suggest that this is the case. Fig. \ref{fig:1} shows the asymptotic MMSE obtained using Theorem \ref{th:1} and using the replica method, as a function of $\beta$, assuming an i.i.d. source with sparsity rate $p=0.1$, and measurement rate $R = 0.3$. Table \ref{fig:2} shows the relative error, defined as $\abs{\text{mmse}_{\text{our}}-\text{mmse}_{\text{replica}}}/\text{mmse}_{\text{our}}$, as a function of $\beta$. It can be seen that both results give approximately the same MMSE. The very small differences between the two results are just numerical, finite precision errors. More enlightening numerical examples can be found in \cite{Tulino,Krzakala1,Krzakala2,TsuBaron}.
\begin{figure}[!t]
\begin{minipage}[b]{1.0\linewidth}
  \centering
	\centerline{\includegraphics[width=15cm,height = 12cm]{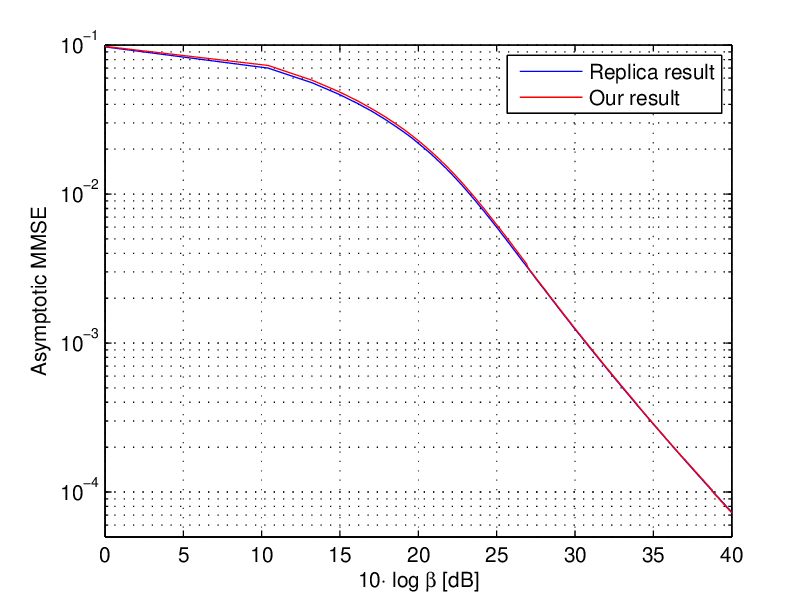}}
	\end{minipage}
\caption{Comparison of the asymptotic MMSE using Theorem \ref{th:1} and the replica method as a function of $\beta$, for sparsity rate $p=0.1$, and measurement rate $R = 0.3$.}
\label{fig:1}
\end{figure}

\begin{table}[ht] 
\caption{Comparison Between Theorem \ref{th:1} and the Replica Method}
\centering 
\begin{tabular}{c c} 
\hline\hline 
$10\log\beta$ & Relative Error \\ [0.5ex] 
\hline 
0 & $5.11\cdot10^{-3}$\\ 
10 & $8.09\cdot10^{-3}$\\ 
15 & $6.12\cdot10^{-3}$\\ 
20 & $6.51\cdot10^{-3}$\\ 
25 & $6.03\cdot10^{-3}$\\ 
30 & $4.65\cdot10^{-3}$\\  
35 & $4.49\cdot10^{-3}$\\ 
40 & $4.59\cdot10^{-3}$\\
\hline 
\end{tabular}
\label{fig:2} 
\end{table}

\section{Proof of Theorem 1}\label{sec:proof}

\subsection{Proof Outline}
In this subsection, before delving into the proof of Theorem \ref{th:1}, we discuss the techniques and the main steps used in the proof. The analysis is essentially composed of three main steps. The first step is finding a generic expression of the MMSE. This is done by using a direct relationship between the MMSE and some partition function, which can be found in Lemma \ref{lemma:partMM}. This expression contains terms that can be asymptotically assessed using the well-known Stieltjes and Shannon transforms. In the second step (appearing in Appendix \ref{app:1}), we derive the asymptotic behavior of these functions (which are extremely complex to analyze for finite $n$). In other words, we show that these functions can be replaced with some other random functions that are much easier to work with, and the loss/gap due to this replacemnt is bounded by a vanishing term. This is done by invoking recent powerful methods from RMT, such as the Bai-Silverstein method \cite{SilversteinBai}. The resulting functions are, in general, random, due to the fact that they depend on the observations $\by$ and the sensing matrix $\bH$. Accordingly, we show that for the calculation of the asymptotic MMSE, it is sufficient to take into account only combinations of typical vectors $\ppp{\by}$ and matrices $\ppp{\bH}$, where typicality is defined in accordance to the above-mentioned asymptotic results. Therefore, at the end of the second step, we obtain an approximation (which is exact as $n\to\infty$) for the MMSE. Finally, in the last step, using this approximation and large deviations theory, we obtain the result stated in Theorem \ref{th:1} (this step can be found in Appendix \ref{app:3}).

\subsection{Definitions}
An important function, which will be pivotal to our derivation, is the \emph{partition function}, which is defined as follows. 
\begin{definition}[Partition Function]\label{def:1}
Let $\bX$ and $\bY$ be random vectors with joint density function $p_{\bXt,\bYt}$. Let $\blam = \p{\lambda_1,\ldots,\lambda_n}^T$ be a deterministic column vector of $n$ real-valued parameters. The partition function w.r.t. $p_{\bXt,\bYt}$, denoted by $Z\p{\by;\blam}$, is defined as
\begin{align}
Z\p{\by;\blam}&\triangleq\int_{\calX^n}\mathrm{d}\bx\;p_{\bXt,\bYt}(\bx,\by)\exp\ppp{\blam^T\bx}.
\label{PartFunc}
\end{align} 
\end{definition}

The motivation of the above definition is the following simple result \cite{Neri1}.
\begin{lemma}[MMSE-partition function relation]\label{lemma:partMM}
Consider the model presented in Subsection \ref{subsec:model}. Then, the following relation between $Z\p{\bY;\blam}$ and the MMSE of $\bX$ given $\bY$, holds true
\begin{align}
\text{mmse}\p{\bX\vert\bY} &\triangleq \sum_{i=1}^n \bE\ppp{\p{X_i-\bE\ppp{X_i\vert\bY}}^2}\nonumber\\
&= \sum_{i=1}^n\pp{\bE\ppp{X_i^2}-\bE\ppp{\left.\pp{\frac{\partial\log Z\p{\bY;\blam}}{\partial\lambda_i}}^2\right|_{\blamt=0}}}.
\end{align}
\end{lemma}
\begin{proof}
The main observation here is that
\begin{align}
\bE\ppp{X_i\vert\bY=\by} = \left.\frac{\partial\log Z\p{\by;\blam}}{\partial\lambda_i}\right|_{\blamt=0}.
\end{align}
Indeed, we note that
\begin{align}
\left.\frac{\partial\log Z\p{\by;\blam}}{\partial\lambda_i}\right|_{\blamt=0} &= \frac{1}{\int_{\calX^n}\mathrm{d}\bx\;p_{\bXt,\bYt}(\bx,\by)}\left.\frac{\partial}{\partial\lambda_i}\int_{\calX^n}\mathrm{d}\bx\;p_{\bXt,\bYt}(\bx,\by)\exp\ppp{\blam^T\bx}\right|_{\blamt=0}\\
& = \frac{\int_{\calX^n}\mathrm{d}\bx\;x_ip_{\bXt,\bYt}(\bx,\by)}{\int_{\calX^n}\mathrm{d}\bx\;p_{\bXt,\bYt}(\bx,\by)}= \bE\ppp{X_i\vert\bY= \by}
\end{align}
where the second equality follows from the following lemma \cite[Lemma 2]{Palomar1}.
\begin{lemma}\label{lem:divIntegral}
Consider a function $f(\bx,\theta)$ and a nonnegative function $g(\bx)$. The relation
\begin{align}
\frac{\partial}{\partial\theta}\int g(\bx)f(\bx;\theta)\mathrm{d}\bx = \int g(\bx)\frac{\partial}{\partial\theta}f(\bx;\theta)\mathrm{d}\bx
\end{align}
holds if for each $\theta_0$ there exists a neighborhood of $\theta_0$, $\calN_{\theta_0}$, and a function $M(\bx;\theta_0)$, such that
\begin{align}
\sup_{\theta\in\calN_{\theta_0}}\abs{\frac{\partial}{\partial\theta}f(\bx;\theta)}\leq M(\bx;\theta_0),\ \text{a.e.}\label{a.s.Connd}
\end{align}
with $\int g(\bx)M(\bx;\theta_0)\mathrm{d}\bx<\infty$. 
\end{lemma}

In our case, we have (substituting $\blam = (0,0,\ldots,\lambda_i,0,\ldots,0)$)
\begin{align}
\abs{\frac{\partial}{\partial\lambda_i}\exp\ppp{\lambda_ix_i}} =\abs{x_i}\exp\p{\lambda_ix_i},\label{partDivLemm}
\end{align}
To apply Lemma \ref{lem:divIntegral}, we need to check that for each $\lambda_{i,0}$ there exists a neighborhood around $\lambda_{i,0}$ such that \eqref{a.s.Connd} holds. Accordingly, to make the right-hand side (r.h.s.) of \eqref{partDivLemm} valid for a neighborhood around each such $\lambda_{i,0}$, we take $M(x_i;\lambda_{i,0}) = \abs{x_i}\exp\pp{(\lambda_{i,0}+\epsilon)\abs{x_i}}$ for some $\epsilon>0$. Finally, under the model presented in Subsection \ref{subsec:model}, the expectation of $M(x_i;\lambda_i)$ is clearly finite since the underlying joint probability distribution of $\bX$ and $\bY$ decays faster than the increase of $M(x_i;\lambda_i)$, and thus Lemma \ref{lem:divIntegral} can be invoked.
\end{proof}

Our analysis will rely heavily on methods and results from RMT. Two efficient tools commonly used in RMT are the \emph{Stieltjes} and \emph{Shannon} transforms, which are defined as follows.
\begin{definition}[Stieltjes Transform]
Let $\mu$ be a finite nonnegative measure with support $\text{supp}\p{\mu}\subset\mathbb{R}$, i.e., $\mu\p{\mathbb{R}}<\infty$. The Stieltjes transform $S_\mu\p{z}$ of $\mu$ is defined for $z\in\mathbb{C}-\text{supp}\p{\mu}$ as
$$
S_\mu\p{z} = \int_{\mathbb{R}}\frac{\mathrm{d}\mu\p{\lambda}}{\lambda-z}.
$$

Let $F_{\bAt}\p{\cdot}$ be the empirical spectral distribution (ESD) of the eigenvalues of a non-negative definite matrix $\bA\in\mathbb{R}^{N\times N}$, namely,
\begin{align}
F_{\bAt}\p{x} \triangleq \frac{1}{N}\ppp{\#\text{ of eigenvalues of $\bA$}\leq x}.
\end{align}
The Stieltjes transform of $F_{\bAt}\p{x}$ is defined as
\begin{align}
S_{\bAt}\p{z} \triangleq \int_{\mathbb{R}^+}\frac{dF_{\bAt}\p{x}}{x-z} = \frac{1}{N}\tr\p{\bA-z\bI}^{-1}
\end{align} 
for $z\in\mathbb{C}\setminus\mathbb{R}^+$. 
\end{definition}
The last equality readily follows by using the spectral decomposition of $\bA$, and the fact that the trace of a matrix equals to the sum of its eigenvalues. For brevity, we will refer to $S_{\bAt}\p{z}$ as the Stieltjes transform of $\bA$, rather than the Stieltjes transform of $F_{\bAt}\p{x}$.
\begin{definition}[Shannon Transform]
The Shannon of transform of a non-negative definite matrix $\bA\in\mathbb{C}^{N\times N}$ is defined as
\begin{align}
\nu_{\bAt}\p{z} \triangleq \frac{1}{N}\log\det\p{\frac{1}{z}\bA+\bI},
\end{align}
for $z>0$.
\end{definition}

The relation between our partition function and the Stieltjes and Shannon transforms will become clear in the sequel. Finally, we define the notion of deterministic equivalence.
\begin{definition}[Deterministic Equivalence]
Let $\p{\Omega,\calF,P}$ be a probability space and let $\ppp{f_n}$ be a series of measurable complex-valued functions, $f_n:\Omega\times\mathbb{C}\to\mathbb{C}$. Let $\ppp{g_n}$ be a series of complex-valued functions, $g_n:\mathbb{C}\to\mathbb{C}$. Then, $\ppp{g_n}$ is said to be a deterministic equivalent of $\ppp{f_n}$ on $D\subset\mathbb{C}$, if there exists a set $A\subset\Omega$ with $P\p{A}=1$, such that
\begin{align}
f_n\p{\omega,z}-g_n\p{z}\to0
\end{align}
as $n\to\infty$ for all $\omega\in A$ and for all $z\in D$. 
\end{definition}

Loosely speaking, $\ppp{g_n}$ is a deterministic equivalent of a sequence of random variables $\ppp{f_n}$ if $g_n\p{z}$ approximates $f_n\p{\omega,z}$ arbitrarily closely as $n$ grows, for every $z\in D$ and every $\omega\in A$.

\subsection{Auxiliary Results}
In our derivations, the following asymptotic results will be used.


\begin{lemma}\label{aux:assympLLN}
Consider a sequence of random variables $\ppp{X_{i,n}}_{i=1}^n$. Assume that
\begin{align}
\max_{1\leq i\leq n}\ppp{\bE\abs{X_{i,n}}^p}\leq\frac{C}{n^{1+\nu}}
\label{LLNASSU}
\end{align}
where $C$, $\nu>0$, and $p\geq1$ are some constants. Then, for any $\delta>0$,
\begin{align}
\pr\ppp{\frac{1}{n}\sum_{i=1}^n\abs{X_{i,n}}>\delta}\leq\frac{C}{\delta^pn^{1+\nu}},
\end{align}
and $n^{-1}\sum_{i=1}^n|X_{i,n}|$ converges to zero in the a.s. sense.
\end{lemma}
\begin{proof}
Using Chebyshev's inequality and then Jensen's inequality, for a given $\delta>0$, we have
\begin{align}
\pr\ppp{\frac{1}{n}\sum_{i=1}^n\abs{X_{i,n}}>\delta}&\leq\frac{1}{\delta^p}\bE\ppp{\p{\frac{1}{n}\sum_{i=1}^n\abs{X_{i,n}}}^p}\\
&\leq\frac{1}{n\delta^p}\sum_{i=1}^n\bE\abs{X_{i,n}}^p\\
&\leq \frac{1}{\delta^p}\max_{1\leq i\leq n}\ppp{\bE\abs{X_{i,n}}^p}\\
&\leq\frac{C}{\delta^pn^{1+\nu}}
\label{LLNineq}
\end{align}
where the last inequality follows by \eqref{LLNASSU}. With \eqref{LLNineq}, the a.s. convergence follows from the Borel-Cantelli lemma. Indeed, as the r.h.s. of \eqref{LLNineq} is summable, by the Borel-Cantelli lemma, we have
\begin{align}
\pr\p{\ppp{\omega\in\Omega:\;\frac{1}{n}\sum_{i=1}^n\abs{X_{i,n}\p{\omega}}\geq\delta\;\text{infinitely often}}}=0.
\end{align}
But since $\delta>0$ is arbitrary, the above holds for all rational $\delta>0$. Since any countable union of sets of zero probability is still a set of zero probability, we conclude that
\begin{align}
\pr\p{\bigcup_{q\in\mathbb{N}}\ppp{\omega\in\Omega:\;\frac{1}{n}\sum_{i=1}^n\abs{X_{i,n}\p{\omega}}\geq\frac{1}{q}\;\text{infinitely often}}}=0.
\end{align}
\end{proof}

The following lemmas deal with the asymptotic behavior of scalar functions of random matrices, in the form of Stieltjes and Shannon transforms, defined earlier. The proofs of the following results are based on a powerful approach by Bai and Silverstein \cite{SilversteinBai}, a.k.a. the Stieltjes transform method in the spectral analysis of large-dimensional random matrices.
\begin{lemma}[\cite{Couillet}]\label{lem:mat1}
Let $\bX_m\in\mathbb{C}^{m\times l}$ be a sequence of random matrices with i.i.d. entries, $\bE\abs{X_{i,j}-\bE X_{i,j}}^2=1/l$, and let $\bG_l = \diag\p{g_1,\ldots,g_l}\in\mathbb{R}^{l\times l}$ be a sequence of deterministic matrices, satisfying $g_j\geq0$ for all $1\leq j\leq l$ and $\sup_jg_j<\infty$. Denote $\bB_m = \bX_m\bG_l\bX_m^H$, and let $l,m\to\infty$ with fixed $0<c \triangleq m/l<\infty$. Then, for every $\gamma>0$
\begin{align}
\frac{1}{m}\log\det\p{\frac{1}{\gamma}\bB_m+\bI_m}-\eta\p{\gamma}\to 0,\ \ \text{a.s.}
\end{align} 
where
\begin{align}
\eta\p{\gamma} &\triangleq \frac{1}{m}\sum_{j=1}^l\log\pp{1+cg_j\bar{S}\p{-\gamma}}-\log\pp{\gamma\bar{S}\p{-\gamma}}-\frac{1}{l}\sum_{j=1}^l\frac{g_j\bar{S}\p{-\gamma}}{1+cg_j\bar{S}\p{-\gamma}}
\label{InfShnann}
\end{align}
and $\bar{S}\p{z}$ is defined by the unique positive solution of the equation
\begin{align}
\bar{S}\p{z} = \p{\frac{1}{l}\sum_{j=1}^l\frac{g_j}{1+cg_j\bar{S}\p{z}}-z}^{-1}.
\label{stelas}
\end{align}
\end{lemma}

The next lemma deals with the asymptotic behavior of the Stieltjes transform. 
\begin{lemma}[\cite{Sebastian}]\label{lem:mat2}
Let $\bX_m$, $\bG_l$, and $\bB_m$ be defined as in Lemma \ref{lem:mat1}. Let $\bThet_m\in\mathbb{C}^{m\times m}$ be a deterministic sequence of matrices having uniformly bounded spectral norms (with respect to $m$)\footnote{Actually we only need to demand the distribution $F_{\bThett_m}$ to be tight, namely, for all $\epsilon>0$ there exists $M>0$ such that $F_{\bThett_m}\p{M}>1-\epsilon$ for all $m$.}. Then, we a.s. have that
\begin{align}
\frac{1}{m}\tr\p{\bThet_m\p{\bB_m-z\bI_m}^{-1}}-\frac{1}{m}\tr\p{\bThet_m}\bar{S}\p{z}\to 0,\ \text{for all}\ z\in\mathbb{C}\setminus \mathbb{R}_+,
\label{infste}
\end{align}
as $m,l\to\infty$.
\end{lemma}
\begin{remark}
In \cite{Palomar}, the authors propose a somewhat more restrictive (but useful) version of Lemma \ref{lem:mat2}. Assuming that $\bThet_m$ has a uniformly bounded Frobenius norm (for all $m$), they show similarly that
\begin{align}
\abs{\tr\p{\bThet_m\p{\bB_m-z\bI_m}^{-1}}-\tr\p{\bThet_m}\bar{S}\p{z}}\to 0,\ \text{for}\ z\in\mathbb{C}\setminus \mathbb{R}_+.
\label{Palomar}
\end{align}
a.s. as $m,l\to\infty$.
\end{remark}

In order to apply the above results in our analysis, a somewhat more general version will be needed. First, the matrix $\bThet_m$ in the Lemma \ref{lem:mat2} is assumed to be deterministic and bounded (in the spectral or Frobenius senses). In our case, however, we will need to deal with a random matrix $\bThet_m$ which is independent of the other random variables. The following proposition accounts for this problem. The proof readily follows by first conditioning on $\bThet_m$ (which is now random and a.s. bounded) and then applying Lemma \ref{lem:mat2}. 
\begin{prop}\label{lem:mat3} The assertion of Lemma \ref{lem:mat2} holds true also for a random $\bThet_m\in\mathbb{C}^{m\times m}$, which is independent of $\bX_m$, and has a uniformly bounded spectral norm (with respect to $m$) in the a.s. sense. 
\end{prop}
\begin{remark}
In Proposition \ref{lem:mat3}, it is assumed that $\bThet_m$ has uniformly bounded spectral norm (uniformly in $m$) in the a.s. sense, namely,
\begin{align}
\limsup_{m\to\infty}\norm{\bThet_m}<\infty
\end{align}
with probability one. In other words, for every $\epsilon>0$, there exists some positive $M_0$ such that for all $m>M_0$ we have that $\norm{\bThet_m}<D+\epsilon$ for some finite constant $D$. 
\end{remark}

The second issue is regarding the assumption that the ratio $c = m/l$, in the previous lemmas, tends to a strictly positive limit. In our case, however, this limit may be zero. Fortunately, it turns out that the previous results still hold true also in this case, namely, a continuity property w.r.t. $c$. Technically speaking, this fact can be verified by repeating the original proofs \cite{Couillet,Sebastian} of the above results and noticing that the positivity assumption is superfluous\footnote{Specifically, one just need to replace every instance of $m$ by $c\cdot l$ in the original proofs and things go along without any issue.}. To give some sense, consider the following two special cases. First, in case that $m$ is fixed while $l$ goes to infinity (and then $c$ vanishes), using the strong law of large numbers (SLLN), it is easy to see that the previous lemmas indeed hold true. Also, if $m\ll\sqrt{l}$, then a simple approach is to show that the diagonal elements of the matrix $\bX_m\bX_m^T$ concentrate around a fixed value, and that the row sum of off-diagonal terms converges to zero. Then using Gershgorin's circle theorem \cite{Thomas} one obtains the deterministic equivalent. Obviously, these two special cases do not cover the whole range of $m=o(l)$, which, as said, can be shown by repeating the original proofs in \cite{Couillet,Sebastian}.

In the following subsection, we prove Theorem 1. The proof contains several tedious calculations and lemmas, which will relegated to appendices for the sake of convenience.


\subsection{Main Steps in the Proof of Theorem \ref{th:1}}\label{sec:mainsteps}
Let $\bs$ and $\br$ be two binary sequences of length $n$, and let $\calS\triangleq\text{supp}\p{\bs}$ and $\calR\triangleq\text{supp}\p{\br}$ designate their respective supports, defined as $\text{supp}\p{\bs} \triangleq \ppp{i\in\ppp{1,2,\ldots,n}:\;S_i\neq 0}$, and similarly for $\br$. Also, define
\begin{align}
\bQsr\triangleq \sum_{j\in\calS\cap\calR}\bet_{m^s_j}\tilde{\bet}_{m^r_j}^T,\label{Qmatdefdef}
\end{align}
where $\bet_{m^s_j}$ and $\tilde{\bet}_{m^r_j}$ denote unit vectors of size\footnote{For a set $\calA$, we use $\abs{\calA}$ to designate its cardinality.} $\abs{\calS}\times 1$ and $\abs{\calR}\times 1$, having ``$1$" at the indexes $m^s_j\triangleq \sum_{l=1}^js_l$ and $m^r_j\triangleq \sum_{l=1}^jr_l$, respectively. Note that $\bQsr$ can be written as $\bQsr = \bQs\bQr^T$, where $\bQs$ is an $\abs{\calS}\times\abs{\calS\cap\calR}$ matrix with column vectors $\ppp{\bet_{m^s_j}}$, and $\bQr$ is an $\abs{\calR}\times\abs{\calS\cap\calR}$ matrix with column vectors $\ppp{\tilde{\bet}_{m^r_j}}$. Also, it is easy to check that $\bQs\bQs^T$ and $\bQr\bQr^T$ are binary diagonal matrices, with unit values at positions correspond to the supports of $\bs$ and $\br$, respectively. It is then obvious that $\bQs\bQs^T\preceq\bItt$ and $\bQr\bQr^T\preceq\bItr$, where $\bItt$ and $\bItr$ are $\abs\calS\times\abs\calS$ and $\abs\calR\times\abs\calR$ unit matrices, respectively. 
\begin{example}
Let $n=6$, and consider $\bs = \p{1,1,0,0,1,1}$ and $\br = \p{0,1,1,0,0,1}$. Then, $\calS = \ppp{1,2,5,6}$, $\calR = \ppp{2,3,6}$, and thus $\calS\cap\calR = \ppp{2,6}$. Whence $m_2^s = 2$, $m_2^r = 1$, $m_6^s = 4$, and $m_6^r = 3$. Accordingly, the matrix $\bQsr$ is given by
\begin{align}
\bQsr^T = \p{\bet_2\tilde{\bet}_1^T+\bet_4\tilde{\bet}_3^T}^T=
\begin{pmatrix}
  0 & 1 & 0 & 0 \\
  0 & 0 & 0 & 0\\
	0 & 0 & 0 & 1
 \end{pmatrix}.\nonumber
\end{align}
\end{example}

For a vector $\bv$ and a matrix $\bV$, with real-valued entries, we define $\bv_{\bst} \triangleq \left.\bv\right|_{\calS}\in\mathbb{R}^{\abs{\calS}\times1}$ and $\bV_{\bst} \triangleq \left.\bV\right|_{\calS}\in\mathbb{R}^{k\times\abs{\calS}}$, which is the restriction of the entries of $\bv$ and the columns of $\bV$ on the support $\calS$, respectively. For brevity, we define the following quantities:
\begin{align}
&\calHs \triangleq \p{\beta\bHtt^T\bHtt+\frac{1}{\sigma^2}\bItt}^{-1},\\
&\calHsi \triangleq \p{\beta\pp{\bHtt^T\bHtt}_{i}+\frac{1}{\sigma^2}\bItt}^{-1},\\
&\calHsii \triangleq \p{\beta\pp{\bHtt^T\bHtt}_{i,j}+\frac{1}{\sigma^2}\bItt}^{-1},
\end{align}
where $\pp{\bHtt^T\bHtt}_{i}\triangleq \bHtt^T\bHtt-\bz_i\bz_i^T$, and $\pp{\bHtt^T\bHtt}_{i,j}\triangleq \pp{\bHtt^T\bHtt}_{i}-\bz_j\bz_j^T$, in which $\bz_i$ is the $i$th row of $\bHtt$. Finally, we define:
\begin{align}
\xi\p{\by,\bHtt} &\triangleq \exp\left\{\frac{\beta^2}{2}\by^T\bHtt\calHs\bHtt^T\by-\frac{1}{2}\log\det\p{\beta\sigma^2\bHtt^T\bHtt+\bI_{\bst}}\right\},\label{JJJ1}\\
J\p{\by,\bHtt,\bHtr}&\triangleq \frac{\beta^2}{n}\by^T\bHtt\calHs\bQsr\calHr\bHtr^T\by.\label{JJJ3}
\end{align}

Under the model described in Section \ref{sec:model},
\begin{align}
&p_{\bYt\vert\bXt,\bHt}(\by\vert\bH,\bx) = \frac{1}{\p{2\pi/\beta}^{k/2}}\exp\p{-\frac{\beta}{2}\norm{\by-\bH\bx}^2},
\end{align}
and
\begin{align}
&p_{\bXt}(\bx) = \sum_{\bs\in\ppp{0,1}^n}P_{\bSt}(\bs)\prod_{i:\;s_i=0}\delta\p{x_i}\prod_{i:\;s_i=1}\frac{1}{\sqrt{2\pi\sigma^2}}e^{-\frac{1}{2\sigma^2}x_i^2}.
\end{align}
Therefore, the partition function in \eqref{PartFunc} is given by
\begin{align}
&Z\p{\by,\bH;\blam} =  \sum_{\bs\in\ppp{0,1}^n}P_{\bSt}(\bs)\int_{\mathbb{R}^n}e^{\blamt^T\bxt}p_{\bYt\vert\bXt,\bHt}(\by\vert\bH,\bx)\prod_{i:\;s_i=0}\delta\p{x_i}\prod_{i:\;s_i=1}\frac{1}{\sqrt{2\pi\sigma^2}}e^{-\frac{1}{2\sigma^2}x_i^2}\mathrm{d}\bx.
\end{align}
In the following, using Lemma \ref{lemma:partMM}, we provide a generic expression for the MMSE. The proof appears in Appendix \ref{app:0}.
\begin{lemma}\label{lem:MMSe_div}
The normalized MMSE is given by:
\begin{align}
\frac{\text{mmse}\p{\bX\vert\bY,\bH}}{n} = \frac{\sigma^2}{n}\sum_{i=1}^n\bE\ppp{S_i}-\bE\ppp{\bE_{\mu_{s\times r}}\pp{J\p{\bY,\bHtt,\bHtr}}}
\label{MMSEequationLL}
\end{align}
where $\bE_{\mu_{s\times r}}$ denotes the expectation w.r.t. the discrete probability distribution
\begin{align}
\mu\p{\bs\vert\bY,\bH}\times\mu\p{\br\vert\bY,\bH}\triangleq\frac{P_{\bSt}(\bs)P_{\bSt}(\br)\xi\p{\bY,\bHtt}\xi\p{\bY,\bHtr}}{\pp{\sum_{\but\in\ppp{0,1}^n}P_{\bSt}\p{\bu}\xi\p{\bY,\bHtu}}^2}.\label{disdiser}
\end{align}
\end{lemma}

At this stage, the relation to the Stieltjes and Shannon transforms is clear: The structure of the various terms in $\xi\p{\by,\bHtt}$ and $J\p{\by,\bHtt,\bHtr}$ suggest an application of an extended version of the Stieltjes and Shannon transforms of the matrix $\bHtt^T\bHtt$. The following proposition is essentially the core of our analysis; it provides approximations (which are asymptotically exact in the a.s. sense) of $\xi\p{\by,\bHtt}$ and $J\p{\by,\bHtt,\bHtr}$. 
Recall the auxiliary variables defined in \eqref{firstt}-\eqref{lastt}. The following result is proved in Appendix \ref{app:1}.
\begin{prop}[Asymptotic approximations]\label{prop:2}
For every $\epsilon,p>0$
\begin{align}
&\max_{\bst\in\ppp{0,1}^n}\pr\ppp{\abs{\frac{1}{n}\tr\calHs-\sigma^2m_{\bst}b\p{m_{\bst}}}>\epsilon}\leq\frac{1}{\epsilon^p}\calO(n^{-p/2}),\label{itm01}\\
&\max_{\bst\in\ppp{0,1}^n}\pr\ppp{\abs{\frac{1}{n}\bY^T\bHtt\calHs\bHtt^T\bY-f_{n}}>\epsilon}\leq\frac{1}{\epsilon^p}\calO(n^{-p/2}),\label{itm03}\\
&\max_{\bst,\brt\in\ppp{0,1}^n}\pr\ppp{\abs{\frac{1}{n}\bY^T\bHtt\calHs\bQsr\calHr\bHtr^T\bY-q_{n}}>\epsilon}\leq\frac{1}{\epsilon^p}\calO(n^{-p/2})\label{itm04},\\
&\max_{\bst\in\ppp{0,1}^n}\pr\ppp{\abs{\frac{1}{n}\log\det\p{\beta\sigma^2\bHtt^T\bHtt+\bItt}-m_{\bst}\bar{I}\p{m_{\bst}}}>\epsilon}\leq\frac{1}{\epsilon^p}\calO(n^{-p/2}),\label{itm02}
\end{align}
where
\begin{align}
f_{n}\triangleq \beta \frac{\sigma^4b^2\p{m_{\bst}}m_{\bst}^2}{g^2\p{m_{\bst}}}\frac{\norm{\bY}^2}{n}+\frac{\sigma^2b\p{m_{\bst}}}{g^2\p{m_{\bst}}}\frac{\norm{\bHtt^T\bY}^2}{n}\label{101eq},
\end{align}
and
\begin{align}
q_{n}\triangleq& \frac{\tilde\alpha\p{m_{\bst},m_{\brt},m_{\bst,\brt}}}{g\p{m_{\bst}}g\p{m_{\brt}}}\frac{\bY^T\bHtt\bQsr\bHtr^T\bY}{n}\nonumber\\
&-\frac{\tilde\alpha\p{m_{\bst},m_{\brt},m_{\bst,\brt}}}{g\p{m_{\bst}}g\p{m_{\brt}}}\beta\sigma^2m_{\bst,\brt}\p{\frac{b\p{m_{\brt}}}{g\p{m_{\brt}}}\frac{\norm{\bHtr^T\bY}^2}{n}+\frac{b\p{m_{\bst}}}{g\p{m_{\bst}}}\frac{\norm{\bHtt^T\bY}^2}{n}}\nonumber\\
&+\frac{\tilde\alpha\p{m_{\bst},m_{\brt},m_{\bst,\brt}}}{g\p{m_{\bst}}g\p{m_r}}\beta\sigma^2m_{\bst,\brt}\p{\frac{b\p{m_{\brt}}}{g\p{m_{\brt}}}m_{\brt}+\frac{b\p{m_{\bst}}}{g\p{m_{\bst}}}m_{\bst}}\frac{\norm{\bY}^2}{n}
\label{g_Nterm}
\end{align}
where
\begin{align}
\tilde\alpha\p{m_{\bst},m_{\brt},m_{\bst,\brt}}\triangleq\frac{1}{\p{\eta_0(m_{\brt})+\sigma^{-2}}\p{\psi_0\p{m_{\bst},m_{\brt},m_{\bst,\brt}}+\sigma^{-2}}},\label{tildealphadef}
\end{align}
in which
\begin{align}
\eta_0(m_{\brt})\triangleq \frac{\beta R}{1+\beta\sigma^2m_{\brt}b\p{m_{\brt}}},
\end{align}
and
\begin{align}
\psi_0\p{m_{\bst},m_{\brt},m_{\bst,\brt}} \triangleq\eta_0(m_{\brt})- \frac{\beta^2\sigma^2Rb\p{m_{\bst}}m_{\bst,\brt}}{\p{1+\beta\sigma^2m_{\bst}b\p{m_{\bst}}}\p{1+\beta\sigma^2m_{\brt}b\p{m_{\brt}}}}.
\end{align} 
\end{prop}
The next step is to apply Proposition \ref{prop:2} to the obtained MMSE. Let $\epsilon>0$, and define
\begin{align}
&\calT^{\bst,\brt}_\epsilon \triangleq \left\{\by\in\mathbb{R}^{k\times 1},\bH\in\mathbb{R}^{k\times n}:\ \abs{\frac{1}{n}\tr\calHs-\sigma^2m_{\bst}b\p{m_{\bst}}}<\epsilon,\right.\abs{\frac{1}{n}\by^T\bHtt\calHs\bHtt^T\by-f_{n}}<\epsilon,\nonumber\\
&\ \ \ \ \ \ \ \ \ \left.\abs{\frac{1}{n}\by^T\bHtt\calHs\bQsr\calHr\bHtr^T\by-q_{n}}<\epsilon,\;\abs{\frac{1}{n}\log\det\p{\beta\sigma^2\bHtt^T\bHtt+\bItt} - m_{\bst}\bar{I}\p{m_{\bst}}}<\epsilon\right\}.\label{typicalSR}
\end{align}
By Proposition \ref{prop:2}, this set has probability tending to one as $k,n\to\infty$, and so, we shall call it a ``typical" set containing $\ppp{\by,\bH}$-pairs of typical observation vectors and sensing matrices. This is summarized in the forthcoming corollary.
\begin{corollary}\label{cor:2}
\begin{align}
&\lim_{n\to\infty}\max_{\bst,\brt\in\ppp{0,1}^n}\pr\ppp{\p{\calT^{\bst,\brt}_\epsilon}^c} =0.
\end{align}
\end{corollary}
\begin{proof}
The result follows directly by using the union bound and Proposition \ref{prop:2}.
\end{proof}

The following main observation is that for the asymptotic evaluation of \eqref{MMSEequationLL}, only typical events (i.e., those defined in \eqref{typicalSR}) are the dominant. Specifically, the MMSE can be decomposed as follows
\begin{align}
\frac{\text{mmse}\p{\bX\vert\bY,\bH}}{n} &= \sigma^2\frac{1}{n}\sum_{i=1}^n\bE\ppp{S_i}-\bE\ppp{\bE_{\mu_{s\times r}}\pp{J\p{\bY,\bHtt,\bHtr}}}\\
& = \sigma^2\frac{1}{n}\sum_{i=1}^n\bE\ppp{S_i}-\bE\ppp{J\p{\bY,\bH_{\tilde\bSti},\bH_{\tilde\bRti}}}\label{tildeexpect}\\
&= \sigma^2\frac{1}{n}\sum_{i=1}^n\bE\ppp{S_i}-\bE\ppp{J\p{\bY,\bH_{\tilde\bSti},\bH_{\tilde\bRti}}\Ind_{\calT_\epsilon^{\tilde s,\tilde r}}}\nonumber\\
&\ \ \ \ \ \ \ \ \ \ \ \ \ \ \ \ \ \ \ \ \ \ -\bE\ppp{J\p{\bY,\bH_{\tilde\bSti},\bH_{\tilde\bRti}}\Ind_{\p{\calT_\epsilon^{\tilde s,\tilde r}}^c}},
%
\label{lastee}
\end{align}
where $\Ind_{\calT_\epsilon^{s,r}}$ is as a shorthand notation for $\Ind_{\ppp{(\bYt,\bHt)\in\calT_\epsilon^{s,r}}}$, and the expectations in \eqref{lastee} are taken w.r.t. the joint distribution of $(\tilde\bS,\tilde\bR,\bY,\bH)$, where $P_{\tilde\bSt,\tilde\bRt\vert\bYt,\bHt} = \mu(\tilde\bS\vert\bY,\bH)\times\mu(\tilde\bR\vert\bY,\bH)$. We claim that the last term at the r.h.s. of \eqref{lastee} is asymptotically negligible. First, by Cauchy-Schwartz inequality, we have
\begin{align}
\abs{\bE\ppp{J\p{\bY,\bH_{\tilde\bSti},\bH_{\tilde\bRti}}\Ind_{\p{\calT_\epsilon^{\tilde s,\tilde r}}^c}}}^2\leq \bE\abs{J\p{\bY,\bH_{\tilde\bSti},\bH_{\tilde\bRti}}}^2\Pr\ppp{\p{\calT_\epsilon^{\bSti,\bRti}}^c}.\label{probchas}
\end{align}
Due to Corollary \ref{cor:2}, we have
\begin{align}
\lim_{n\to\infty}\pr\ppp{\p{\calT^{\bSti,\bRti}_\epsilon}^c}\leq\lim_{n\to\infty}\max_{\bst,\brt\in\ppp{0,1}^n}\pr\ppp{\p{\calT^{\bst,\brt}_\epsilon}^c}=0.\label{tul2}
\end{align}
Therefore, according to \eqref{probchas}, in order to show that the last term at the r.h.s. of \eqref{lastee} is asymptotically negligible, we will show that the expectation at the r.h.s. of \eqref{probchas} is finite and independent of $n$, that is,
\begin{align}
\bE\abs{J\p{\bY,\bH_{\tilde\bSti},\bH_{\tilde\bRti}}}^2\leq M<\infty.
\end{align}
To this end, recall that (see, \eqref{Qmatdefdef}, and the discussion that follows) the matrix $\bQsr$ can be represented as $\bQsr = \bQs\bQr^T$. We get
\begin{align}
\abs{J\p{\by,\bHtt,\bHtr}}^2 &= \frac{\beta^4}{n^2}\abs{\by^T\bHtt\calHs\bQsr\calHr\bHtr^T\by}^2\\
& = \frac{\beta^4}{n^2}\abs{\by^T\bHtt\calHs\bQs\bQr^T\calHr\bHtr^T\by}^2\\
&\leq \frac{\beta^4}{n^2}\norm{\by^T\bHtt\calHs\bQs}_2^2\norm{\by^T\bHtr\calHr\bQr}_2^2\\
&\leq \frac{\beta^4}{n^2}\norm{\by^T\bHtt\calHs}_2^2\norm{\by^T\bHtr\calHr}_2^2\\
&\leq \frac{\beta^4}{2n^2}\pp{\norm{\by^T\bHtt\calHs}_2^4+\norm{\by^T\bHtr\calHr}_2^4}
\end{align}
where the first inequality is due to Cauchy-Schwartz inequality, the second inequality is due to the fact that $\bQs\bQs^T\preceq\bItt$ for all $\bs$, and the last inequality follows from the simple inequality $2ab\leq a^2+b^2$, for $a,b\in\mathbb{R}$. Whence,
\begin{align}
\bE\abs{J\p{\bY,\bH_{\tilde\bSti},\bH_{\tilde\bRti}}}^2&\leq \frac{\beta^4}{2n^2}\bE\ppp{\norm{\bY^T\bH_{\tilde\bSti}\calHst}_2^4+\norm{\bY^T\bH_{\tilde\bRti}\calHrt}_2^4}\\
&= \frac{\beta^4}{n^2}\bE\norm{\bY^T\bH_{\tilde\bSti}\calHst}_2^4\\
& \leq \frac{\sigma^8\beta^2}{n^2}\bE\norm{\bY^T\bH_{\tilde\bSti}}_2^4\\
&\leq \frac{8\sigma^8\beta^2}{n^2}\bE\ppp{\norm{\bX^T\bH^T\bH_{\tilde\bSti}}_2^4+\norm{\bN^T\bH_{\tilde\bSti}}_2^4}
\label{tul1}
\end{align}
where in the second inequality we have used the fact that $\calHs\preceq\sigma^2\bItt$, and the last inequality follows from the fact that $\bY = \bH\bX+\bN$, and the inequality $\norm{\ba+\bb}_2^4\leq 8\cdot(\norm{\ba}_2^4+\norm{\bb}_2^4)$, for any $\ba,\bb\in\mathbb{R}^n$. We claim that the last two terms at the r.h.s. of \eqref{tul1} are finite. Indeed, for example,
\begin{align}
\frac{8\sigma^8\beta^2}{n^2}\bE\norm{\bN^T\bH_{\tilde\bSti}}_2^4 &=\frac{8\sigma^8\beta^2}{n^2}\bE\p{\sum_{i\in\tilde\bSti}\p{\bN^T\bh_i}^2}^2\\
&\leq \frac{8\sigma^8\beta^2}{n^2}\bE\p{\sum_{i=1}^n\p{\bN^T\bh_i}^2}^2\\
&\leq \frac{8\sigma^8\beta^2}{n}\sum_{i=1}^n\bE\ppp{\bN^T\bh_i}^4\\
& = 8\sigma^8\beta^2\bE\ppp{\bN^T\bh_1}^4,\label{boundedvarnoise}
\end{align}
where $\bh_i$ is the $i$th column of $\bH$, the first inequality follows from the fact that $\sum_{i=1}^n\tilde S_i\leq n$ w.p. 1, and the second inequality is due to the fact that $(\sum_{i=1}^na_i)^r\leq n^{r-1}\sum_{i=1}^na_i^r$, for any $r\in\mathbb{N}$. Given $\bh_1$, the random variable $\bN^T\bh_1$ is Gaussian, with zero mean, and variance $\beta^{-1}\norm{\bh_1}_2^2$. Thus, $\bE\pp{\p{\bN^T\bh_1}^4\vert\bh_1} = 3\beta^{-2}\norm{\bh_1}_2^4$. Therefore,
\begin{align}
8\sigma^8\beta^2\bE\ppp{\bN^T\bh_1}^4 &= 24\sigma^8\bE\norm{\bh_1}_2^4\\
&= 24\sigma^8\bE\p{\sum_{i=1}^kH_{1,i}^2}^2\\
&\leq 24k\sigma^8\bE\p{\sum_{i=1}^kH_{1,i}^4}\\
& = 24k^2\sigma^8\bE(H_{1,1}^4) = 24\sigma^8\bE(\sqrt{k}H_{1,1})^4<\infty
\end{align}
where the first inequality follows from $(\sum_{i=1}^na_i)^r\leq n^{r-1}\sum_{i=1}^na_i^r$, and the last inequality is due to the assumption $\bE(\sqrt{n}H_{1,1})^4<\infty$. In the same way, it can be shown that
\begin{align}
\frac{8\sigma^8\beta^2}{n^2}\bE\norm{\bX^T\bH^T\bH_{\tilde\bSti}}^4<\infty,
\end{align}
and thus the term at the r.h.s. of \eqref{tul1} is finite, that is,
\begin{align}
\bE\abs{J\p{\bY,\bH_{\tilde\bSti},\bH_{\tilde\bRti}}}^2&<\infty.\label{finiteexpectte}
\end{align}
To conclude, using \eqref{probchas}, \eqref{tul2}, and \eqref{finiteexpectte}, we get
\begin{align}
\bE\ppp{J\p{\bY,\bH_{\tilde\bSti},\bH_{\tilde\bRti}}\Ind_{\p{\calT_\epsilon^{\tilde s,\tilde r}}^c}}\to0,\label{negdomin}
\end{align} 
as $k,n\to\infty$.
Accordingly, for the asymptotic calculation of the MMSE, only the first two terms at the r.h.s. of \eqref{lastee} prevail, and the calculation of the asymptotic MMSE boils down to the calculation of:
\begin{align}
\limsup_{n\to\infty}\frac{\text{mmse}\p{\bX\vert\bY,\bH}}{n} &= \limsup_{n\to\infty}\left[\sigma^2\frac{1}{n}\sum_{i=1}^n\bE\ppp{S_i}-\bE\ppp{J\p{\bY,\bHttt,\bHtrt}\Ind_{\calT_\epsilon^{\tilde s,\tilde r}}}\right]\\
&=\limsup_{n\to\infty}\left[\sigma^2\frac{1}{n}\sum_{i=1}^n\bE\ppp{S_i}-\bE\ppp{\bE_{\mu_{s\times r}}\pp{J\p{\bY,\bHtt,\bHtr}\Ind_{\calT_\epsilon^{s,r}}}}\right]\\
& = \sigma^2m_a - \limsup_{n\to\infty}\bE\ppp{\bE_{\mu_{s\times r}}\pp{J\p{\bY,\bHtt,\bHtr}\Ind_{\calT_\epsilon^{s,r}}}}
\label{lastee2}
\end{align}
where the last equality is due to \eqref{expecmagaprior}. Using Proposition \ref{prop:2}, and large deviations theory, the asymptotic MMSE, given in Theorem 1, is derived in Appendix \ref{app:3}. 

\section{Conclusion}
\label{sec:Conclusion}
In this paper, we considered the calculation of the asymptotic MMSE under sparse representation modeling. As opposed to the popular worst-case approach, we adopt a statistical framework for compressed sensing by modeling the input signal as a random process rather than as an individual sequence. In contrast to previous derivations, which were based on the (non-rigorous) replica method, the analysis carried out in this paper is rigorous. The derivation builds upon a simple relation between the MMSE and a
certain function, which can be viewed as a partition function, and hence can be analyzed using methods of statistical mechanics. It was shown that the MMSE can be represented in a special form that contains functions of the Stieltjes and Shannon transforms. This observation allowed us to invoke some powerful results from RMT concerning the asymptotic behavior of these transforms. Although our asymptotic MMSE formula seems to be different from the one that is obtained by the replica method, numerical calculations suggest that they are actually the same. This supports the results of the replica method. 

Finally, we believe that the tools developed in this paper, for handling the MMSE, can be used in order to obtain the MMSE estimator itself. An example for such calculation can be found in a recent paper \cite{Wasim}, where the MMSE (or, more generally, the mismatched MSE), along with the estimator itself, were derived for a model of a codeword (from a randomly selected code), corrupted by a Gaussian vector channel. Also, we believe that our results, can be generalized to the case of mismatch, namely, mismatched compressed sensing. An example for an interesting mismatch model could be a channel mismatch, namely, the receiver has a wrong assumption on the channel $\bH$, which can be modeled as $\hat{\bH} = \tau\bH+\sqrt{1-\tau^2}\bQ$, where $\bQ$ is some random matrix, independent of $\bH$, and $0\leq\tau\leq1$ quantifies the proximity between $\hat{\bH}$ and $\bH$. Another mismatch configuration could be noise-variance mismatch, namely, the receiver has wrong knowledge about the noise variance. It is then interesting to investigate the resulted MSE in these cases, and in particular, to check whether there are new phase transitions caused by the mismatch.


\appendices
\numberwithin{equation}{section}

\section{}\label{app:0}
\begin{proof}[Proof of Lemma \ref{lem:MMSe_div}]
Under our model, the partition function in \eqref{PartFunc} is given by
\begin{align}
&Z\p{\by,\bH;\blam} =  \sum_{\bs\in\ppp{0,1}^n}P_{\bSt}(\bs)\int_{\mathbb{R}^n}\mathrm{d}\bx\frac{\exp\p{-\beta\norm{\by-\bH\bx}^2/2+\blam^T\bx}}{\p{2\pi/\beta}^{k/2}}\nonumber\\
& \ \ \  \ \ \ \ \ \ \ \ \ \ \ \ \ \ \ \ \ \ \ \ \ \ \ \ \ \ \ \ \ \ \ \ \ \ \ \ \times\prod_{i:\;s_i=0}\delta\p{x_i}\prod_{i:\;s_i=1}\frac{1}{\sqrt{2\pi\sigma^2}}e^{-\frac{1}{2\sigma^2}x_i^2}.
\end{align}
First, note that
\begin{align}
\norm{\by-\bH\bx}^2\prod_{s_i=0}\delta\p{x_i} &= \pp{\norm{\by}^2-2\sum_{i\in\calS}\bh_i^T\by x_i+\sum_{i,j\in\calS}x_ix_j\bh_i^T\bh_j}\prod_{s_i=0}\delta\p{x_i}\\
&  =\pp{\norm{\by}^2-2\bx_{\bst}^T\bH_{\bst}^T\by+\bx_{\bst}^T\bH^T_{\bst}\bH_{\bst}\bx_{\bst}}\prod_{s_i=0}\delta\p{x_i}
\end{align}
where $\bh_i$ denotes the $i$th column of $\bH$, and similarly,
\begin{align}
\blam^T\bx\prod_{s_i=0}\delta\p{x_i} &= \p{\sum_{i\in\calS}x_i\lambda_i}\prod_{s_i=0}\delta\p{x_i}\\
& = \blam_{\bst}^T\bx_{\bst}\prod_{s_i=0}\delta\p{x_i}.
\end{align}
Using the fact that $\delta\p{\cdot}$ is a measure on $\mathbb{R}$, one may conclude that 
\begin{align}
Z\p{\by,\bH;\blam} &=  \sum_{\bs\in\ppp{0,1}^n}P_{\bSt}(\bs)\frac{1}{\p{2\pi/\beta}^{k/2}}\frac{1}{\p{\sqrt{2\pi\sigma^2}}^{\abs{\calS}}}\exp\p{-\frac{\beta}{2}\norm{\by}^2}\nonumber\\
&\ \ \ \ \ \ \ \ \times\int_{\mathbb{R}^{\abs{\calS}}}\exp\p{-\bxtt^T\p{\frac{\beta}{2}\bH_{\bst}^T\bH_{\bst}+\frac{1}{2\sigma^2}\bI_{\bst}}\bx_{\bst}+\bx_{\bst}^T\p{\blam_{\bst}+\beta\bH_{\bst}^T\by}}
\mathrm{d}\bx_{\bst}\\
&=  \sum_{\bs\in\ppp{0,1}^n}\frac{P_{\bSt}(\bs)\exp\p{-\frac{\beta}{2}\norm{\by}^2}}{\p{2\pi/\beta}^{k/2}\p{\sqrt{2\pi\sigma^2}}^{\abs{\calS}}\det^{1/2}\pp{\frac{1}{2\pi}\p{\beta\bHtt^T\bHtt+\frac{1}{\sigma^2}\bI_{\bst}}}}\nonumber\\
&\ \ \ \ \ \ \ \ \ \times\exp\ppp{\frac{1}{2}\p{\beta\bHtt^T\by+\blamtt}^T\p{\beta\bHtt^T\bHtt+\frac{1}{\sigma^2}\bI_{\bst}}^{-1}\p{\beta\bHtt^T\by+\blamtt}}\\
&= C\cdot \sum_{\bs\in\ppp{0,1}^n}P_{\bSt}(\bs)\frac{\exp\ppp{\frac{1}{2}\p{\beta\bHtt^T\by+\blamtt}^T\calHs\p{\beta\bHtt^T\by+\blamtt}}}{\sqrt{\det\p{\beta\sigma^2\bHtt^T\bHtt+\bI_{\bst}}}}
\label{partitionProof}
\end{align}
where $C$ is independent of $\blam$, but it depends on $\beta$ and $\by$. We are now in a position to find a preliminary expression of the MMSE, using Lemma \ref{lemma:partMM}. Let 
\begin{align}
\xi\p{\by,\bHtt,\blamtt} \triangleq &\exp\left\{\frac{1}{2}\p{\beta\bHtt^T\by+\blamtt}^T\calHs\p{\beta\bHtt^T\by+\blamtt}-\frac{1}{2}\log\det\p{\beta\sigma^2\bHtt^T\bHtt+\bI_{\bst}}\right\},
\label{JJJ1ap}
\end{align}
and therefore
\begin{align}
Z\p{\by,\bH;\blam} = C\cdot \sum_{\bs\in\ppp{0,1}^n}P_{\bSt}(\bs)\xi\p{\by,\bHtt,\blamtt}.
\label{partitionProof2}
\end{align}
Now,
\begin{align}
\frac{\partial}{\partial\lambda_i}\ppp{\frac{1}{2}\p{\beta\bHtt^T\by+\blamtt}^T\calHs\p{\beta\bHtt^T\by+\blamtt}} = \bet^T_i\calHs\p{\beta\bHtt^T\by+\blamtt}\Ind_{i\in\calS},
\end{align}
and thus
\begin{align}
\frac{\partial}{\partial\lambda_i}\xi\p{\by,\bHtt,\blamtt} = \bet^T_i\calHs\p{\beta\bHtt^T\by+\blamtt}\Ind_{i\in\calS}\xi\p{\by,\bHtt,\blamtt}.
\label{partitionProof3}
\end{align}
Recall that for a positive, twice differential function $f$,
\begin{align}
&\frac{\mathrm{d}}{\mathrm{d} x}\log f\p{x} = \frac{1}{f\p{x}}\p{\frac{\mathrm{d}}{\mathrm{d}x}f\p{x}}.\label{divsim1}
\end{align}
Thus, using \eqref{partitionProof2} and \eqref{partitionProof3}, we have that (for $1\leq i\leq n$), 
\begin{align}
\frac{\partial}{\partial\lambda_i}\log Z\p{\by,\bH;\blam} = \frac{\sum_{\bs\in\ppp{0,1}^n}P_{\bSt}(\bs)\bet^T_i\calHs\p{\beta\bHtt^T\by+\blamtt}\Ind_{i\in\calS}\xi\p{\by,\bHtt,\blamtt}}{Z\p{\by,\bH;\blam}}.\label{MMSEmidmis0}
\end{align}
Let $\xi\p{\by,\bHtt}\triangleq\xi\p{\by,\bHtt,\bze}$. We have:
\begin{align}
\left.\frac{\partial}{\partial\lambda_i}\log Z\p{\by,\bH;\blam}\right|_{\blamt=0} &= \frac{\sum_{\bs\in\ppp{0,1}^n}P_{\bSt}(\bs)\bet^T_i\calHs\beta\bHtt^T\by\Ind_{i\in\calS}\xi\p{\by,\bHtt}}{\sum_{\bs\in\ppp{0,1}^n}P_{\bSt}(\bs)\xi\p{\by,\bHtt}}.
\label{MMSEmidmis}
\end{align}
By Lemma \ref{lemma:partMM}, the MMSE is given by
\begin{align}
\frac{\text{mmse}\p{\bX\vert\bY,\bH}}{n} &= \frac{1}{n}\sum_{i=1}^n\pp{\bE\ppp{X_i^2}-\bE\ppp{\left.\pp{\frac{\partial}{\partial\lambda_i}\log Z\p{\by,\bH;\blam}}^2\right|_{\blamt=0}}}\\
&=\frac{\sigma^2}{n}\sum_{i=1}^n\bE\ppp{S_i}-\frac{1}{n}\sum_{i=1}^n\bE\ppp{\left.\pp{\frac{\partial}{\partial\lambda_i}\log Z\p{\by,\bH;\blam}}^2\right|_{\blamt=0}}.
\end{align}
Recall that for an $n\times n$ matrix $\bA$, the trace operator can be represented as $\tr\p{\bA} = \sum_{i=1}^n\hat{\bet}^T_i\bA\hat{\bet}_i$ where $\hat{\bet}_i$ is the $i$th column of the $n\times n$ identity matrix. Thus, we have that
\begin{align}
&\pp{\sum_{\bs\in\ppp{0,1}^n}P_{\bSt}(\bs)\bet^T_i\calHs\beta\bHtt^T\by\Ind_{i\in\calS}\xi\p{\by,\bHtt}}^2 &\nonumber\\
&= \sum_{\bs\in\ppp{0,1}^n}\sum_{\br\in\ppp{0,1}^n}P_{\bSt}(\bs)P_{\bSt}(\br)\bet^T_i\calHs\beta^2\bHtt^T\by\by^T\bHtr\calHr\tilde{\bet}_i\Ind_{i\in\calS\cap\calR}\xi\p{\by,\bHtt}\xi\p{\by,\bHtr}.\label{idepend}
\end{align}
Note that $\bs$ and $\br$ may not have the same support, and in particular, they may not have even the same support size. This explains the appearance of $\tilde{\bet}_i$ which is of size $\abs{\calR}\times 1$. Next, summing the terms that depend on $i$ in \eqref{idepend}, over $1\leq i\leq n$, we get
\begin{align}
&\sum_{i=1}^n\bet^T_i\calHs\beta^2\bHtt^T\by\by^T\bHtr\calHr\tilde{\bet}_i\Ind_{i\in\calS\cap\calR} = \sum_{i=1}^n\tr\p{\bet^T_i\calHs\beta^2\bHtt^T\by\by^T\bHtr\calHr\tilde{\bet}_i\Ind_{i\in\calS\cap\calR}}\\
& = \tr\p{\calHs\beta^2\bHtt^T\by\by^T\bHtr\calHr\sum_{i=1}^n\tilde{\bet}_i\bet^T_i\Ind_{i\in\calS\cap\calR}}\\
& = \beta^2\by^T\bHtt\calHs\bQsr\calHr\bHtr^T\by,
\end{align}
where we have used the fact that
\begin{align}
\bQsr^T = \sum_{i=1}^n\tilde{\bet}_i\bet^T_i\Ind_{i\in\calS\cap\calR}.
\end{align}
Let $J\p{\by,\bHtt,\bHtr}$ be defined as in \eqref{JJJ3}. Then, we obtain
\begin{align}
&\frac{1}{n}\sum_{i=1}^n\bE\ppp{\left.\pp{\frac{\partial}{\partial\lambda_i}\log Z\p{\by,\bH;\blam}}^2\right|_{\blamt=0}} = \frac{1}{n}\sum_{i=1}^n\frac{\pp{\sum_{\bs\in\ppp{0,1}^n}P_{\bSt}(\bs)\bet^T_i\calHs\beta\bHtt^T\by\Ind_{i\in\calS}\xi\p{\by,\bHtt}}^2}{\pp{\sum_{\bs\in\ppp{0,1}^n}P_{\bSt}(\bs)\xi\p{\by,\bHtt}}^2}\nonumber\\
& \ \ \ \ \ \ \ \ \ \ \ \ \ = \frac{\sum_{\bst,\brt\in\ppp{0,1}^n}P_{\bSt}(\bs)P_{\bSt}(\br)J\p{\by,\bHtt,\bHtr}\xi\p{\by,\bHtt}\xi\p{\by,\bHtr}}{\p{\sum_{\bs\in\ppp{0,1}^n}P_{\bSt}(\bs)\xi\p{\by,\bHtt}}^2}.
\end{align}
Thus, the MMSE can be represented as
\begin{align}
\frac{\text{mmse}\p{\bX\vert\bY,\bH}}{n} = \frac{\sigma^2}{n}\sum_{i=1}^n\bE\ppp{S_i}-\bE\ppp{\bE_{\mu_{s\times r}}\pp{J\p{\bY,\bHtt,\bHtr}}}
\label{MMSEequationLLap}
\end{align}
where $\bE_{\mu_{s\times r}}$ denotes the expectation taken w.r.t. the discrete probability distribution $\mu\p{\bs\vert\bY,\bH}\times\mu\p{\br\vert\bY,\bH}$, defined in \eqref{disdiser}.
\end{proof}


\section{Proof of Proposition \ref{prop:2}}\label{app:1}
\subsection{A note on Stieltjes transform}\label{app2a}
Before delving into the proofs of \eqref{itm01}-\eqref{itm02}, we make a short comment on the Stieltjes transform of the matrix
\begin{align}\calHs=\p{\beta\bHtt^T\bHtt+\frac{1}{\sigma^2}\bItt}^{-1} = \sigma^2\p{\beta\sigma^2\bHtt^T\bHtt+\bItt}^{-1},\label{centralQuant}\end{align}
which appears in Proposition \ref{prop:2}. Lemma \ref{lem:mat2} provides the asymptotic behavior of the Stieltjes transform of \eqref{centralQuant}. In order to use this lemma, one needs to calculate $\bar S\p{z}$ given in \eqref{stelas}. For our problem, we substitute: $\bX = \bHtt^T$, $\bG = \beta\sigma^2R\bItt$, $\bThet_m = \bI_m$, $c=\abs{\calS}/k = m_{\bst}/R$, which yields $\bB = \bX\bG\bX^T = \beta\sigma^2\bHtt^T\bHtt$. Then, using \eqref{stelas} for $z = -1$, we find that $\bar S\p{-1}$ is given by the solution of the equation
$$
\bar S\p{-1} = \p{\frac{1}{\abs{\calS}}\sum_{l=1}^{\abs{\calS}}\frac{g_l}{1+cg_l\bar S\p{-1}}+1}^{-1}.
$$
Substituting $g_l = \beta\sigma^2R$ (independently of $l$) and $c = m_{\bst}/R$,  we obtain
\begin{align}
\bar S\p{-1} &= \p{\frac{\beta\sigma^2R}{1+\beta\sigma^2R\frac{m_{\bst}}{R}\bar S\p{-1}}+1}^{-1}\\
& = \frac{1+\sigma^2\beta m_{\bst}\bar S\p{-1}}{1+\beta\sigma^2R+\beta\sigma^2m_{\bst}\bar S\p{-1}},
\label{beqqq1}
\end{align}
whose solution is
\begin{align}
\bar S\p{-1} = \frac{-\pp{1+\beta\sigma^2\p{R-m_{\bst}}}+\sqrt{\pp{1+\beta\sigma^2\p{R-m_{\bst}}}^2+4\beta\sigma^2m_{\bst}}}{2\beta\sigma^2m_{\bst}}.
\label{bTerm1}
\end{align}
Note that $\bar S\p{-1}$ is recognized as $b\p{m_{\bst}}$ defined in \eqref{firstt}, which will be used from now on. It follows, by Lemma \ref{lem:mat2}, that
\begin{align}
\frac{1}{n}\tr\calHs-m_{\bst}\sigma^2b(m_{\bst})\to0,
\end{align}
a.s. as $n\to\infty$.


\subsection{Proof of \eqref{itm01} and \eqref{itm03}}
We start with \eqref{itm03}. Eq. \eqref{itm01} will follow from \eqref{itm03}, as will be shown in the sequel. Let $\bz_i$ denote the $i$th row of the matrix $\bHtt$, and hence
\begin{align}
\bHtt^T\bY = \sum_{i=1}^kY_i\bz_i.
\end{align}
Thus,
\begin{align}
\frac{1}{n}\bY^T\bHtt\calHs\bHtt^T\bY &= \frac{1}{n}\sum_{i=1}^kY_i^2\bz_i^T\calHs\bz_i+\frac{1}{n}\sum_{i\neq j}^kY_iY_j\bz_i^T\calHs\bz_j.
\label{twotermsasym}
\end{align}
We next analyze the two terms at the r.h.s. of \eqref{twotermsasym} separately. First, recall that for any triplet of random variables $(X,Y,Z)$, and $\epsilon>0$, the following holds
\begin{align}
\pr\ppp{\abs{X-Z}>\epsilon}\leq\pr\ppp{\abs{X-Y}>\frac{\epsilon}{2}}+\pr\ppp{\abs{Y-Z}>\frac{\epsilon}{2}}.\label{simpfact}
\end{align}
Define:
\begin{align}
&f_{n,1} \triangleq \frac{m_{\bst}\sigma^2b\p{m_{\bst}}}{1+\beta\sigma^2m_{\bst}b\p{m_{\bst}}}\frac{\norm{\bY}^2}{n},\\
&f_{n,2} \triangleq \frac{\sigma^2b\p{m_{\bst}}}{\p{1+\beta\sigma^2m_{\bst}b\p{m_{\bst}}}^2}\pp{\frac{\norm{\bHtt^T\bY}^2}{n}-m_{\bst}\frac{\norm{\bY}^2}{n}}.
\end{align}
It is easy to verify that $f_n = f_{n,1}+f_{n,2}$, where $f_n$ is defined in \eqref{101eq}. Thus, according to \eqref{simpfact}, to prove \eqref{itm03}, it is sufficient to show that
\begin{align}
\max_{\bst\in\ppp{0,1}^n}\pr\ppp{\abs{\frac{1}{n}\sum_{i=1}^kY_i^2\bz_i^T\calHs\bz_i-f_{n,1}}>\epsilon}\leq\frac{1}{\epsilon^p}\calO(n^{-p/2})\label{uniformconFirst}
\end{align}
and
\begin{align}
&\max_{\bst\in\ppp{0,1}^n}\pr\left\{\left|\frac{1}{n}\sum_{i\neq j}^kY_iY_j\bz_i^T\calHs\bz_j-f_{n,2}\right|>\epsilon\right\}\leq\frac{1}{\epsilon^p}\calO(n^{-p/2}).\label{uniformconFirst2}
\end{align}
Since the same arguments that will be used to prove \eqref{uniformconFirst} can be used to prove \eqref{uniformconFirst2}, for the sake of brevity, in the following, we will prove only \eqref{uniformconFirst}.

First, note that
\begin{align}
\bHtt^T\bHtt = \sum_{i=1}^k\bz_i\bz_i^T.
\end{align}
Using the matrix inversion lemma (Lemma \ref{aplem:1}), we get
\begin{align}
\frac{1}{n}\sum_{i=1}^kY_i^2\bz_i^T\calHs\bz_i = \frac{1}{n}\sum_{i=1}^kY_i^2\frac{\bz_i^T\calHsi\bz_i}{1+\beta\bz_i^T\calHsi\bz_i}.\label{matinvlem}
\end{align}
Clearly, $\calHsi$ is statistically independent of $\bz_i$. Consider the following lemmas.
\begin{lemma}\label{lem:subsub1}
For any $p,\epsilon>0$,
\begin{align}
&\max_{\bst\in\ppp{0,1}^n}\pr\ppp{\abs{\frac{1}{n}\sum_{i=1}^kY_i^2\p{\frac{\bz_i^T\calHsi\bz_i}{1+\beta\bz_i^T\calHsi\bz_i}-\frac{\frac{1}{n}\tr\calHs}{1+\beta\frac{1}{n}\tr\calHs}}}>\epsilon}\leq\frac{1}{\epsilon^p}\calO(n^{-p/2}).\label{showthisinProb}
\end{align} 
\end{lemma}
\begin{lemma}\label{lem:subsub2}
For any $p,\epsilon>0$,
\begin{align}
&\max_{\bst\in\ppp{0,1}^n}\pr\ppp{\abs{\frac{\frac{1}{n}\tr\calHs}{1+\beta\frac{1}{n}\tr\calHs}\frac{\norm{\bY}^2}{n}-f_{n,1}}>\epsilon}\leq\frac{1}{\epsilon^p}\calO(n^{-p/2}).\label{showthisinProb2}
\end{align}
\end{lemma}

Using Lemmas \ref{lem:subsub1} and \ref{lem:subsub2}, it is easy to see that \eqref{uniformconFirst} follows by using \eqref{simpfact} once again. We end this subsection by proving these lemmas.

\begin{proof}[\underline{Proof of Lemma \ref{lem:subsub1}}]
To obtain \eqref{showthisinProb}, by Lemma \ref{aux:assympLLN}, it is sufficient to prove that
\begin{align}
&\max_{1\leq i\leq k}\bE\ppp{Y_i^{2p}\abs{\frac{\bz_i^T\calHsi\bz_i}{1+\beta\bz_i^T\calHsi\bz_i}-\frac{\frac{1}{n}\tr\calHs}{1+\beta\frac{1}{n}\tr\calHs}}^p}\leq\calO(n^{-(1+\delta)}),
\label{showthis0m}
\end{align} 
for some $\delta>0$. Using the Cauchy-Schwartz inequality and the fact that $\bE Y_i^{4p}$ is bounded (similarly as in \eqref{boundedvarnoise}), it is enough to prove\footnote{The exponentiation in \eqref{showthis0} should be $\tilde p=2p$, due to Cauchy-Schwartz inequality. However, since $p$ is arbitrary, we use $p$ instead of $\tilde p$.}
\begin{align}
&\max_{1\leq i\leq k}\bE\ppp{\abs{\frac{\bz_i^T\calHsi\bz_i}{1+\beta\bz_i^T\calHsi\bz_i}-\frac{\frac{1}{n}\tr\calHs}{1+\beta\frac{1}{n}\tr\calHs}}^p}\leq\calO(n^{-(1+\delta)}).
\label{showthis0}
\end{align} 
Finally, instead of showing \eqref{showthis0}, we equivalently show\footnote{The equivalence readily follows by adding and subtracting a common term and then using the triangle inequality.}
\begin{align}
\max_{1\leq i\leq k}\bE\ppp{\abs{\frac{\bz_i^T\calHsi\bz_i-\frac{1}{n}\tr\calHs}{1+\beta\bz_i^T\calHsi\bz_i}}^p}\leq\calO(n^{-(1+\delta)}),
\label{showthis}
\end{align}
and
\begin{align}
\max_{1\leq i\leq k}\bE\ppp{\abs{\frac{\frac{1}{n}\tr\calHs}{1+\beta\bz_i^T\calHsi\bz_i}-\frac{\frac{1}{n}\tr\calHs}{1+\beta\frac{1}{n}\tr\calHs}}^p}\leq\calO(n^{-(1+\delta)}).
\label{showthis2}
\end{align}
Fig. \ref{fig:tree0} gives a schematic representation of the various consolidation steps used to prove \eqref{showthisinProb}.
\begin{figure}
\begin{pspicture}(0,-0.6976563)(10.862812,0.6976563)
\usefont{T1}{ptm}{m}{n}
\rput(0.70234376,0.50921875){\eqref{showthis}}
\usefont{T1}{ptm}{m}{n}
\rput(0.70234376,-0.47078124){\eqref{showthis2}}
\usefont{T1}{ptm}{m}{n}
\rput(3.0823438,0.10921875){\eqref{showthis0}}
\usefont{T1}{ptm}{m}{n}
\rput(7.8523436,0.04921875){\eqref{showthisinProb}}
\psline[linewidth=0.04cm,arrowsize=0.05291667cm 2.0,arrowlength=1.4,arrowinset=0.4]{->}(1.3809375,0.49921876)(2.3009374,0.05921875)
\psline[linewidth=0.04cm,arrowsize=0.05291667cm 2.0,arrowlength=1.4,arrowinset=0.4]{->}(1.3809375,-0.42078125)(2.3009374,0.0)
\usefont{T1}{ptm}{m}{n}
\rput(5.4423437,0.06921875){\eqref{showthis0m}}
\psline[linewidth=0.04cm,arrowsize=0.05291667cm 2.0,arrowlength=1.4,arrowinset=0.4]{->}(6.2809377,0.05921875)(7.1809373,0.03921875)
\psline[linewidth=0.04cm,arrowsize=0.05291667cm 2.0,arrowlength=1.4,arrowinset=0.4]{->}(3.8809376,0.09921875)(4.7809377,0.07921875)
\end{pspicture} 
\centering
\caption{Consolidation steps.}
\label{fig:tree0}
\end{figure}
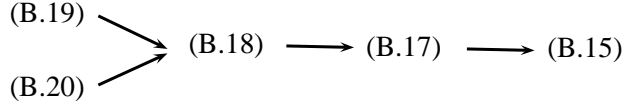

\underline{Proof of \eqref{showthis}:}
First, note that
\begin{align}
\abs{e_i}&\triangleq\abs{\frac{\bz_i^T\calHsi\bz_i-\frac{1}{n}\tr\calHs}{1+\beta\bz_i^T\calHsi\bz_i}}\\
&\stackrel{(a)}{\leq}\abs{\bz_i^T\calHsi\bz_i-\frac{1}{n}\tr\calHs}\\
&\stackrel{(b)}{\leq}\abs{\bz_i^T\calHsi\bz_i-\frac{1}{n}\tr\calHsi}+\abs{\frac{1}{n}\tr\calHsi-\frac{1}{n}\tr\calHs}
\label{twoterms2}
\end{align}
where $(a)$ follows from the fact that $\bz_i^T\calHsi\bz_i$ is non-negative, and $(b)$ follows by adding and subtracting the term $n^{-1}\tr\calHsi$, and then using the triangle inequality. Applying Lemma \ref{aplem:tracedif} (Appendix \ref{sec:mathtoll}) to the second term at the r.h.s. of \eqref{twoterms2}, one readily obtains
\begin{align}
\abs{\frac{1}{n}\tr\calHsi-\frac{1}{n}\tr\calHs}\leq \frac{\sigma^2\norm{\bItt}}{n} = \frac{\sigma^2}{n},
\end{align}
uniformly in $\bs$. Applying Lemma \ref{aplem:trace} (Appendix \ref{sec:mathtoll}) to the first term at the r.h.s. of \eqref{twoterms2},
\begin{align}
&\bE\ppp{\abs{\bz_i^T\calHsi\bz_i-\frac{1}{n}\tr\calHsi}^p}\leq \frac{\tilde{C}}{n^{p/2}},
\label{upptilde}
\end{align}
where according to Lemma \ref{aplem:trace}, $\tilde{C}$ is given by
\begin{align}
\tilde{C} &= C_p\cdot\bE\p{\frac{1}{\abs{\calS}}\tr(\calHsi)^{-1}}^{p/2}\\
&\leq C_p\sigma^{2p},
\end{align}
and where in the last inequality, we have used the fact that $\pp{\bHtt^T\bHtt}_i$ is non-negative, and thus
\begin{align}
\calHsi\preceq \p{\frac{1}{\sigma^2}\bItt}^{-1} = \sigma^2\bItt\label{upperboundinversetr}
\end{align}
where for two matrices $\bA\in\mathbb{R}^{N\times N}$ and $\bB\in\mathbb{R}^{N\times N}$ the notation $\bA\preceq\bB$ means that the difference $\bB-\bA$ is non-negative definite. Thus, the bound in \eqref{upptilde} is uniform in $\bs$. Therefore,
\begin{align}
\bE\ppp{\abs{e_i}^p}\leq\calO(n^{-p/2}).
\end{align}
Consequently, taking any $p>2$, we obtain \eqref{showthis}. 

\underline{Proof of \eqref{showthis2}:} As before,
\begin{align}
\abs{\tilde{e}_i}&\triangleq\abs{\frac{\frac{1}{n}\tr\calHs}{1+\beta\bz_i^T\calHsi\bz_i}-\frac{\frac{1}{n}\tr\calHs}{1+\beta\frac{1}{n}\tr\calHs}}\nonumber\\
& = \frac{\frac{\beta}{n}\tr\calHs\abs{\bz_i^T\calHsi\bz_i-\frac{1}{n}\tr\calHs}}{\p{1+\beta\bz_i^T\calHsi\bz_i}\p{1+\beta\frac{1}{n}\tr\calHs}}\nonumber\\
&\leq \frac{\beta}{n}\tr\calHs\abs{\bz_i^T\calHsi\bz_i-\frac{1}{n}\tr\calHs}\nonumber\\
&\leq \beta\sigma^2\abs{\bz_i^T\calHsi\bz_i-\frac{1}{n}\tr\calHs}
\end{align}
where the last inequality follows from
\begin{align}
\frac{1}{n}\tr\calHs\leq \frac{1}{n}\tr\p{\frac{1}{\sigma^2}\bItt}^{-1}= \sigma^2.\label{B32}
\end{align}
Therefore, as before, by Lemma \ref{aplem:trace}, $\bE\abs{\tilde{e}_i}^p\leq\calO\p{n^{-p/2}}$ as required. 
\end{proof}


\begin{proof}[\underline{Proof of Lemma \ref{lem:subsub2}}]
Let
\begin{align}
\hat{e}\triangleq\frac{\frac{1}{n}\tr\calHs}{1+\beta\frac{1}{n}\tr\calHs}-\frac{m_{\bst}\sigma^2b\p{m_{\bst}}}{1+\beta\sigma^2m_{\bst}b\p{m_{\bst}}}.
\label{showthis0mi2}
\end{align}
By Lemma \ref{aux:assympLLN}, it is sufficient to show that $\bE\abs{\hat{e}}^p\leq\calO\p{{n^{-(1+\delta)}}}$. First, we see that 
\begin{align}
\abs{\hat{e}} &= \frac{\frac{1}{n}\tr\calHs}{1+\beta\frac{1}{n}\tr\calHs}-\frac{m_{\bst}\sigma^2b\p{m_{\bst}}}{1+\beta\sigma^2m_{\bst}b\p{m_{\bst}}}\\
&=\frac{\abs{\frac{1}{n}\tr\calHs-m_{\bst}\sigma^2b\p{m_{\bst}}}}{\p{1+\beta\sigma^2m_{\bst}b\p{m_{\bst}}}\p{1+\beta\frac{1}{n}\tr\calHs}}\\
&\leq \abs{\frac{1}{n}\tr\calHs-m_{\bst}\sigma^2b\p{m_{\bst}}},
\label{A35}
\end{align}
where the last inequality follows from the facts $1+\beta n^{-1}\tr\calHs\geq1$ and $1+\beta\sigma^2m_{\bst}b\p{m_{\bst}}\geq1$. Recall that $b\p{m_{\bst}}$ is the solution of equation \eqref{beqqq1}, i.e.,
\begin{align}
b\p{m_{\bst}} = \p{\frac{\beta\sigma^2R}{1+\beta\sigma^2m_{\bst}b\p{m_{\bst}}}+1}^{-1}.
\label{bdefmid}
\end{align}
Let us define
\begin{align}
w&\triangleq\frac{1}{n}\tr\calHs-\frac{\sigma^2}{n}\tr\pp{\p{\frac{R\beta\sigma^2}{1+\beta\frac{1}{n}\tr\calHs}+1}^{-1}\bItt}.
\label{wdefnif}
\end{align}
Then, note that
\begin{align}
&\calHs - \sigma^2\p{\frac{R\beta\sigma^2}{1+\beta\frac{1}{n}\tr\calHs}+1}^{-1}\bItt\nonumber\\
&\stackrel{(a)}{=}\calHs\pp{\frac{R\beta\sigma^2}{1+\beta\frac{1}{n}\tr\calHs}\bItt+\bItt-\beta\sigma^2\bHtt^T\bHtt-\bItt}\p{\frac{R\beta\sigma^2}{1+\beta\frac{1}{n}\calHs}+1}^{-1}\bItt\\
& = \calHs\pp{\frac{R\beta\sigma^2}{1+\beta\frac{1}{n}\tr\calHs}\bItt-\beta\sigma^2\bHtt^T\bHtt}\p{\frac{R\beta\sigma^2}{1+\beta\frac{1}{n}\tr\calHs}+1}^{-1}\bItt\\
& = -\vartheta\calHs\beta\sigma^2\bHtt^T\bHtt+\vartheta\calHs\frac{R\beta\sigma^2}{1+\beta\frac{1}{n}\tr\calHs}
\label{nexuse}
\end{align}
where  $(a)$ is due to Lemma \ref{aplem:3}, and in the last equalities we canceled out and rearranged the various terms, and defined
\begin{align}
\vartheta&\triangleq\p{\frac{R\beta\sigma^2}{1+\beta\frac{1}{n}\tr\calHs}+1}^{-1}\\
&\stackrel{\eqref{upperboundinversetr}}{\leq} \p{\frac{R\beta\sigma^2}{1+\beta\sigma^2\frac{1}{n}\tr\p{\bItt}^{-1}}+1}^{-1}\\
&\leq \p{\frac{R\beta\sigma^2}{1+\beta\sigma^2}+1}^{-1}\\ 
&= \frac{1+\beta\sigma^2}{1+\beta\sigma^2+R\beta\sigma^2}\triangleq \tilde\vartheta,\label{uppertildevartheta}
\end{align}
namely, $\vartheta$ can be upper bounded by $\tilde\vartheta$, which is independent on $\bs$. Therefore, using \eqref{nexuse},
\begin{align}
w&=\frac{1}{n}\tr\calHs-\frac{\sigma^2}{n}\tr\pp{\p{\frac{R\beta\sigma^2}{1+\beta\frac{1}{n}\tr\calHs}+1}^{-1}\bItt}\nonumber\\
& \stackrel{\eqref{nexuse}}{=} -\vartheta\beta\sigma^2\frac{1}{n}\tr\p{\calHs\bHtt^T\bHtt}+\vartheta\frac{R\sigma^2\beta\frac{1}{n}\tr\calHs}{1+\beta\frac{1}{n}\tr\calHs}\\
& = -\vartheta\frac{1}{n}\sum_{i=1}^k\beta\sigma^2\bz_i^T\calHs\bz_i+\vartheta\frac{1}{n}\sum_{i=1}^k\frac{\beta\sigma^2\frac{1}{n}\tr\calHs}{1+\beta\frac{1}{n}\tr\calHs}
\label{qqqqww}
\end{align}
where in the last equality we have used the fact that $R = k/n$, and that
\begin{align}
\tr\p{\calHs\bHtt^T\bHtt} &= \tr\p{\calHs\sum_{i=1}^k\bz_i\bz_i^T}= \sum_{i=1}^k\bz_i^T\calHs\bz_i.
\end{align}
Therefore,
\begin{align}
\abs{w}&=\abs{\frac{1}{n}\tr\calHs-\frac{\sigma^2}{n}\tr\pp{\p{\frac{R\beta\sigma^2}{1+\beta\frac{1}{n}\tr\calHs}+1}^{-1}\bItt}}\nonumber\\
& \stackrel{\eqref{qqqqww}}{=} \abs{\vartheta\frac{1}{n}\sum_{i=1}^k\pp{\beta\sigma^2\bz_i^T\calHs\bz_i-\frac{\beta\sigma^2\frac{1}{n}\tr\calHs}{1+\beta\frac{1}{n}\tr\calHs}}}\nonumber\\
& \stackrel{(a)}{\leq} \tilde\vartheta\frac{1}{n}\abs{\sum_{i=1}^k\pp{\frac{\beta\sigma^2\bz_i^T\calHsi\bz_i}{1+\beta\bz_i^T\calHsi\bz_i}-\frac{\beta\sigma^2\frac{1}{n}\tr\calHs}{1+\beta\frac{1}{n}\tr\calHs}}},\label{wdef}
\end{align}
where $(a)$ follows by the matrix inversion lemma (Lemma \ref{aplem:1}) and \eqref{uppertildevartheta}. Comparing the upper bound on $w$ in \eqref{wdef} with \eqref{showthisinProb}, we readily conclude that $\bE\abs{w}^p\leq \calO(n^{-p/2})$, and uniformly in $\bs$. Now, note that 
\begin{align}
&\abs{\frac{1}{n}\tr\calHs-m_{\bst}\sigma^2b\p{m_{\bst}}} \nonumber\\
&\stackrel{(a)}{=} \left|\frac{1}{n}\tr\calHs-m_{\bst}\sigma^2\p{\frac{R\beta\sigma^2}{1+\beta\frac{1}{n}\tr\calHs}+1}^{-1}+m_{\bst}\sigma^2\p{\frac{R\beta\sigma^2}{1+\beta\frac{1}{n}\tr\calHs}+1}^{-1}-m_{\bst}\sigma^2b\p{m_{\bst}}\right|\nonumber\\
& \stackrel{(b)}{\leq} \abs{\frac{1}{n}\tr\calHs-m_{\bst}\sigma^2\p{\frac{R\beta\sigma^2}{1+\beta\frac{1}{n}\tr\calHs}+1}^{-1}} + m_{\bst}\sigma^2\abs{\p{\frac{R\beta\sigma^2}{1+\beta\frac{1}{n}\tr\calHs}+1}^{-1}-b\p{m_{\bst}}}\\
&\stackrel{(c)}{=}\abs{w}+m_{\bst}\sigma^2\abs{\p{\frac{R\beta\sigma^2}{1+\beta\frac{1}{n}\tr\calHs}+1}^{-1}-b\p{m_{\bst}}}
\label{ineqproofit1}
\end{align}
where in $(a)$ we added and subtracted a common term, in $(b)$ we used the triangle inequality, and in $(c)$ we noticed that the first term is $w$ given in \eqref{wdefnif}. But using \eqref{bdefmid},
\begin{align}
\abs{\p{\frac{R\beta\sigma^2}{1+\beta\frac{1}{n}\tr\calHs}+1}^{-1}-b\p{m_{\bst}}}&= \abs{\p{\frac{R\beta\sigma^2}{1+\beta\frac{1}{n}\tr\calHs}+1}^{-1}-\p{\frac{\beta\sigma^2R}{1+\beta\sigma^2m_{\bst}b\p{m_{\bst}}}+1}^{-1}}\nonumber\\
& =  \abs{\frac{1+\beta\frac{1}{n}\tr\calHs}{1+\beta\sigma^2R+\beta\frac{1}{n}\tr\calHs}-\frac{1+\beta\sigma^2m_{\bst}b\p{m_{\bst}}}{1+\beta\sigma^2R+\beta\sigma^2m_{\bst}b\p{m_{\bst}}}}\\
&= \frac{\beta^2\sigma^2R\abs{\frac{1}{n}\calHs-m_{\bst}\sigma^2b\p{m_{\bst}}}}{\p{1+\beta\sigma^2R+\beta\frac{1}{n}\tr\calHs}\p{1+\beta\sigma^2R+\beta\sigma^2m_{\bst}b\p{m_{\bst}}}}\\
&\triangleq\kappa\abs{\frac{1}{n}\tr\calHs-m_{\bst}\sigma^2b\p{m_{\bst}}}
\label{ineqproofit2}
\end{align}
where
\begin{align}
\kappa\triangleq \frac{\beta^2\sigma^2R}{\p{1+\beta\sigma^2R+\beta\frac{1}{n}\tr\calHs}\p{1+\beta\sigma^2R+\beta\sigma^2m_{\bst}b\p{m_{\bst}}}}.
\end{align}
Thus, using \eqref{ineqproofit1} and \eqref{ineqproofit2},
\begin{align}
\abs{\frac{1}{n}\tr\calHs-m_{\bst}\sigma^2b\p{m_{\bst}}}\leq\abs{w}+\kappa m_{\bst}\sigma^2\abs{\frac{1}{n}\tr\calHs-m_{\bst}\sigma^2b\p{m_{\bst}}}.
\label{ineqconvv}
\end{align}
In the following, we show that $0<\kappa m_{\bst}\sigma^2<1$. First, for $m_{\bst}\leq R$ we see that
\begin{align}
\kappa m_{\bst}\sigma^2 &= \frac{\beta^2\sigma^4Rm_{\bst}}{\p{1+\beta\sigma^2R+\beta\frac{1}{n}\tr\calHs}\p{1+\beta\sigma^2R+\beta\sigma^2m_{\bst}b\p{m_{\bst}}}}\\
&\stackrel{(a)}{\leq}\frac{\beta^2\sigma^4R^2}{\p{1+\beta\sigma^2R}^2}\leq 1.
\end{align}
where $(a)$ follows from the facts that $\tr\p{\beta\sigma^2\bHtt^T\bHtt+\bItt}^{-1}\geq0$ and that $b\p{m_{\bst}}\geq0$. For $m_{\bst}>R$, we first note that $b\p{m_{\bst}}\geq \p{m_{\bst}-R}/m_{\bst}$, which follows from the facts that $b\p{m_{\bst}}$ is monotonically decreasing in $\beta$ (by definition), and that
\begin{align}
\lim_{\beta\to\infty}b\p{m_{\bst}} = \frac{m_{\bst}-R}{m_{\bst}}.
\end{align} 
Whence,
\begin{align}
\kappa m_{\bst}\sigma^2 &= \frac{\beta^2\sigma^4Rm_{\bst}}{\p{1+\beta\sigma^2R+\beta\frac{1}{n}\tr\calHs}\p{1+\beta\sigma^2R+\beta\sigma^2m_{\bst}b\p{m_{\bst}}}}\\
&\leq\frac{\beta^2\sigma^4Rm_{\bst}}{\p{1+\beta\sigma^2R}\p{1+\beta\sigma^2R+\beta\sigma^2m_{\bst}\frac{m_{\bst}-R}{m_{\bst}}}}\\
&=\frac{\beta^2\sigma^4Rm_{\bst}}{\p{1+\beta\sigma^2R}\p{1+\beta\sigma^2m_{\bst}}}\\
&\leq\frac{\beta^2\sigma^4R}{\p{1+\beta\sigma^2R}\p{1+\beta\sigma^2}}\leq 1.
\end{align}
Thus, using \eqref{ineqconvv}, we get
\begin{align}
\abs{\frac{1}{n}\tr\calHs-m_{\bst}\sigma^2b\p{m_{\bst}}}&\leq\frac{1}{1-m_{\bst}\kappa\sigma^2}\abs{w}\\
&\leq\tilde{\kappa}\abs{w}
\label{A36}
\end{align}
where $\tilde{\kappa}>0$ can be upper bounded by a term that depends solely on $\beta,\sigma^2$ and $R$. Accordingly, based on \eqref{A35}, the fact that $\bE\abs{w}^p\leq \calO(n^{-p/2})$, and \eqref{A36}, we can conclude that
\begin{align}
\bE\abs{\hat{e}}^p&\leq\tilde{\kappa}^p\bE\abs{w}^p\\
&\leq \calO(n^{-p/2}),
\end{align}
which proves \eqref{showthisinProb2}. 
Finally, note that \eqref{A36} and the fact that $\bE\abs{w}^p\leq \calO(n^{-p/2})$, proves also \eqref{itm01}.
\end{proof}
\begin{remark}\label{rem:str}
Note that it is easier to prove the a.s. convergence of the terms in \eqref{uniformconFirst} compared to the above uniform convergence. Indeed, recall \eqref{matinvlem}, and note that the matrix $\calHsi$ is statistically independent on $\bz_i$. Then,
\begin{align}
\frac{1}{n}\sum_{i=1}^k\frac{Y_i^2\bz_i^T\calHsi\bz_i}{1+\beta\bz_i^T\calHsi\bz_i}&\asymp\frac{1}{n}\sum_{i=1}^k\frac{Y_i^2\frac{1}{n}\tr\calHsi}{1+\beta\frac{1}{n}\tr\calHsi}\label{reds01}\\
&\asymp\frac{1}{n}\sum_{i=1}^k\frac{Y_i^2\frac{1}{n}\tr\calHs}{1+\beta\frac{1}{n}\tr\calHs}\label{secpassa}\\
&\asymp\frac{1}{n}\sum_{i=1}^k\frac{Y_i^2m_{\bst}\sigma^2b\p{m_{\bst}}}{1+\beta\sigma^2m_{\bst}b\p{m_{\bst}}}\label{thirdpassa}\\
&= \frac{m_{\bst}\sigma^2b\p{m_{\bst}}}{1+\beta\sigma^2m_{\bst}b\p{m_{\bst}}}\frac{\norm{\bY}^2}{n}
\label{reds1}
\end{align}
where in the first passage, we applied the trace lemma (Lemma \ref{aplem:4}) and Lemma \ref{aplem:5}, in the second passage we have used the rank-1 perturbation lemma (Lemma \ref{aplem:6}), and the third passage is due to Lemma \ref{lem:mat2} (see, Appendix \ref{app2a}). This proves the a.s. convergence of the first term at the r.h.s. of \eqref{twotermsasym}.
\end{remark}

\subsection{Proof of \eqref{itm04}}
%
Let $\bz_i$ and $\bzti_i$ denote the $i$th rows of the matrices $\bHtt$ and $\bHtr$, respectively. Then, using
\begin{align}
\bY^T\bHtt &= \sum_{i=1}^kY_i\bz_i^T,
\end{align}
and
\begin{align}
\bHtr^T\bY &= \sum_{i=1}^kY_i\bzti_i,
\end{align}
we have that
\begin{align}
&\frac{1}{n}\bY^T\bHtt\calHs\bQsr\calHr\bHtr^T\bY= \frac{1}{n}\sum_{i=1}^kY_i^2\bz_i^T\calHs\bQsr\calHr\bzti_i+ \frac{1}{n}\sum_{i\neq j}^kY_iY_j\bz_i^T\calHs\bQsr\calHr\bzti_j.
\label{righrf}
\end{align}
Define
\begin{align}
&q_{n,1} \triangleq \frac{\tilde\alpha m_{\bst,\brt}}{\p{1+\beta\sigma^2m_{\bst}b\p{m_{\bst}}}\p{1+\beta\sigma^2m_{\brt}b\p{m_{\brt}}}}\frac{\norm{\bY}^2}{n},\\
&q_{n,2} \triangleq \frac{\tilde\alpha\pp{\bY^T\bHtt\bQsr\bHtr^T\bY-m_{\bst,\brt}\norm{\bY}^2}}{n\p{1+\beta\sigma^2m_{\bst}b\p{m_{\bst}}}\p{1+\beta\sigma^2m_{\brt}b\p{m_{\brt}}}}\nonumber\\
&\ \ \ \ \ \ \ \ \ \ -\frac{\tilde\alpha m_{\bst,\brt}\beta\sigma^2b\p{m_{\brt}}\pp{\norm{\bY^T\bHtr}^2-m_{\brt}\norm{\bY}^2}}{n\p{1+\beta\sigma^2m_{\bst}b\p{m_{\bst}}}\p{1+\beta\sigma^2m_{\brt}b\p{m_{\brt}}}^2}\nonumber\\
&\ \ \ \ \ \ \ \ \ \ -\frac{\tilde\alpha m_{\bst,\brt}\beta\sigma^2b\p{m_{\bst}}\pp{\norm{\bY^T\bHtt}^2-m_{\bst}\norm{\bY}^2}}{n\p{1+\beta\sigma^2m_{\bst}b\p{m_{\bst}}}^2\p{1+\beta\sigma^2m_{\brt}b\p{m_{\brt}}}},\label{qn2}
\end{align}
and observe that $q_n = q_{n,1}+q_{n,2}$. Thus, to prove \eqref{itm04}, it is sufficient to show that
\begin{align}
\max_{\bst,\brt\in\ppp{0,1}^n}\pr\ppp{\abs{\frac{1}{n}\sum_{i=1}^kY_i^2\bz_i^T\calHs\bQsr\calHr\bzti_i-q_{n,1}}>\epsilon}\leq\frac{1}{\epsilon^p}\calO(n^{-p/2}),\label{uniformconFirst3}
\end{align}
and
\begin{align}
&\max_{\bst,\brt\in\ppp{0,1}^n}\pr\left\{\left|\frac{1}{n}\sum_{i\neq j}^kY_iY_j\bz_i^T\calHs\bQsr\calHr\bzti_j-q_{n,2}\right|>\epsilon\right\}\leq\frac{1}{\epsilon^p}\calO(n^{-p/2}).\label{uniformconFirst4}
\end{align}
Since the same arguments that will be used to prove \eqref{uniformconFirst3} can be used to prove \eqref{uniformconFirst4}, for the sake of brevity, in the following, we focus on \eqref{uniformconFirst3}.

Applying the matrix inversion lemma (Lemma \ref{aplem:1}) we obtain
\begin{align}
&\frac{1}{n}\sum_{i=1}^kY_i^2\bz_i^T\calHs\bQsr\calHr\bzti_i= \frac{1}{n}\sum_{i=1}^k\frac{Y_i^2\bz_i^T\calHsi\bQsr\calHri\bzti_i}{\p{1+\beta\bz_i^T\calHsi\bz_i}\p{1+\beta\bzti_i^T\calHri\bzti_i}}.
\label{righrf1}
\end{align}
Note that contrary to the previous case \eqref{itm03}, where already at this stage, we were able to continue the asymptotic analysis (see, \eqref{reds01}-\eqref{reds1}), in this case we cannot, because currently, we do not know how the numerator behaves. Let
\begin{align}
\eta_{n}\triangleq &\frac{\beta R}{1+\beta \frac{1}{n}\tr\calHr},\label{etadefi}\\
\psi_{n}\triangleq &\frac{\beta R}{1+\beta \frac{1}{n}\tr\calHr}- \frac{\beta^2 R\frac{1}{n}\tr\p{\bQsr\calHr\bQsr^T}}{\p{1+\beta \frac{1}{n}\tr\calHs}\p{1+\beta \frac{1}{n}\tr\calHr}},
\label{psidefi}
\end{align}
and $\tilde\alpha_n\triangleq (\psi_{n}+\sigma^{-2})^{-1}(\eta_{n}+\sigma^{-2})^{-1}$. The following lemma provides the asymptotic behavior of the numerator in \eqref{righrf1}.
\begin{lemma}\label{lem:subsecc}
For any $\epsilon,p>0$,
\begin{align}
\max_{\bst,\brt\in\ppp{0,1}^n}\pr\ppp{\abs{\bz_i^T\calHsi\bQsr\calHri\bzti_i-\tilde\alpha_nm_{\bst,\brt}}>\epsilon}\leq\frac{1}{\epsilon^p}\calO(n^{-p/2}),\label{conbef1}
\end{align}
and
\begin{align}
\max_{\bst,\brt\in\ppp{0,1}^n}\pr\ppp{\abs{\tilde\alpha_nm_{\bst,\brt}-\tilde\alpha(m_{\bst},m_{\brt},m_{\bst,\brt}) m_{\bst,\brt}}>\epsilon}\leq\frac{1}{\epsilon^p}\calO(n^{-p/2}),\label{conbef2}
\end{align}
where $\tilde{\alpha}(m_{\bst},m_{\brt},m_{\bst,\brt})$ is defined in \eqref{tildealphadef}.
\end{lemma}

Given Lemma \ref{lem:subsecc}, by using exactly the same arguments as in \eqref{showthisinProb} and \eqref{showthisinProb2}, it can be shown that for any $p,\epsilon>0$,
\begin{align}
&\max_{\bst,\brt\in\ppp{0,1}^n}\pr\ppp{\abs{\sum_{i=1}^k\frac{Y_i^2}{n}\p{\frac{\bz_i^T\calHsi\bQsr\calHri\bzti_i}{\p{1+\beta\bz_i^T\calHsi\bz_i}\p{1+\beta\bzti_i^T\calHri\bzti_i}}-\frac{\tilde\alpha_n m_{\bst,\brt}}{\p{1+\frac{\beta}{n}\tr\calHs}\p{1+\frac{\beta}{n}\tr\calHr}}}}>\epsilon}\nonumber\\
&\ \ \ \ \ \ \ \ \ \ \ \ \leq\frac{1}{\epsilon^p}\calO(n^{-p/2}),\label{wht1}
\end{align} 
and
\begin{align}
&\max_{\bst,\brt\in\ppp{0,1}^n}\pr\ppp{\abs{\frac{\tilde\alpha_n m_{\bst,\brt}}{\p{1+\beta\frac{1}{n}\tr\calHs}\p{1+\beta\frac{1}{n}\tr\calHr}}\frac{\norm{\bY}^2}{n}-q_{n,1}}>\epsilon}\leq\frac{1}{\epsilon^p}\calO(n^{-p/2}).\label{wht2}
\end{align}
Using \eqref{righrf1}, \eqref{wht1}, and \eqref{wht2}, we obtain \eqref{uniformconFirst3}, as required. We end this subsection by proving Lemma \ref{lem:subsecc}.

\begin{proof}[\underline{Proof of Lemma \ref{lem:subsecc}}]
We start with \eqref{conbef1}. Let $h_n\triangleq \tilde{\alpha}_nm_{\bst,\brt}$. Similarly as in \eqref{upptilde}, using Lemma \ref{aplem:trace}, it can be verified that
\begin{align}
\max_{\bst,\brt\in\ppp{0,1}^n}\pr\ppp{\abs{\bz_i^T\calHsi\bQsr\calHri\bzti_i-\frac{1}{n}\tr\pp{\calHsi\bQsr\calHri\bQsr^T}}>\epsilon}\leq\frac{1}{\epsilon^p}\calO(n^{-p/2}).\label{Whatwish0with}
\end{align}
Accordingly, it is sufficient to show that
\begin{align}
\max_{\bst,\brt\in\ppp{0,1}^n}\pr\ppp{\abs{\frac{1}{n}\tr\pp{\calHsi\bQsr\calHri\bQsr^T}-h_{n}}>\epsilon}\leq\frac{1}{\epsilon^p}\calO(n^{-p/2}).\label{Whatwish}
\end{align}
Note that $h_n$ can be written as follows
\begin{align}
h_{n} \triangleq \frac{1}{n}\tr\p{\bD_s^{-1}\bQsr\bD_r^{-1}\bQsr^T}
\label{h_Nterm}
\end{align}
where $\bD_s\triangleq\p{\psi_{n}+\frac{1}{\sigma^2}}\bItt$ and $\bD_r\triangleq\p{\eta_{n}+\frac{1}{\sigma^2}}\bItr$, and we have used the fact that $\tr\p{\bQsr\bQsr^T} = \sum_{i=1}^ns_ir_i = nm_{\bst,\brt}$. Accordingly,
\begin{align}
\calHsi\bQsr\calHri\bQsr^T-\bD_s^{-1}\bQsr\bD_r^{-1}\bQsr^T&= \calHsi\bQsr\calHri\bQsr^T- \bD_s^{-1}\bQsr\calHri\bQsr^T\nonumber\\
&\ \ \ \ +\bD_s^{-1}\bQsr\calHri\bQsr^T-\bD_s^{-1}\bQsr\bD_r^{-1}\bQsr^T\\
&= \pp{\calHsi-\bD_s^{-1}}\bQsr\calHri\bQsr^T\nonumber\\
&\ \ \ \ +\bD_s^{-1}\bQsr\pp{\calHri-\bD_r^{-1}}\bQsr^T.
\label{B33}
\end{align}
Thus, to prove \eqref{Whatwish}, it is sufficient to show that
\begin{align}
\max_{\bst,\brt\in\ppp{0,1}^n}\pr\ppp{\abs{\frac{1}{n}\tr\p{\pp{\calHsi-\bD_s^{-1}}\bQsr\calHri\bQsr^T}}>\epsilon}\leq\frac{1}{\epsilon^p}\calO(n^{-p/2}),\label{whatwish00}
\end{align}
and 
\begin{align}
\max_{\bst,\brt\in\ppp{0,1}^n}\pr\ppp{\abs{\frac{1}{n}\tr\p{\bD_s^{-1}\bQsr\pp{\calHri-\bD_r^{-1}}\bQsr^T}}>\epsilon}\leq\frac{1}{\epsilon^p}\calO(n^{-p/2}).\label{whatwish01}
\end{align}
Fig. \ref{fig:tree} gives a schematic representation of the various consolidation steps used to prove \eqref{uniformconFirst3}.
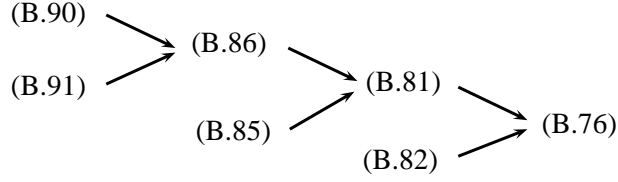
\begin{figure}
\begin{pspicture}(0,-1.1876563)(8.362812,1.1876563)
\usefont{T1}{ptm}{m}{n}
\rput(0.6123437,0.99921876){\eqref{whatwish00}}
\usefont{T1}{ptm}{m}{n}
\rput(0.6123437,0.01921875){\eqref{whatwish01}}
\usefont{T1}{ptm}{m}{n}
\rput(3.0123436,0.5992187){\eqref{Whatwish}}
\usefont{T1}{ptm}{m}{n}
\rput(3.0723438,-0.5807812){\eqref{Whatwish0with}}
\usefont{T1}{ptm}{m}{n}
\rput(5.3323436,0.07921875){\eqref{conbef1}}
\usefont{T1}{ptm}{m}{n}
\rput(5.2923436,-0.9607813){\eqref{conbef2}}
\usefont{T1}{ptm}{m}{n}
\rput(7.6723437,-0.48078126){\eqref{uniformconFirst3}}
\psline[linewidth=0.04cm,arrowsize=0.05291667cm 2.0,arrowlength=1.4,arrowinset=0.4]{->}(1.3809375,0.9892188)(2.3009374,0.5492188)
\psline[linewidth=0.04cm,arrowsize=0.05291667cm 2.0,arrowlength=1.4,arrowinset=0.4]{->}(3.8209374,-0.53078127)(4.6809373,-0.03078125)
\psline[linewidth=0.04cm,arrowsize=0.05291667cm 2.0,arrowlength=1.4,arrowinset=0.4]{->}(1.3809375,0.06921875)(2.3009374,0.48921874)
\psline[linewidth=0.04cm,arrowsize=0.05291667cm 2.0,arrowlength=1.4,arrowinset=0.4]{->}(3.8009374,0.5492188)(4.7209377,0.10921875)
\psline[linewidth=0.04cm,arrowsize=0.05291667cm 2.0,arrowlength=1.4,arrowinset=0.4]{->}(6.0609374,0.02921875)(6.9809375,-0.41078126)
\psline[linewidth=0.04cm,arrowsize=0.05291667cm 2.0,arrowlength=1.4,arrowinset=0.4]{->}(6.0609374,-0.8907812)(6.9809375,-0.53078127)
\end{pspicture} 
\centering
\caption{Consolidation steps.}
\label{fig:tree}
\end{figure}

\underline{Proof of \eqref{whatwish00}:}
Let 
\begin{align}
\hat{z}_i\triangleq \frac{1}{n}\tr\p{\pp{\calHsi-\bD_s^{-1}}\bQsr\calHri\bQsr^T},
\end{align}
and we need to show that $\max_{\bst,\brt}\pr\ppp{\abs{\hat{z}_i}>\epsilon}\leq\calO(n^{-p/2})$. By Lemma \ref{aplem:3},
\begin{align}
\calHsi-\bD_s^{-1} &= \bD_s^{-1}\pp{\bD_s-\beta\pp{\bHtt^T\bHtt}_{i}-\frac{1}{\sigma^2}\bItt}\calHsi\\
& = C\bItt\pp{\psi_n\bItt-\beta\pp{\bHtt^T\bHtt}_{i}}\calHsi\\
& = C\psi_n\calHsi-C\beta\pp{\bHtt^T\bHtt}_{i}\calHsi
\label{twotwoterms}
\end{align} 
where in the second equality we substitute $\bD_s=\p{\psi_{n}+\frac{1}{\sigma^2}}\bItt$, and defined $C \triangleq 1/\p{\psi_n+1/\sigma^2}$. Therefore,
\begin{align}
\hat{z}_i= \frac{C\psi_n}{n}\tr\p{\calHsi\bQsr\calHri\bQsr^T}-\frac{C\beta}{n}\tr\p{\bHtt^T\bHtt\calHsi\bQsr\calHri\bQsr^T}.
\label{midcalca}
\end{align}
Using \eqref{midcalca}, to prove \eqref{whatwish00}, we need to show that
\begin{align}
\max_{\bst,\brt\in\ppp{0,1}^n}\pr\ppp{\abs{\frac{C\psi_{n}}{n}\tr\p{\calHsi\bQsr\calHri\bQsr^T}-\frac{C\beta}{n}\tr\p{\bHtt^T\bHtt\calHsi\bQsr\calHri\bQsr^T}}>\epsilon}\leq\frac{\calO(n^{-p/2})}{\epsilon^p}.\label{Whatwishpsi}
\end{align}
Now,
\begin{align}
\frac{C\beta}{n}\tr\p{\bHtt^T\bHtt\calHsi\bQsr\calHri\bQsr^T}&\stackrel{(a)}{=}\frac{C\beta}{n}\tr\p{\sum_{j=1}^k\bz_j\bz_j^T\calHsi\bQsr\calHri\bQsr^T}\nonumber\\
&\stackrel{(b)}{=}\frac{C\beta}{n}\sum_{j=1}^k\bz_j^T\calHsi\bQsr\calHri\bQsr^T\bz_j\\
&\stackrel{(c)}{=}\frac{C\beta}{n}\sum_{j=1}^k\frac{\bz_j^T\calHsii\bQsr\calHri\bQsr^T\bz_j}{1+\beta\bz_j^T\calHsii\bz_j}\label{apptoHi}
\end{align}
where in $(a)$ we have used the fact that $\bHtt^T\bHtt = \sum_{i=1}^k\bz_i\bz_i^T$, in $(b)$ we have used the cyclic property of the trace operator, and $(c)$ is by the matrix inversion lemma. Applying Lemma \ref{aplem:2} to $\calHri$ in \eqref{apptoHi} we obtain (removing the $\tilde\bz_j\tilde\bz_j^T$ element from $\calHri$)
\begin{align}
&\frac{C\beta}{n}\tr\p{\bHtt^T\bHtt\calHsi\bQsr\calHri\bQsr^T}\nonumber\\
&=\frac{C\beta}{n}\sum_{j=1}^k\frac{\bz_j^T\calHsii\bQsr\calHrii\bQsr^T\bz_j}{1+\beta\bz_j^T\calHsii\bz_j} - \frac{C\beta}{n}\sum_{j=1}^k\frac{\bz_j^T\calHsii\bQsr\calHrii\beta\bzti_j\bzti_j^T\calHrii\bQsr^T\bz_j}{\p{1+\beta\bz_j^T\calHsii\bz_j}\p{1+\beta\bzti_j^T\calHrii\bzti_j}}.
\end{align}
Thus, using the last equality,
\begin{align}
&\hat{z}_i=\frac{C\psi_n}{n}\tr\p{\calHsi\bQsr\calHri\bQsr^T}- \frac{C\beta}{n}\sum_{j=1}^k\frac{\bz_j^T\calHsii\bQsr\calHrii\bQsr^T\bz_j}{1+\beta\bz_j^T\calHsii\bz_j}\nonumber\\
& \ \ \ \ \  +\frac{C\beta}{n}\sum_{j=1}^k\frac{\bz_j^T\calHsii\bQsr\calHrii\beta\bzti_j\bzti_j^T\calHrii\bQsr^T\bz_j}{\p{1+\beta\bz_j^T\calHsii\bz_j}\p{1+\beta\bzti_j^T\calHrii\bzti_j}}.
\label{difftozero}
\end{align}
Substituting $\psi_n$ (see, \eqref{psidefi}) in \eqref{difftozero}, we get
\begin{align}
\hat{z}_i&=\frac{C\beta}{n}\sum_{j=1}^k\pp{\frac{\frac{1}{n}\tr\p{\calHsi\bQsr\calHri\bQsr^T}}{1+\beta \frac{1}{n}\tr\calHr} - \frac{\bz_j^T\calHsii\bQsr\calHrii\bQsr^T\bz_j}{1+\beta\bz_j^T\calHsii\bz_j}}\nonumber\\
&\ \ \ \ \ +\frac{C\beta}{n}\sum_{j=1}^k\left[\frac{\bz_j^T\calHsii\bQsr\calHrii\beta\bzti_j\bzti_j^T\calHrii\bQsr^T\bz_j}{\p{1+\beta\bz_j^T\calHsii\bz_j}\p{1+\beta\bzti_j^T\calHrii\bzti_j}}\right.\nonumber\\
&\left.\ \ \ \ \ \ \ \ \ \ \ \ \ \ \ \ \ \ \ \ \ \ \ \ \ -\frac{\frac{1}{n}\tr\p{\calHsi\bQsr\calHri\bQsr^T}\beta\frac{1}{n}\tr\p{\bQsr\calHr\bQsr^T}}{\p{1+\beta \frac{1}{n}\tr\calHs}\p{1+\beta \frac{1}{n}\tr\calHr}}\right].\label{twosummations}
\end{align}
Whence, to prove \eqref{Whatwishpsi}, it is sufficient to show that the two terms (summations) at r.h.s. of \eqref{twosummations} converge to zero uniformly in $\bs,\br$. The convergence of the first term, can be shown exactly as was already done for \eqref{showthisinProb}. The convergence of the second term is essentially very similar to the first term, but with more terms involved (actually, the second term can be seen as an extension of the first term). Indeed, by Lemma \ref{aux:assympLLN}, it is enough to prove that
\begin{align}
&\bE\ppp{\abs{\frac{\bz_j^T\calHsii\bQsr\calHrii\beta\bzti_j\bzti_j^T\calHrii\bQsr^T\bz_j}{\p{1+\beta\bz_j^T\calHsii\bz_j}\p{1+\beta\bzti_j^T\calHrii\bzti_j}}-\frac{\frac{1}{n}\tr\p{\calHsi\bQsr\calHri\bQsr^T}\beta\frac{1}{n}\tr\p{\bQsr\calHr\bQsr^T}}{\p{1+\beta \frac{1}{n}\tr\calHs}\p{1+\beta \frac{1}{n}\tr\calHr}}}^p}\nonumber\\
&\leq\calO(n^{-(1+\delta)}),\label{showthis31}
\end{align}
or equivalently that (again, we add and subtract a common term and then we use the triangle inequality):
\begin{align}
&\bE\ppp{\abs{\frac{\bz_j^T\calHsii\bQsr\calHrii\beta\bzti_j\bzti_j^T\calHrii\bQsr^T\bz_j-\frac{1}{n}\tr\p{\calHsi\bQsr\calHri\bQsr^T}\beta\frac{1}{n}\tr\p{\bQsr\calHr\bQsr^T}}{\p{1+\beta\bz_j^T\calHsii\bz_j}\p{1+\beta\bzti_j^T\calHrii\bzti_j}}}^p}\nonumber\\
&\leq\calO(n^{-(1+\delta)}),
\label{showthis3}
\end{align}
and that
\begin{align}
&\bE\left\{\left|\frac{\frac{1}{n}\tr\p{\calHsi\bQsr\calHri\bQsr^T}\beta\frac{1}{n}\tr\p{\bQsr\calHr\bQsr^T}}{\p{1+\beta\bz_j^T\calHsii\bz_j}\p{1+\beta\bzti_j^T\calHrii\bzti_j}}\right.\right.\nonumber\\
&\left.\left.\ \ \ \ \ \ \ \ \ \ \ \ \ \ -\frac{\frac{1}{n}\tr\p{\calHsi\bQsr\calHri\bQsr^T}\beta\frac{1}{n}\tr\p{\bQsr\calHr\bQsr^T}}{\p{1+\beta \frac{1}{n}\tr\calHs}\p{1+\beta \frac{1}{n}\tr\calHr}}\right|^p\right\}\leq\calO(n^{-(1+\delta)}).
\label{showthis4}
\end{align}
Fig. \ref{fig:tree2} gives a schematic representation of the various consolidation steps used to prove \eqref{whatwish00}.
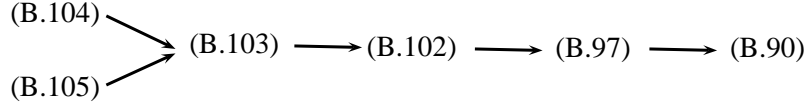
\begin{figure}
\begin{pspicture}(0,-0.6976563)(10.862812,0.6976563)
\usefont{T1}{ptm}{m}{n}
\rput(0.70234376,0.50921875){\eqref{showthis3}}
\usefont{T1}{ptm}{m}{n}
\rput(0.70234376,-0.47078124){\eqref{showthis4}}
\usefont{T1}{ptm}{m}{n}
\rput(3.0823438,0.10921875){\eqref{showthis31}}
\usefont{T1}{ptm}{m}{n}
\rput(10.172344,0.04921875){\eqref{whatwish00}}
\usefont{T1}{ptm}{m}{n}
\rput(7.8523436,0.04921875){\eqref{Whatwishpsi}}
\psline[linewidth=0.04cm,arrowsize=0.05291667cm 2.0,arrowlength=1.4,arrowinset=0.4]{->}(1.3809375,0.49921876)(2.3009374,0.05921875)
\psline[linewidth=0.04cm,arrowsize=0.05291667cm 2.0,arrowlength=1.4,arrowinset=0.4]{->}(1.3809375,-0.42078125)(2.3009374,0.0)
\usefont{T1}{ptm}{m}{n}
\rput(5.4423437,0.06921875){\eqref{twosummations}}
\psline[linewidth=0.04cm,arrowsize=0.05291667cm 2.0,arrowlength=1.4,arrowinset=0.4]{->}(6.2809377,0.05921875)(7.1809373,0.03921875)
\psline[linewidth=0.04cm,arrowsize=0.05291667cm 2.0,arrowlength=1.4,arrowinset=0.4]{->}(8.600938,0.05921875)(9.500937,0.03921875)
\psline[linewidth=0.04cm,arrowsize=0.05291667cm 2.0,arrowlength=1.4,arrowinset=0.4]{->}(3.8809376,0.09921875)(4.7809377,0.07921875)
\end{pspicture} 
\centering
\caption{Consolidation steps.}
\label{fig:tree2}
\end{figure}

Let us show \eqref{showthis3}. First, note that
\begin{align}
\abs{d_i}&\triangleq\abs{\frac{\bz_j^T\calHsii\bQsr\calHrii\beta\bzti_j\bzti_j^T\calHrii\bQsr^T\bz_j-\frac{1}{n}\tr\p{\calHsi\bQsr\calHri\bQsr^T}\beta\frac{1}{n}\tr\p{\bQsr\calHr\bQsr^T}}{\p{1+\beta\bz_j^T\calHsii\bz_j}\p{1+\beta\bzti_j^T\calHrii\bzti_j}}}\\
&\leq\abs{\bz_j^T\calHsii\bQsr\calHrii\beta\bzti_j\bzti_j^T\calHrii\bQsr^T\bz_j-\frac{1}{n}\tr\p{\calHsi\bQsr\calHri\bQsr^T}\beta\frac{1}{n}\tr\p{\bQsr\calHr\bQsr^T}}\\
&\stackrel{(a)}{=}\left|\bz_j^T\calHsii\bQsr\calHrii\beta\bzti_j\bzti_j^T\calHrii\bQsr^T\bz_j-\frac{1}{n}\tr\p{\calHsi\bQsr\calHri\bQsr^T}\beta\bzti_j^T\calHrii\bQsr^T\bz_j\right.\nonumber\\
&\left.+\frac{1}{n}\tr\p{\calHsi\bQsr\calHri\bQsr^T}\beta\bzti_j^T\calHrii\bQsr^T\bz_j-\frac{1}{n}\tr\p{\calHsi\bQsr\calHri\bQsr^T}\beta\frac{1}{n}\tr\p{\bQsr\calHr\bQsr^T}\right|\\
&\stackrel{(b)}{\leq}\abs{\bz_j^T\calHsii\bQsr\calHrii\bzti_j-\frac{1}{n}\tr\p{\calHsi\bQsr\calHri\bQsr^T}}\abs{\beta\bzti_j^T\calHrii\bQsr^T\bz_j}\nonumber\\
&\ \ +\beta\abs{\frac{1}{n}\tr\p{\calHsi\bQsr\calHri\bQsr^T}}\abs{\bzti_j^T\calHrii\bQsr^T\bz_j-\frac{1}{n}\tr\p{\bQsr\calHr\bQsr^T}}
\label{twoterms3}
\end{align}
where $(a)$ follows by adding and subtracting the term
$$
\frac{1}{n}\tr\p{\calHsi\bQsr\calHri\bQsr^T}\beta\bzti_j^T\calHrii\bQsr^T\bz_j,
$$
and $(b)$ follows from the triangle inequality and pulling out the common factor. Using the Cauchy-Schwartz inequality,
\begin{align}
&\bE\ppp{\abs{\bz_j^T\calHsii\bQsr\calHrii\bzti_j-\frac{1}{n}\tr\p{\calHsi\bQsr\calHri\bQsr^T}}^p\abs{\beta\bzti_j^T\calHrii\bQsr^T\bz_j}^p}\nonumber\\
&\leq\p{\bE\abs{\bz_j^T\calHsii\bQsr\calHrii\bzti_j-\frac{1}{n}\tr\p{\calHsi\bQsr\calHri\bQsr^T}}^{2p}}^{1/2}\p{\bE\abs{\beta\bzti_j^T\calHrii\bQsr^T\bz_j}^{2p}}^{1/2}\nonumber\\
&\leq\calO(n^{-p/2})
\end{align}
where the last inequality follows from Lemma \ref{aplem:trace} and the fact that $\bE\abs{\beta\bzti_j^T\calHrii\bQsr^T\bz_j}^{2p}$ is bounded (Lemma \ref{aplem:finite}). Let $\ba_j$ be the $j$th row of $\bH$, and let $\mathbf{\Pi}_r$ and $\mathbf{\Pi}_s$ be $n\times\abs{\calR}$ and $n\times\abs{\calS}$ binary \emph{projection} matrices such that $\bzti_j = \mathbf{\Pi}_r^T\ba_j$ and $\bz_j = \mathbf{\Pi}_s^T\ba_j$, respectively (note that $\bQsr = \mathbf{\Pi}_s\mathbf{\Pi}_r$), i.e., each of the $\abs{\calR}$ columns of $\mathbf{\Pi}_r$ has a single unit entry corresponding to an index from $\calR$ (and the same for $\mathbf{\Pi}_s$). Then, note that
\begin{align}
\bzti_j^T\calHrii\bQsr^T\bz_j &= \ba_j^T\mathbf{\Pi}_r\calHrii\bQsr^T\mathbf{\Pi}_s^T\ba_j\\
& = (\sqrt{n}\ba_j)^T\pp{\frac{1}{n}\mathbf{\Pi}_r\calHrii\bQsr^T\mathbf{\Pi}_s^T}(\sqrt{n}\ba_j).
\end{align}
Also, recall that the matrix $\calHrii$ is defined as
\begin{align}
&\calHrii = \pp{\beta\p{\bHtt^T\bHtt -\bzti_i\bzti_i^T-\bzti_j\bzti_j^T}+\frac{1}{\sigma^2}\bItr}^{-1},
\end{align}
and thus $\mathbf{\Pi}_r\calHrii\bQsr^T\mathbf{\Pi}_s^T$ is independent of $\ba_j$. Finally, note that w.p. 1, $\calHrii\preceq\sigma^2\bItr$, and thus, with the same probability,
\begin{align}
\tr\pp{\frac{1}{n}\p{\mathbf{\Pi}_r\calHrii\bQsr^T\mathbf{\Pi}_s^T}^T\p{\mathbf{\Pi}_r\calHrii\bQsr^T\mathbf{\Pi}_s^T}}&\leq \tr\p{\frac{\sigma^4}{n}\mathbf{\Pi}_s\bQsr\mathbf{\Pi}_r^T\mathbf{\Pi}_r\bQsr^T\mathbf{\Pi}_s^T}\\
&\leq\sigma^4,
\end{align}
which establishes the boundness condition in Lemma \ref{aplem:finite}. The second term in \eqref{twoterms3} is handled similarly. Thus, taking any $p>2$, we obtain \eqref{showthis3}. Similar arguments can be applied to show that \eqref{showthis4} holds true. This establishes the proof of \eqref{whatwish00}.

\underline{Proof of \eqref{whatwish01}:}
Using Lemma \ref{aplem:3},
\begin{align}
\calHri-\bD_r^{-1} = \bD_r^{-1}\pp{\bD_r-\beta\pp{\bHtr^T\bHtr}_{i}-\frac{1}{\sigma^2}\bItr}\calHri.
\end{align}
Substituting $\bD_r$, we obtain
\begin{align}
&\bD_s^{-1}\bQsr\pp{\calHri-\bD_r^{-1}} = \eta_{n}\bD_s^{-1}\bQsr\bD_r^{-1}\calHri\nonumber\\
&\ \ \ \ \ \ \ \ \ \ \ \ \ \ \ \ \ \ \ \ \ \ \ -\bD_s^{-1}\bQsr\bD_r^{-1}\beta\pp{\bHtr^T\bHtr}_{i}\calHri.
\end{align}
Let $\tilde{C} = 1/\p{\eta_{n}+1/\sigma^2}$. Then,
\begin{align}
\frac{1}{n}\tr\p{\bD_s^{-1}\bQsr\bD_r^{-1}\beta\pp{\bHtr^T\bHtr}_{i}\calHri\bQsr^T} = \frac{\tilde{C}C\beta}{n}\sum_{j=1}^k\frac{\bzti_j^T\calHrii\bQsr^T\bQsr\bzti_j}{1+\beta\bzti_j^T\calHrii\bzti_j}
\end{align}
where, as before, we have used $\bHtt^T\bHtt = \sum_{i=1}^k\bz_i\bz_i^T$, the cyclic property of the trace operator, and the matrix inversion lemma. Substituting \eqref{etadefi} in \eqref{whatwish01} we get
\begin{align}
\frac{1}{n}\tr\p{\bD_s^{-1}\bQsr\pp{\calHri-\bD_r^{-1}}\bQsr^T} = \frac{\tilde{C}C\beta}{n}\sum_{j=1}^k\pp{\frac{\bzti_j^T\calHrii\bQsr^T\bQsr\bzti_j}{1+\beta\bzti_j^T\calHrii\bzti_j}-\frac{\frac{1}{n}\tr\p{\bQsr\calH^r_i\bQsr^T}}{1+\beta \frac{1}{n}\tr\calH^r}}.\label{conasbef}
\end{align}
Noting to the similarities between \eqref{conasbef} and \eqref{showthisinProb}, using the same arguments used to prove \eqref{showthisinProb}, it can be shown that \eqref{conasbef} converges to zero uniformly in $\bs$ and $\br$, yielding \eqref{whatwish01}. 

It remains to show \eqref{wht2}. The main observation is that $\tilde\alpha_n$ involves the Stieltjes transforms $\frac{1}{n}\tr\calHr$, $\frac{1}{n}\tr\calHs$ and $\frac{1}{n}\tr\bQsr\calHr\bQsr^T$. Thus, repeating the same steps as in \eqref{showthis0mi2}-\eqref{A36}, we can readily prove \eqref{wht2}. Note that similarly as in Remark \ref{rem:str}, using Lemma \ref{lem:mat2}, it is easy to see that
\begin{align}
&\eta_{n}=\frac{\beta R}{1+\beta \frac{1}{n}\tr\calHr}\asymp \frac{\beta R}{1+\beta\sigma^2 m_{\brt}b(m_{\brt})}\triangleq\eta_{\infty},
\end{align}
and
\begin{align}
\psi_{n}&= \frac{\beta R}{1+\beta \frac{1}{n}\tr\calHr}- \frac{\beta^2 R\frac{1}{n}\tr\p{\bQsr\calHr\bQsr^T}}{\p{1+\beta \frac{1}{n}\tr\calHs}\p{1+\beta \frac{1}{n}\tr\calHr}},\\
&\asymp \frac{\beta R}{1+\beta\sigma^2m_{\brt}b\p{m_{\brt}}}- \frac{\beta^2\sigma^2Rb\p{m_{\bst}}m_{\bst,\brt}}{\p{1+\beta\sigma^2m_{\bst}b\p{m_{\bst}}}\p{1+\beta\sigma^2m_{\brt}b\p{m_{\brt}}}}\triangleq\psi_{\infty}.
\end{align}
Whence,
\begin{align}
\tilde\alpha_n&= (\psi_{n}+\sigma^{-2})^{-1}(\eta_{n}+\sigma^{-2})^{-1}\\
&\asymp (\psi_{\infty}+\sigma^{-2})^{-1}(\eta_{\infty}+\sigma^{-2})^{-1} = \tilde\alpha,
\end{align}
which establishes the a.s. convergence of $\tilde\alpha_n$ to $\tilde\alpha$.\end{proof}

\subsection{Proof of \eqref{itm02}}
Recall the definition of $S(-1)$ in \eqref{bTerm1}, and note that $\eta\p{\gamma}$ in Lemma \ref{lem:mat1} boils down to
\begin{align}
\eta\p{\gamma} &\triangleq \frac{1}{k}\sum_{l=1}^{\abs{\calS}}\log\p{1+cg_lS\p{-\gamma}}-\log\p{\gamma^2S\p{-\gamma}}-\frac{1}{\abs{\calS}}\sum_{l=1}^{\abs{\calS}}\frac{g_lS\p{-\gamma}}{1+cg_lS\p{-\gamma}}.
\end{align}
Thus, for $\gamma=1$,
\begin{align}
\eta\p{1} = \frac{R}{m_{\bst}}\log\pp{1+\beta\sigma^2b\p{m_{\bst}}m_{\bst}}-\log b\p{m_{\bst}}-\frac{\beta\sigma^2Rb\p{m_{\bst}}}{1+\beta\sigma^2b\p{m_{\bst}}m_{\bst}},
\end{align}
which is $\bar{I}\p{m_{\bst}}$ defined in \eqref{firs2tt}. Thus, by Lemma \ref{lem:mat1},
\begin{align}
\frac{1}{n}\log\det\p{\beta\sigma^2\bHtt^T\bHtt+\bItt}-m_{\bst}\bar{I}\p{m_{\bst}}\to 0\label{asConlogdet}
\end{align} 
a.s. as $n\to\infty$. From \eqref{asConlogdet} one cannot deduce \eqref{itm02}. Nonetheless, the uniformity w.r.t. $\bs$ follows from the original proof of Lemma \ref{lem:mat1} in \cite{Couillet} and \cite[Appendix B]{ChaoKaiWen}. In short, the uniformity is due to the following facts: First, the Shannon transform of any non-negative definite matrix can be expressed as a functional of the Stieltjes transform of the same matrix \cite[Eq. (3.5)]{coulbook}. Second, in \cite[Appendix B, eq. (47)]{Couillet}, it was shown that the same functional relation holds also between their respective deterministic equivalents, $\bar{S}_s(z)$ and $\eta_s(z)$ (see the notation in Lemma \ref{lem:mat1}). Finally, using the fact that the convergence of the Stieltjes transform of $\beta\bHtt^T\bHtt$ to $\bar{S}_s(z)$ is uniform w.r.t. $\bs$, it can be shown that this is the case also for the Shannon transform of $\beta\bHtt^T\bHtt$ and $\eta_s(z)$.

\section{}\label{app:3}
\begin{proof}[Derivation of \eqref{DasymMMSE}]
In this appendix, using the previous asymptotic results, we derive the asymptotic MMSE. As was shown in Subsection \ref{sec:mainsteps}, our objective is to evaluate \eqref{lastee2} which is given by
\begin{align}
\limsup_{n\to\infty}\frac{\text{mmse}\p{\bX\vert\bY,\bH}}{n} &= \sigma^2m_a-\limsup_{n\to\infty}\bE\ppp{\bE_{\mu_{s\times r}}\pp{J\p{\bY,\bHtt,\bHtr}\Ind_{\calT_\epsilon^{s,r}}}}.
\label{lastee3}
\end{align}
Note that
\begin{align}
\by^T\bHtt\bQsr\bHtr^T\by &= \sum_{i=1}^n\abs{\bh_i^T\by}^2s_ir_i,\\
\norm{\bHtt^T\by}^2 &= \sum_{i=1}^n\abs{\bh_i^T\by}^2s_i,
\end{align}
and
\begin{align}
\norm{\bHtr^T\by}^2 &= \sum_{i=1}^n\abs{\bh_i^T\by}^2r_i.
\end{align}
Over $\calT_\epsilon^{\bst,\brt}$, using the definitions of $\xi\p{\by,\bHtt}$ and $J\p{\by,\bHtt,\bHtr}$ in \eqref{JJJ1} and \eqref{JJJ3}, respectively,
\begin{align}
\abs{\frac{1}{n}\log\xi\p{\by,\bHtt}-\frac{\beta^2}{2}f_n-\frac{1}{2}m_{\bst}\bar{I}\p{m_{\bst}}}<\epsilon,
\label{typtyp}
\end{align}
and
\begin{align}
&\abs{J\p{\by,\bHtt,\bHtr}-\beta^2q_n}<\epsilon.\label{appeox12}
\end{align}
For brevity, we let
\begin{align}
J^{\epsilon}\p{\by,\bHtt,\bHtr}&\triangleq \beta^2q_n+\epsilon,\\
\xi^{\epsilon}\p{\by,\bHtt}&\triangleq\exp\ppp{n\p{\frac{\beta^2}{2}f_n+\frac{1}{2}m_{\bst}\bar{I}\p{m_{\bst}}+\epsilon}},\\
\mu_{s\times r}^{\epsilon}\p{\by,\bHtt,\bHtr}&\triangleq\frac{P_{\bSt}(\bs)P_{\bSt}(\br)\xi^{\epsilon}\p{\by,\bHtt}\xi^{\epsilon}\p{\by,\bHtr}}{\pp{\sum_{\but\in\ppp{0,1}^n}P_{\bSt}\p{\bu}\xi^\epsilon\p{\by,\bHtu}}^2}.
\end{align}
Thus,
\begin{align}
\bE_{\mu_{s\times r}}\pp{J\p{\bY,\bHtt,\bHtr}\Ind_{\calT_\epsilon^{s,r}}}&\leq \bE_{\mu^\epsilon_{s\times r}}\pp{J^\epsilon\p{\bY,\bHtt,\bHtr}\Ind_{\calT_\epsilon^{s,r}}}\\
&  = \bE_{\mu^\epsilon_{s\times r}}\pp{J^\epsilon\p{\bY,\bHtt,\bHtr}}- \bE_{\mu^\epsilon_{s\times r}}\pp{J^\epsilon\p{\bY,\bHtt,\bHtr}\Ind_{(\calT_\epsilon^{s,r})^c}},\label{lastexpecApp}
\end{align}
and on the other hand,
\begin{align}
\bE_{\mu_{s\times r}}\pp{J\p{\bY,\bHtt,\bHtr}\Ind_{\calT_\epsilon^{s,r}}}&\geq \bE_{\mu^{-\epsilon}_{s\times r}}\pp{J^{-\epsilon}\p{\bY,\bHtt,\bHtr}\Ind_{\calT_\epsilon^{s,r}}}\\
&  = \bE_{\mu^{-\epsilon}_{s\times r}}\pp{J^{-\epsilon}\p{\bY,\bHtt,\bHtr}}- \bE_{\mu^{-\epsilon}_{s\times r}}\pp{J^{-\epsilon}\p{\bY,\bHtt,\bHtr}\Ind_{(\calT_\epsilon^{s,r})^c}}.\label{lastexpecApp02}
\end{align}
Now, similarly as in \eqref{negdomin}, the last terms of \eqref{lastexpecApp} and \eqref{lastexpecApp02} tend to zero as $n,k\to\infty$. Thus,
\begin{align}
\bE\ppp{\bE_{\mu^{-\epsilon}_{s\times r}}\pp{J^{-\epsilon}\p{\bY,\bHtt,\bHtr}}}-o(1)&\leq\bE\ppp{\bE_{\mu_{s\times r}}\pp{J\p{\bY,\bHtt,\bHtr}\Ind_{\calT_\epsilon^{s,r}}}}\nonumber\\
&\ \ \ \ \ \ \ \ \ \ \ \ \leq\bE\ppp{\bE_{\mu^\epsilon_{s\times r}}\pp{J^\epsilon\p{\bY,\bHtt,\bHtr}}}+o(1).\label{lastexpecApp2}
\end{align} 
Our next objective is to analyze the asymptotic behavior of the terms at the l.h.s. and the r.h.s. of \eqref{lastexpecApp2}. Let
$$
\mathscr{Z}\p{\by,\bH}\triangleq\sum_{\bs\in\ppp{0,1}^n}\sum_{\br\in\ppp{0,1}^n}P_{\bSt}(\bs)P_{\bSt}(\br)J\p{\by,\bHtt,\bHtr}\xi\p{\by,\bHtt}\xi\p{\by,\bHtr}.
$$
We denote
\begin{align}
\frac{\beta^2}{2}f_n &= \frac{\beta^3\sigma^4b^2\p{m_{\bst}}m_{\bst}^2}{2g^2\p{m_{\bst}}}\frac{\norm{\by}^2}{n}+\frac{\beta^2\sigma^2b\p{m_{\bst}}}{2g^2\p{m_{\bst}}}\frac{\norm{\bHtt^T\by}^2}{n}\nonumber\\
&\triangleq V\p{m_{\bst}}\frac{\norm{\by}^2}{n}+L\p{m_{\bst}}\frac{\sum_{i=1}^n\abs{\by^T\bh_i}^2s_i}{n},
\end{align}
and $q_\epsilon(\bs,\br)\triangleq \beta^2q_n+\epsilon$.
Using \eqref{lastexpecApp2}, for large $n$ and $k$, the function $\mathscr{Z}\p{\by,\bH}$ is lower and upper bounded as follows
\begin{align}
\mathscr{Z}_{-\epsilon}\p{\by,\bH}\leq \mathscr{Z}\p{\by,\bH}\leq \mathscr{Z}_{\epsilon}\p{\by,\bH}
\end{align}
where
\begin{align}
&\mathscr{Z}_{\epsilon}\p{\by,\bH}\triangleq\sum_{\bs\in\ppp{0,1}^n}\sum_{\br\in\ppp{0,1}^n}q_{\epsilon}\p{\bs,\br}\exp\left\{n\left(\tilde{t}\p{m_{\bst}}+\tilde{t}\p{m_{\brt}}+L\p{m_{\bst}}\frac{1}{n}\sum_{i=1}^n\abs{\by^T\bh_i}^2s_i\right.\right.\nonumber\\
&\ \ \ \ \ \ \ \ \ \ \ \ \ \ \ \ \ \ \ \ \ \ \ \ \ \ \ \ \ \ \ \ \ \ \ \ \ \ \ \ \ \ \ \ \ \ \ \ \ \ \ \ \left.\left.+L\p{m_{\brt}}\frac{1}{n}\sum_{i=1}^n\abs{\by^T\bh_i}^2r_i+\epsilon\right)\right\}
\label{partplusminus}
\end{align}
in which
\begin{align}
\tilde{t}\p{m} \triangleq f\p{m}-\frac{m}{2}\bar{I}\p{m}+V\p{m}\frac{\norm{\by}^2}{n}.
\end{align}
Based on \eqref{partplusminus}, we need to handle a double summation (over $\bs$ and $\br$). We first assess the exponential order of the sum over $\br$. First, we rewrite $\mathscr{Z}_{\epsilon}\p{\by,\bH}$ as follows
\begin{align}
\mathscr{Z}_{\epsilon}\p{\by,\bH}&=\sum_{\bs\in\ppp{0,1}^n}\exp\ppp{n\p{\tilde{t}\p{m_{\bst}}+L\p{m_{\bst}}\frac{1}{n}\sum_{i=1}^n\abs{\by^T\bh_i}^2s_i}}\nonumber\\
&\ \ \ \ \ \ \ \ \ \sum_{\br\in\ppp{0,1}^n}q_{\epsilon}\p{\bs,\br}\exp\left\{n\left(\tilde{t}\p{m_{\brt}}+L\p{m_{\brt}}\frac{1}{n}\sum_{i=1}^n\abs{\by^T\bh_i}^2r_i+\epsilon\right)\right\}\label{B10}\\
& \triangleq \sum_{\bs\in\ppp{0,1}^n}\exp\ppp{n\p{\tilde{t}\p{m_{\bst}}+L\p{m_{\bst}}\frac{1}{n}\sum_{i=1}^n\abs{\by^T\bh_i}^2s_i}}\tilde{\mathscr{Z}}_{\epsilon}\p{\by,\bH,\bs}
\label{partplusminus2}
\end{align}
where
\begin{align}
\tilde{\mathscr{Z}}_{\epsilon}\p{\by,\bH,\bs} \triangleq \sum_{\br\in\ppp{0,1}^n}q_{\epsilon}\p{\bs,\br}\exp\ppp{n\p{\tilde{t}\p{m_{\brt}}+L\p{m_{\brt}}\frac{1}{n}\sum_{i=1}^n\abs{\by^T\bh_i}^2r_i+\epsilon}}.
\end{align}
Now, $\tilde{\mathscr{Z}}_{\epsilon}\p{\by,\bH,\bs}$ can be equivalently rewritten as
\begin{align}
\tilde{\mathscr{Z}}_{\epsilon}\p{\by,\bH,\bs}  = \sum_{m_r}\exp\ppp{n\p{\tilde{t}\p{m_r}+\epsilon}}\hat{\mathscr{Z}}_\epsilon\p{\by,\bH,\bs,m_r}
\label{partplusemin}
\end{align}
where the summation is over $m_r\in\ppp{0/n,1/n,\ldots,n/n}$, and
\begin{align}
\hat{\mathscr{Z}}_\epsilon\p{\by,\bH,\bs,m_r}\triangleq\sum_{\brt:\;m_{\brt} = m_r}q_{\epsilon}\p{\bs,\br}\exp\p{L\p{m_r}\sum_{i=1}^n\abs{\by^T\bh_i}^2r_i}\label{Zparthat}
\end{align}
where with slight abuse of notation, the summation is performed over sequences $\ppp{\br}$ with magnetization, $m_{\brt}=n^{-1}\sum_{i=1}^nr_i$, fixed to $m_r$. For conciseness we omit the dependency of the above terms on $\epsilon$. 
 
We next assess the asymptotic behavior of $\hat{\mathscr{Z}}\p{\by,\bH,\bs,m_r}$, and then the asymptotic behavior of $\tilde{\mathscr{Z}}\p{\by,\bH,\bs}$. For $\hat{\mathscr{Z}}\p{\by,\bH,\bs,m_r}$, we need to count the number of binary sequences $\ppp{\br}$, having a given magnetization $m_r$, and are subject to some linear constraints (finite number of them). Accordingly, consider the following set
\begin{align}
\mathcal{F}_\delta\p{\ppp{\rho_l}_{l=1}^L,m}\triangleq\ppp{\bv\in\ppp{0,1}^n:\;\abs{\sum_{i=1}^nv_i-nm}\leq\delta, \;\abs{\sum_{i=1}^nv_iu_{i,l}-n\rho_l}\leq\delta,\;l=1,\ldots,L}
\label{fcals}
\end{align}
where $L\in\mathbb{N}$ is fixed, and $\ppp{u_{i,l}}_{i=1}^n$ for $l=1,\ldots,L$, are given sequences of real numbers. We will upper and lower bound the cardinality of $\mathcal{F}_\delta\p{\ppp{\rho_l}_{l=1}^L,m}$ for a given $\delta>0$, $m$, and $\ppp{\rho_l}_{l=1}^L$. Then, we will use the result in order to approximate $\hat{\mathscr{Z}}\p{\by,\bH,\bs,m_r}$. 
\begin{lemma}
The cardinality of $\mathcal{F}_\delta\p{\ppp{\rho_l}_{l=1}^L,m}$ satisfies, for any $\tau>0$,
\begin{align}
(1-\tau)R_{-\delta}\leq\abs{\mathcal{F}_\delta\p{\ppp{\rho_l}_{l=1}^L,m}}\leq R_{\delta}
\end{align}
where 
\begin{align}
R_{\delta}&\triangleq\exp\left\{\frac{1}{2}\p{\sum_{l=1}^L\alpha_l^\circ \sum_{i=1}^nu_{i,l}-n\gamma^\circ}-\p{\sum_{l=1}^L\alpha_l^\circ\p{n\rho_l-\delta}-\gamma^\circ\p{nm-\delta}}\right.\nonumber\\
&\left.\ \ \ \ \ \ \ \ \ \ +\sum_{i=1}^n\log\pp{2\cosh\p{\frac{\sum_{l=1}^L\alpha_l^\circ u_{i,l}-\gamma^\circ}{2}}}\right\},
\end{align}
and $\ppp{\ppp{\alpha^\circ_l}_{l=1}^L,\gamma^\circ}$ are given by the solution of the following set of equations
\begin{align}
\rho_l = \frac{\delta}{n}+\frac{1}{2n}\sum_{i=1}^nu_{i,l}+\frac{1}{2n}\sum_{i=1}^n\tanh\p{\frac{\sum_{l=1}^L\alpha_l^\circ u_{i,l}-\gamma^\circ}{2}}u_{i,l},\ \ l=1,\ldots,L,
\end{align}
and
\begin{align}
m = \frac{\delta}{n}+\frac{1}{2}+\frac{1}{2n}\sum_{i=1}^n\tanh\p{\frac{\sum_{l=1}^L\alpha_l u_{i,l}-\gamma^\circ}{2}}.
\end{align}
\end{lemma}
\begin{proof}
Define
\begin{align}
P\p{v_i;\ppp{\alpha_l}_{l=1}^L,\gamma\vert \ppp{u_{i,l}}_{l=1}^L} \triangleq \frac{\exp\ppp{\sum_{l=1}^L\alpha_l v_iu_{i,l}-\gamma v_i}}{2\exp\ppp{\frac{1}{2}\p{\sum_{l=1}^L\alpha_l u_{i,l}-\gamma}}\cosh\p{\frac{\sum_{l=1}^L\alpha_l u_{i,l}-\gamma}{2}}}
\end{align}
where $\ppp{\alpha_l}_{l=1}^L$ and $\gamma$ are auxiliary parameters. Now, for $\bv = \p{v_1,\ldots,v_n}$, let
\begin{align}
P\p{\bv;\ppp{\alpha_l}_{l=1}^L,\gamma\vert \ppp{\bu_l}_{l=1}^L} \triangleq \frac{\exp\ppp{\sum_{l=1}^L\alpha_l\sum_{i=1}^nv_iu_{i,l}-\gamma\sum_{i=1}^nv_i}}{2^n\exp\ppp{\frac{1}{2}\p{\sum_{l=1}^L\alpha_l \sum_{i=1}^nu_{i,l}-n\gamma}}\prod_{i=1}^n\cosh\p{\frac{\sum_{l=1}^L\alpha_l u_{i,l}-\gamma}{2}}}.
\label{condmeas}
\end{align}
Then, we have that 
\begin{align}
1 &\geq\pr\p{\bv\in\mathcal{F}_\delta\p{\rho,m};\ppp{\alpha_l}_{l=1}^L,\gamma\vert \ppp{\bu_l}_{l=1}^L}\\
& = \sum_{\bv\in\mathcal{F}_\delta}\frac{\exp\ppp{\sum_{l=1}^L\alpha_l\sum_{i=1}^nv_iu_{i,l}-\gamma\sum_{i=1}^nv_i}}{2^n\exp\ppp{\frac{1}{2}\p{\sum_{l=1}^L\alpha_l \sum_{i=1}^nu_{i,l}-n\gamma}}\prod_{i=1}^n\cosh\p{\frac{\sum_{l=1}^L\alpha_l u_{i,l}-\gamma}{2}}}\\
& \geq \sum_{\bv\in\mathcal{F}_\delta}\frac{\exp\ppp{\sum_{l=1}^L\alpha_l\p{n\rho_l-\delta}-\gamma\p{nm-\delta}}}{2^n\exp\ppp{\frac{1}{2}\p{\sum_{l=1}^L\alpha_l \sum_{i=1}^nu_{i,l}-n\gamma}}\prod_{i=1}^n\cosh\p{\frac{\sum_{l=1}^L\alpha_l u_{i,l}-\gamma}{2}}}\\
& = \abs{\mathcal{F}_\delta\ppp{\p{\rho_l}_{l=1}^L,m}}\frac{\exp\ppp{\sum_{l=1}^L\alpha_l\p{n\rho_l-\delta}-\gamma\p{nm-\delta}}}{2^n\exp\ppp{\frac{1}{2}\p{\sum_{l=1}^L\alpha_l \sum_{i=1}^nu_{i,l}-n\gamma}}\prod_{i=1}^n\cosh\p{\frac{\sum_{l=1}^L\alpha_l u_{i,l}-\gamma}{2}}}.
\label{ap2}
\end{align}
It is easy to verify that $\ppp{\ppp{\alpha^\circ_l}_{l=1}^L,\gamma^\circ}$ given by the solution of the following set of equations
\begin{align}
\rho_l = \frac{\delta}{n}+\frac{1}{2n}\sum_{i=1}^nu_{i,l}+\frac{1}{2n}\sum_{i=1}^n\tanh\p{\frac{\sum_{l=1}^L\alpha_l^\circ u_{i,l}-\gamma^\circ}{2}}u_{i,l},\ \ l=1,\ldots,L,
\label{saddleeq1}
\end{align}
and
\begin{align}
m = \frac{\delta}{n}+\frac{1}{2}+\frac{1}{2n}\sum_{i=1}^n\tanh\p{\frac{\sum_{l=1}^L\alpha_l u_{i,l}-\gamma^\circ}{2}},
\label{saddleeq2}
\end{align}
maximize the right hand side of \eqref{ap2} (w.r.t. $\ppp{\alpha}_{l=1}^L$ and $\gamma$). Thus, using the last results, we have the following upper bound
\begin{align}
&\abs{\mathcal{F}_\delta\p{\ppp{\rho_l}_{l=1}^L,m}}\leq\frac{\exp\ppp{\frac{1}{2}\p{\sum_{l=1}^L\alpha_l^\circ \sum_{i=1}^nu_{i,l}-n\gamma^\circ}}\prod_{i=1}^n2\cosh\p{\frac{\sum_{l=1}^L\alpha_l^\circ u_{i,l}-\gamma^\circ}{2}}}{\exp\ppp{\sum_{l=1}^L\alpha_l^\circ\p{n\rho_l-\delta}-\gamma^\circ\p{nm-\delta}}}\nonumber\\
&=\exp\left\{\frac{1}{2}\p{\sum_{l=1}^L\alpha_l^\circ \sum_{i=1}^nu_{i,l}-n\gamma^\circ}-\p{\sum_{l=1}^L\alpha_l^\circ\p{n\rho_l-\delta}-\gamma^\circ\p{nm-\delta}}\right.\nonumber\\
&\left.\ \ \ \ \ \ \ \ \ \ +\sum_{i=1}^n\log\pp{2\cosh\p{\frac{\sum_{l=1}^L\alpha_l^\circ u_{i,l}-\gamma^\circ}{2}}}\right\}\\
& \triangleq R_{\delta}.
\label{eeqq1}
\end{align}
For a lower bound, we first note that
\begin{align}
1 &= \pr\p{\bv\in\mathcal{F}_\delta\p{\ppp{\rho_l}_{l=1}^L,m};\ppp{\alpha_l}_{l=1}^L,\gamma\vert \ppp{\bu_l}_{l=1}^L}\nonumber\\
&\ \ \ +\pr\p{\bv\in\mathcal{F}^c_\delta\p{\ppp{\rho_l}_{l=1}^L,m};\ppp{\alpha_l}_{l=1}^L,\gamma\vert \ppp{\bu_l}_{l=1}^L}\\
&\leq\abs{\mathcal{F}_\delta\p{\ppp{\rho_l}_{l=1}^L,m}}\frac{1}{R_{-\delta}}+ \pr\p{\bv\in\mathcal{F}^c_\delta\p{\ppp{\rho_l}_{l=1}^L,m};\ppp{\alpha_l}_{l=1}^L,\gamma\vert \ppp{\bu_l}_{l=1}^L}
\end{align}
where the last inequality follows by the same considerations we have used for obtaining \eqref{ap2} (but now with $\delta$ instead of $-\delta$). Using Boole's inequality,
\begin{align}
&\pr\p{\bv\in\mathcal{F}^c_\delta\p{\ppp{\rho_l}_{l=1}^L,m};\ppp{\alpha_l}_{l=1}^L,\gamma\vert \ppp{\bu_l}_{l=1}^L}\leq\pr\p{\bv:\;\abs{\sum_{i=1}^nv_i-nm}>\delta;\ppp{\alpha_l}_{l=1}^L,\gamma\vert \ppp{\bu_l}_{l=1}^L}\nonumber\\
&\ \ \ \ \ \ \ \ \ \ \ \ \ \ \ \ \ \ \ \ \ \ \ \ +\pr\p{\bv:\;\abs{\sum_{i=1}^nv_iu_{i,l}-n\rho_l}>\delta,\;l=1,\ldots,L;\ppp{\alpha_l}_{l=1}^L,\gamma\vert \ppp{\bu_l}_{l=1}^L}.
\label{twomeasures}
\end{align}
It is easy to verify that the parameters $\ppp{\alpha_l}_{l=1}^L$ and $\gamma$ that are solving the following the following equations
\begin{align}
&\bE\ppp{\frac{1}{n}\sum_{i=1}^nv_iu_{i,l}\Biggm\vert\ppp{\bu_l}_{l=1}^L} = \rho_l, \ l=1,\ldots,L,
\end{align}
and
\begin{align}
&\bE\ppp{\frac{1}{n}\sum_{i=1}^nv_i\Biggm\vert\ppp{\bu_l}_{l=1}^L} = m
\end{align}
where the expectation is taken w.r.t. the conditional distribution \eqref{condmeas}, are also maximizing the conditional distribution (maximum-likelihood)\footnote{Essentially, this follows from the fact that \eqref{condmeas} maintains all the sufficient statistics induced by $\mathcal{F}_\delta(\ppp{\rho_l}_{l=1}^L,m)$.}. Therefore, using the strong law of large numbers (SLLN), the two terms on the right hand side of \eqref{twomeasures} are negligible as $n\to\infty$, namely,
\begin{align}
\pr\p{\bv\in\mathcal{F}^c_\delta\p{\ppp{\rho_l}_{l=1}^L,m};\alpha,\gamma\vert \ppp{\bu_l}_{l=1}^L}&\leq \tau
\label{twomeasures2}
\end{align}
for any $\tau>0$. Thus,
\begin{align}
\abs{\mathcal{F}_\delta\p{\ppp{\rho_l}_{l=1}^L,m}}\geq\p{1-\tau}R_{-\delta}.
\label{eeqq2}
\end{align} 
Whence, \eqref{eeqq1} and \eqref{eeqq2} provide tight (as $\delta\to0$) upper and lower bounds on cardinality of $\mathcal{F}_\delta\p{\ppp{\rho_l}_{l=1}^L,m}$.
\end{proof}

Returning to our problem, we will use the above result in order to find an asymptotic estimate of $\hat{\mathscr{Z}}\p{\by,\bH,\bs,m_r}$ in \eqref{Zparthat}. 
Recall that (see, \eqref{g_Nterm}) $q\p{\bs,\br}$ depends on $\bs,\br$ only through $m_s$, $m_r$, $\sum_{i=1}^n\abs{\by^T\bh_i}^2r_i$, $m_{s,r}$, $\sum_{i=1}^n\abs{\by^T\bh_i}^2s_ir_i,$ and $\sum_{i=1}^n\abs{\by^T\bh_i}^2s_i$. Accordingly, let
\begin{align}
q\p{\bs,\br} = \tilde q\p{m_s,m_r,\sum_{i=1}^n\abs{\by^T\bh_i}^2r_i,m_{s,r},\sum_{i=1}^n\abs{\by^T\bh_i}^2s_ir_i,\sum_{i=1}^n\abs{\by^T\bh_i}^2s_i}.
\label{recallq}
\end{align}
In accordance to the notations used in the definition of $\mathcal{F}_\delta(\ppp{\rho_l}_{l=1}^L,m)$ in \eqref{fcals}, define $u_{i,1} \triangleq \abs{\by^T\bh_i}^2$, $u_{i,2} \triangleq s_i$, and $u_{i,3} \triangleq \abs{\by^T\bh_i}^2s_i$, i.e., the coefficients of the terms which depend on $\br$ (recall \eqref{recallq}). Now, the main observation here is that $\hat{\mathscr{Z}}\p{\by,\bH,\bs,m_r}$ can be represented as 
\begin{align}
\hat{\mathscr{Z}}\p{\by,\bH,\bs,m_r} = 2^n\int_{\calD\subset\mathbb{R}^3}\tilde q\p{m_s,m_r,\rho_1,\rho_2,\rho_3,\sum_{i=1}^n\abs{\by^T\bh_i}^2s_i}\exp\p{nL\p{m_r}\rho_1}\mathscr{C}_n\p{\mathrm{d}\rho_1,\mathrm{d}\rho_2,\mathrm{d}\rho_3}
\label{rhoesin}
\end{align}
where $\calD$ is the codomain\footnote{Note that we do not need to explicitly define $\calD$ simply due to the fact that the exponential term in \eqref{rhoesin} is concave (see \eqref{convexratemax}), and thus the dominating $\rho_1,\rho_2,\rho_3$ are the same over $\calD$ or over $\mathbb{R}^3$.} of $\p{\rho_1,\rho_2,\rho_3}$, and $\ppp{\mathscr{C}_n}$ is a sequence of probability measures that are proportional to the number of sequences $\br$ with $\sum_{i=1}^nr_iu_{i,j}\approx n\rho_j$ for $j=1,2,3$, and $\sum_{i=1}^nr_i\approx nm_r$. These probability measures satisfy the large deviations principle (LDP) \cite[Ch. 2]{Dembo}, with the following lower semi-continuous rate function
\begin{align}
I\p{\rho_1,\rho_2,\rho_3} = 
\begin{cases}
\log 2-\frac{1}{n}\log R_{0}, &\text{if}\;\ppp{\rho_l}_{l=1}^3\in\calD\\
\infty, &\text{else}
\end{cases}
\end{align}
where $R_{0} \triangleq \lim_{\delta\to0}R_\delta$ is given in \eqref{eeqq1}. Indeed, by definition, the probability measure $\mathscr{C}_n$ is the ratio between $\abs{\mathcal{F}_\delta\p{\ppp{\rho_l}_{l=1}^3,m_r}}$ and $2^n$ (the number of possible sequences). Thus, for any Borel set $\calB\subset\calD$, $\lim_{n\to\infty}\frac{1}{n}\log\mathscr{C}_n\p{\calB} = -I\p{\rho_1,\rho_2,\rho_3}$. Accordingly, due to its large deviations properties, applying Varadhan's theorem \cite[Ch. 4.3]{Dembo} on \eqref{rhoesin}, one obtains
\begin{align}
\hat{\mathscr{Z}}\p{\by,\bH,\bs,m_r} &\sim P_n\cdot \tilde q\p{m_s,m_r,\rho_1^\circ,\rho_2^\circ,\rho_3^\circ,\sum_{i=1}^n\abs{\by^T\bh_i}^2s_i}\nonumber\\
&\ \ \ \ \ \ \ \times\exp\ppp{n\p{\log 2+L\p{m_r}\rho_1^\circ-I\p{\ppp{\rho_l^\circ}_{l=1}^3}}}\label{midLapVar}
\end{align}
where $\ppp{\rho^\circ_l}_{l=1}^3$ are given by (using the fact that the exponential term is convex)
\begin{align}
\p{\rho^\circ_1,\rho^\circ_2,\rho^\circ_3}& =\arg\max_{\rho_1,\rho_2,\rho_3\in\mathbb{R}}\ppp{\log 2+L\p{m_r}\rho_1-I\p{\ppp{\rho_l}_{l=1}^3}}\nonumber\\
&=\arg\max_{\rho_1,\rho_2,\rho_3\in\mathbb{R}}\ppp{L\p{m_r}\rho_1+\frac{1}{n}\log R_0},
\label{convexratemax}
\end{align}
and $P_n$ is a polynomial function of $n$, depending solely on the terms inside the exponent at the r.h.s. of \eqref{midLapVar}, namely, $P_n= P_n(m_r,\rho_1^\circ,\rho_2^\circ,\rho_3^\circ)$. We do not provide the explicit form of $P_n$, due to the fact that it will also appear in the normalization factor in \eqref{lastee3}, and thus, essentially, will be canceled. Continuing, the maximizers in \eqref{convexratemax} are the solutions of the following equations: $\rho_1^\circ$ is the solution of
\begin{align}
L\p{m_r}+\frac{1}{n}\frac{\partial}{\partial\rho_1}\log R_0 = 0,
\label{maxrho}
\end{align}
and $\rho_j^\circ$ for $j=2,3$, are the solutions of
\begin{align}
\frac{\partial}{\partial\rho_j}\log R_0 = 0.
\label{maxrho2}
\end{align}
We have that (for $i=1,2,3$)
\begin{align}
\frac{1}{n}\frac{\partial}{\partial\rho_i}\log R_0 &= \frac{1}{2n}\sum_{l=1}^3\frac{\partial\alpha_l^\circ}{\partial\rho_i}\sum_{i=1}^nu_{i,l}-\frac{1}{2}\frac{\partial\gamma^\circ}{\partial\rho_i}-\sum_{l=1}^3\rho_l\frac{\partial\alpha_l^\circ}{\partial\rho_i}-\alpha^\circ_i+m\frac{\partial\gamma^\circ}{\partial\rho_i}\nonumber\\
&\ \ +\frac{1}{2n}\sum_{i=1}^n\tanh\p{\frac{\sum_{l=1}^3\alpha_l^\circ u_{i,l}-\gamma^\circ}{2}}\pp{\sum_{l=1}^3u_{i,l}\frac{\partial\alpha_l^\circ}{\partial\rho_i}-\frac{\partial\gamma^\circ}{\partial\rho_i}}\\
& = -\alpha_i^\circ+\sum_{l=1}^3\frac{\partial\alpha_l^\circ}{\partial\rho_i}\pp{\frac{1}{2n}\sum_{i=1}^nu_{i,l}+\frac{1}{2n}\sum_{i=1}^n\tanh\p{\frac{\sum_{l=1}^3\alpha_l^\circ u_{i,l}-\gamma^\circ}{2}}u_{i,l}-\rho_l}\nonumber\\
&\ \ +\frac{\partial\gamma^\circ}{\partial\rho_i}\pp{m-\frac{1}{2}-\frac{1}{2n}\sum_{i=1}^n\tanh\p{\frac{\sum_{l=1}^3\alpha_l^\circ u_{i,l}-\gamma^\circ}{2}}},
\end{align}
and by using the saddle point equations \eqref{saddleeq1} and \eqref{saddleeq2}, the last two terms in the above equations vanish, and we remain with
\begin{align}
\frac{1}{n}\frac{\partial}{\partial\rho_i}\log R_0 &= -\alpha_i^\circ.
\end{align}
Thus, combined with \eqref{maxrho} and \eqref{maxrho2}, we conclude that $\alpha_1^\circ=L\p{m_r}$, and that $\alpha_2^\circ=\alpha_3^\circ = 0$. Accordingly, the exponential term in \eqref{convexratemax} boils down to
\begin{align}
L\p{m_r}\rho_1^\circ+\left.\frac{1}{n}\log R_0\right|_{\rho^\circ} &= L\p{m_r}\rho_1^\circ+\frac{1}{2n}\p{L\p{m_r}\sum_{i=1}^nu_{i,1}-n\gamma^\circ}-L\p{m_r}\rho_1^\circ+m_r\gamma^\circ\nonumber\\
&\ \ \ \ +\frac{1}{n}\sum_{i=1}^n\log\pp{2\cosh\p{\frac{L\p{m_r} u_{i,1}-\gamma^\circ}{2}}}\nonumber\\
&=m_r\gamma^\circ+\frac{1}{n}\sum_{i=1}^n\frac{L\p{m_r}u_{i,1}-\gamma^\circ}{2}+\frac{1}{n}\sum_{i=1}^n\log\pp{2\cosh\p{\frac{L\p{m_r} u_{i,1}-\gamma^\circ}{2}}}\nonumber\\
&\triangleq h\p{\delta^\circ,m_r}.
\label{hdefinre}
\end{align}
Hence, we obtained that (with the substitution of $u_{i,1} = \abs{\by^T\bh_i}^2$)
\begin{align}
&\hat{\mathscr{Z}}\p{\by,\bH,\bs,m_r} \sim P_n\cdot \tilde q\p{m_s,m_r,\rho_1^\circ,\rho_2^\circ,\rho_3^\circ,\sum_{i=1}^n\abs{\by^T\bh_i}^2s_i}\exp\p{nh\p{\gamma^\circ,m_r}}
\end{align}
where $\gamma^\circ,\ppp{\rho_l^\circ}_{l=1}^3$ solve the following set of equations (based on \eqref{saddleeq1} and \eqref{saddleeq2})
\begin{subequations}
\begin{align}
&m_r = \frac{1}{2n}\sum_{i=1}^n\pp{1+\tanh\p{\frac{L\p{m_r}\abs{\by^T\bh_i}^2-\gamma^\circ}{2}}},\\
&\rho_1^\circ = \frac{1}{2n}\sum_{i=1}^n\pp{1+\tanh\p{\frac{L\p{m_r}\abs{\by^T\bh_i}^2-\gamma^\circ}{2}}}\abs{\by^T\bh_i}^2,\\
&\rho_2^\circ = \frac{1}{2n}\sum_{i=1}^n\pp{1+\tanh\p{\frac{L\p{m_r}\abs{\by^T\bh_i}^2-\gamma^\circ}{2}}}s_i,\\
&\rho_3^\circ = \frac{1}{2n}\sum_{i=1}^n\pp{1+\tanh\p{\frac{L\p{m_r}\abs{\by^T\bh_i}^2-\gamma^\circ}{2}}}\abs{\by^T\bh_i}^2s_i.
\end{align}
\end{subequations}

Thus far, we have approximated $\hat{\mathscr{Z}}\p{\by,\bH,\bs,m_r}$. Recalling \eqref{partplusemin}, the next step in our analysis is to approximate $\tilde{\mathscr{Z}}\p{\by,\bH,\bs}$. Using the last approximation, and applying once again Varadhan's theorem (or simply, the Laplace method \cite{NeriMono,Bruijn}) on \eqref{partplusemin}, one obtains that
\begin{align}
&\tilde{\mathscr{Z}}\p{\by,\bH,\bs}  = \sum_{m_r}\exp\pp{n\p{\tilde{t}\p{m_r}}}\hat{\mathscr{Z}}\p{\by,\bH,\bs,m_r}\nonumber\\
&\sim \tilde{P}_n\cdot \tilde q\p{m_s,m_r^\circ,\rho_1^\circ\p{m_r^\circ},\rho_2^\circ\p{m_r^\circ,\bs},\rho_3^\circ\p{m_r^\circ,\bs},\sum_{i=1}^n\abs{\by^T\bh_i}^2s_i}\exp\ppp{n\p{h\p{\gamma^\circ,m_r^\circ}+\tilde{t}\p{m_r^\circ}}}
\label{asymZtildeti}
\end{align}
where $\tilde{P}_n = P_n(m_r^\circ,\rho_1^\circ,\rho_2^\circ,\rho_3^\circ)$, and the dominating $m_r^\circ$ is the saddle point, i.e., one of the solutions to the equation
\begin{align}
\frac{\partial}{\partial m}f\p{m}-\frac{1}{2}\bar{I}\p{m}-\frac{m}{2}\frac{\partial}{\partial m}\bar{I}\p{m}+\frac{1}{n}\frac{\partial}{\partial m}V\p{m}\frac{\norm{\by}^2}{n}+\frac{\partial}{\partial m}h\p{\gamma^\circ,m}=0
\label{sattleh}
\end{align}
where we have used the fact that $\tilde{t}\p{m} = f\p{m}-\frac{m}{2}\bar{I}\p{m}+V\p{m}\norm{\by}^2/n$. Simple calculations reveal that the derivative of $h\p{\gamma^\circ,m}$ w.r.t. $m$ is given by (note that $\gamma^\circ$ also depends on $m_r$)
\begin{align}
\frac{\partial}{\partial m}h\p{\gamma^\circ,m} & = \gamma^\circ+m\frac{\partial}{\partial m}\gamma^\circ+\frac{1}{n}\sum_{i=1}^n\frac{1}{2}\pp{\frac{\partial}{\partial m}L\p{m}u_{i,1}-\frac{\partial}{\partial m}\gamma^\circ}\nonumber\\
&+\frac{1}{n}\sum_{i=1}^n\tanh\p{\frac{L\p{m} u_{i,1}-\gamma^\circ}{2}}\frac{1}{2}\pp{\frac{\partial}{\partial m}L\p{m}u_{i,1}-\frac{\partial}{\partial m}\gamma^\circ}\\
& = \gamma^\circ+\frac{1}{2n}\sum_{i=1}^n\pp{1+\tanh\p{\frac{L\p{m}\abs{\by^T\bh_i}^2-\gamma^\circ}{2}}}\frac{\partial L\p{m}}{\partial m}\abs{\by^T\bh_i}^2\nonumber\\
&  +\frac{\partial}{\partial m}\gamma^\circ\pp{m-\frac{1}{2}-\frac{1}{2n}\sum_{i=1}^n\tanh\p{\frac{L\p{m} u_{i,1}-\gamma^\circ}{2}}},
\end{align}
but the last term in r.h.s. of the above equation is zero (due to \eqref{saddleeq2}), and thus
\begin{align}
\frac{\partial}{\partial m}h\p{\gamma^\circ,m}= \gamma^\circ+\frac{1}{2n}\sum_{i=1}^n\pp{1+\tanh\p{\frac{L\p{m}\abs{\by^T\bh_i}^2-\gamma^\circ}{2}}}\frac{\partial L\p{m}}{\partial m}\abs{\by^T\bh_i}^2.
\end{align}
Thus, substituting the last result in \eqref{sattleh},
\begin{align}
\gamma^\circ\p{m_r^\circ} =&-\frac{1}{2n}\sum_{i=1}^n\pp{1+\tanh\p{\frac{L\p{m_r^\circ}\abs{\by^T\bh_i}^2-\gamma^\circ}{2}}}\frac{\partial L\p{m_r^\circ}}{\partial m_r^\circ}\abs{\by^T\bh_i}^2-\frac{\partial}{\partial m_r^\circ}f\p{m_r^\circ}+\frac{1}{2}\bar{I}\p{m_r^\circ}\nonumber\\
&+\frac{m_r^\circ}{2}\frac{\partial}{\partial m_r^\circ}\bar{I}\p{m_r^\circ}-\frac{\partial}{\partial m_r^\circ}V\p{m_r^\circ}\frac{\norm{\by}^2}{n}.
\end{align}
So, hitherto, we obtained that the asymptotic behavior of $\tilde{\mathscr{Z}}\p{\by,\bH,\bs}$ is given by \eqref{asymZtildeti}, and the various dominating terms are given by
\begin{subequations}
\begin{align}
&\gamma^\circ\p{m_r^\circ} =-\frac{1}{2n}\sum_{i=1}^n\pp{1+\tanh\p{\frac{L\p{m_r^\circ}\abs{\by^T\bh_i}^2-\gamma^\circ}{2}}}\frac{\partial L\p{m_r^\circ}}{\partial m_r^\circ}\abs{\by^T\bh_i}^2-\frac{\partial}{\partial m_r^\circ}f\p{m_r^\circ}+\frac{1}{2}\bar{I}\p{m_r^\circ}\nonumber\\
&\ \ \ \ \ \ \ \ \ \ +\frac{m_r^\circ}{2}\frac{\partial}{\partial m_r^\circ}\bar{I}\p{m_r^\circ}-\frac{\partial}{\partial m_r^\circ}V\p{m_r^\circ}\frac{\norm{\by}^2}{n},\\
&m_r^\circ = \frac{1}{2n}\sum_{i=1}^n\pp{1+\tanh\p{\frac{L\p{m_r^\circ}\abs{\by^T\bh_i}^2-\gamma^\circ}{2}}},\\
&\rho_1^\circ = \frac{1}{2n}\sum_{i=1}^n\pp{1+\tanh\p{\frac{L\p{m_r^\circ}\abs{\by^T\bh_i}^2-\gamma^\circ}{2}}}\abs{\by^T\bh_i}^2,\\
&\rho_2^\circ = \frac{1}{2n}\sum_{i=1}^n\pp{1+\tanh\p{\frac{L\p{m_r^\circ}\abs{\by^T\bh_i}^2-\gamma^\circ}{2}}}s_i,\\
&\rho_3^\circ = \frac{1}{2n}\sum_{i=1}^n\pp{1+\tanh\p{\frac{L\p{m_r^\circ}\abs{\by^T\bh_i}^2-\gamma^\circ}{2}}}\abs{\by^T\bh_i}^2s_i.
\end{align}
\end{subequations}

This concludes the asymptotic analysis of the summation over $\br$ in \eqref{B10}. We now take care of the summation over $\bs$ in \eqref{partplusminus2}. Let
\begin{align}
\hat q\p{\bs}\triangleq \tilde q\p{m_{\bst},m_r^\circ,\rho_1^\circ\p{m_r^\circ},\rho_2^\circ\p{m_r^\circ,\bs},\rho_3^\circ\p{m_r^\circ,\bs},\sum_{i=1}^n\abs{\by^T\bh_i}^2s_i}.
\end{align}
Applying \eqref{asymZtildeti} on \eqref{partplusminus2},
\begin{align}
\mathscr{Z}\p{\by,\bH}&\sim \tilde{P}_n\cdot e^{\ppp{n\p{h\p{\gamma^\circ,m_r^\circ}+\tilde{t}\p{m_r^\circ}}}}\sum_{\bs\in\ppp{0,1}^n}\hat q\p{\bs}\exp\left\{n\left(\tilde{t}\p{m_{\bst}}+L\p{m_{\bst}}\frac{1}{n}\sum_{i=1}^n\abs{\by^T\bh_i}^2s_i\right)\right\}\nonumber\\
&\triangleq \tilde{P}_n\cdot e^{\ppp{n\p{h\p{\gamma^\circ,m_r^\circ}+\tilde{t}\p{m_r^\circ}}}}\sum_{m_s}\exp\p{n\tilde{t}\p{m_s}}\bar{\mathscr{Z}}\p{\by,\bH,m_s}\label{b57}
\end{align}
where as before
\begin{align}
\bar{\mathscr{Z}}\p{\by,\bH,m_s}\triangleq\sum_{\bst:\;m_{\bst} = m_s}\hat q\p{\bs}\exp\p{L\p{m_s}\sum_{i=1}^n\abs{\by^T\bh_i}^2s_i}.
\end{align}
However, $\bar{\mathscr{Z}}\p{\by,\bH,m_s}$ has essentially the same form of $\tilde{\mathscr{Z}}\p{\by,\bH,\bs,m_r}$, which we have analyzed earlier. So, using the same technique,
\begin{align}
\bar{\mathscr{Z}}\p{\by,\bH,m_s}\sim P_n\cdot \bar{q}\p{m_s}\exp\p{nh\p{\tilde{\gamma}^\circ,m_s}}
\end{align}
where $h\p{\tilde{\gamma}^\circ,m_s}$ is defined as in \eqref{hdefinre} (note that the exponential term is similar to the previous one), and
\begin{align}
\bar{q}\p{m_s}\triangleq \tilde q\p{m_s,m_r^\circ,\rho_1^\circ\p{m_r^\circ},\rho_2^\circ\p{m_r^\circ,m_s},\rho_3^\circ\p{m_r^\circ,m_s},\rho_4^\circ\p{m_s}},
\end{align}
in which $\tilde{\gamma}^\circ,\ppp{\rho_l^\circ}_{l=2}^4$ solve the following set of equations
\begin{subequations}
\begin{align}
&m_s = \frac{1}{2n}\sum_{i=1}^n\pp{1+\tanh\p{\frac{L\p{m_s}\abs{\by^T\bh_i}^2-\tilde{\gamma}^\circ}{2}}},\\
&\rho_2^\circ = \frac{1}{4n}\sum_{i=1}^n\pp{1+\tanh\p{\frac{L\p{m_r^\circ}\abs{\by^T\bh_i}^2-\gamma^\circ}{2}}}\pp{1+\tanh\p{\frac{L\p{m_s}\abs{\by^T\bh_i}^2-\tilde{\gamma}^\circ}{2}}},\\
&\rho_3^\circ = \frac{1}{4n}\sum_{i=1}^n\pp{1+\tanh\p{\frac{L\p{m_r^\circ}\abs{\by^T\bh_i}^2-\gamma^\circ}{2}}}\pp{1+\tanh\p{\frac{L\p{m_s}\abs{\by^T\bh_i}^2-\tilde{\gamma}^\circ}{2}}}\abs{\by^T\bh_i}^2\\
&\rho_4^\circ = \frac{1}{2n}\sum_{i=1}^n\pp{1+\tanh\p{\frac{L\p{m_s}\abs{\by^T\bh_i}^2-\tilde{\gamma}^\circ}{2}}}\abs{\by^T\bh_i}^2.
\end{align}
\end{subequations}

Finally, the summation over $m_s$ in \eqref{b57} is again estimated by using the Laplace method, and we similarly obtain
\begin{align}
\mathscr{Z}\p{\by,\bH}&\sim \tilde{P}_n^2\cdot \tilde q\p{m_s^\circ,m_r^\circ,\rho_1^\circ\p{m_r^\circ},\rho_2^\circ\p{m_r^\circ,m_s^\circ},\rho_3^\circ\p{m_r^\circ,m_s^\circ},\rho_4^\circ\p{m_s^\circ}}\nonumber\\
&\ \ \ \ \times\exp\ppp{n\p{h\p{\gamma^\circ,m_r^\circ}+h\p{\tilde{\gamma}^\circ,m_s^\circ}+\tilde{t}\p{m_r^\circ}+\tilde{t}\p{m_s^\circ}}}\label{lastssa}
\end{align}
where 
\begin{align}
&\gamma^\circ\p{m_r^\circ} =-\frac{1}{2n}\sum_{i=1}^n\pp{1+\tanh\p{\frac{L\p{m_r^\circ}\abs{\by^T\bh_i}^2-\gamma^\circ}{2}}}\frac{\partial L\p{m_r^\circ}}{\partial m_r^\circ}\abs{\by^T\bh_i}^2-\frac{\partial}{\partial m_r^\circ}f\p{m_r^\circ}+\frac{1}{2}\bar{I}\p{m_r^\circ}\nonumber\\
&\ \ \ \ \ \ \ \ \ \ \ \ \ +\frac{m_r^\circ}{2}\frac{\partial}{\partial m_r^\circ}\bar{I}\p{m_r^\circ}-\frac{\partial}{\partial m_r^\circ}V\p{m_r^\circ}\frac{\norm{\by}^2}{n},\nonumber\\
&\tilde{\gamma}^\circ\p{m_s^\circ} =-\frac{1}{2n}\sum_{i=1}^n\pp{1+\tanh\p{\frac{L\p{m_s^\circ}\abs{\by^T\bh_i}^2-\tilde{\gamma}^\circ}{2}}}\frac{\partial L\p{m_s^\circ}}{\partial m_s^\circ}\abs{\by^T\bh_i}^2-\frac{\partial}{\partial m_s^\circ}f\p{m_s^\circ}+\frac{1}{2}\bar{I}\p{m_s^\circ}\nonumber\\
&\ \ \ \ \ \ \ \ \ \ \ \ \ +\frac{m_s^\circ}{2}\frac{\partial}{\partial m_s^\circ}\bar{I}\p{m_s^\circ}-\frac{\partial}{\partial m_s^\circ}V\p{m_s^\circ}\frac{\norm{\by}^2}{n},\nonumber\\
&m_r^\circ = \frac{1}{2n}\sum_{i=1}^n\pp{1+\tanh\p{\frac{L\p{m_r^\circ}\abs{\by^T\bh_i}^2-\gamma^\circ}{2}}},\nonumber\\
&m_s^\circ = \frac{1}{2n}\sum_{i=1}^n\pp{1+\tanh\p{\frac{L\p{m_s^\circ}\abs{\by^T\bh_i}^2-\tilde{\gamma}^\circ}{2}}},\nonumber\\
&\rho_1^\circ = \frac{1}{2n}\sum_{i=1}^n\pp{1+\tanh\p{\frac{L\p{m_r^\circ}\abs{\by^T\bh_i}^2-\gamma^\circ}{2}}}\abs{\by^T\bh_i}^2,\\
&\rho_2^\circ = \frac{1}{4n}\sum_{i=1}^n\pp{1+\tanh\p{\frac{L\p{m_r^\circ}\abs{\by^T\bh_i}^2-\gamma^\circ}{2}}}\pp{1+\tanh\p{\frac{L\p{m_s^\circ}\abs{\by^T\bh_i}^2-\tilde{\gamma}^\circ}{2}}},\nonumber\\
&\rho_3^\circ = \frac{1}{4n}\sum_{i=1}^n\pp{1+\tanh\p{\frac{L\p{m_r^\circ}\abs{\by^T\bh_i}^2-\gamma^\circ}{2}}}\pp{1+\tanh\p{\frac{L\p{m_s^\circ}\abs{\by^T\bh_i}^2-\tilde{\gamma}^\circ}{2}}}\abs{\by^T\bh_i}^2\nonumber\\
&\rho_4^\circ = \frac{1}{2n}\sum_{i=1}^n\pp{1+\tanh\p{\frac{L\p{m_s^\circ}\abs{\by^T\bh_i}^2-\tilde{\gamma}^\circ}{2}}}\abs{\by^T\bh_i}^2.
\end{align}
Due to the symmetry between $\bs$ and $\br$, it can be seen that $m_s^\circ = m_r^\circ$, and whence the above set of equations reduce to
\begin{subequations}
\begin{align}
&\gamma^\circ =-\frac{1}{2n}\sum_{i=1}^n\pp{1+\tanh\p{\frac{L\p{m^\circ}\abs{\by^T\bh_i}^2-\gamma^\circ}{2}}}\frac{\partial L\p{m^\circ}}{\partial m^\circ}\abs{\by^T\bh_i}^2-\frac{\partial}{\partial m^\circ}f\p{m^\circ}+\frac{1}{2}\bar{I}\p{m^\circ}\nonumber\\
&\ \ \ \ \ \ \ \ \ \ \ \ \ +\frac{m^\circ}{2}\frac{\partial}{\partial m^\circ}\bar{I}\p{m^\circ}-\frac{\partial}{\partial m^\circ}V\p{m^\circ}\frac{\norm{\by}^2}{n},\label{saddlepoo1}\\
&m^\circ = \frac{1}{2n}\sum_{i=1}^n\pp{1+\tanh\p{\frac{L\p{m^\circ}\abs{\by^T\bh_i}^2-\gamma^\circ}{2}}},\label{saddlepoo2}\\
&\rho_1^\circ = \rho_4^\circ=\frac{1}{2n}\sum_{i=1}^n\pp{1+\tanh\p{\frac{L\p{m^\circ}\abs{\by^T\bh_i}^2-\gamma^\circ}{2}}}\abs{\by^T\bh_i}^2,\label{saddlepoo3}\\
&\rho_2^\circ = \frac{1}{4n}\sum_{i=1}^n\pp{1+\tanh\p{\frac{L\p{m^\circ}\abs{\by^T\bh_i}^2-\gamma^\circ}{2}}}^2,\label{saddlepoo4}\\
&\rho_3^\circ = \frac{1}{4n}\sum_{i=1}^n\pp{1+\tanh\p{\frac{L\p{m^\circ}\abs{\by^T\bh_i}^2-\gamma^\circ}{2}}}^2\abs{\by^T\bh_i}^2,\label{saddlepoo5}
\end{align}\label{saddlepoo}%
\end{subequations}
and by using \eqref{g_Nterm}
\begin{align}
\tilde q(m^\circ,\ppp{\rho_l^\circ}_{l=1}^3)=& \beta^2\frac{\alpha\p{m^\circ,\rho_2^\circ}}{g^2\p{m^\circ}}\rho_3^\circ-2\frac{\alpha\p{m^\circ,\rho_2^\circ}b\p{m^\circ}}{g^3\p{m^\circ}}\beta^3\sigma^2\rho_2^\circ\pp{\rho_1^\circ-m^\circ\frac{\norm{\by}^2}{n}},
\end{align}
where $\alpha\p{x,y} \triangleq \tilde\alpha\p{x,x,y}$. Finally,
\begin{align}
\mathscr{Z}\p{\by,\bH}&\sim \tilde{P}_n^2\cdot \tilde q(m^\circ,\ppp{\rho_l^\circ}_{l=1}^3)\exp\ppp{2n\pp{h\p{\gamma^\circ,m^\circ}+\tilde{t}\p{m^\circ}}}.\label{lastAsympComp}
\end{align}

Based on \eqref{lastee3}, we also need to find the asymptotic behavior of
\begin{align}
\sum_{\bst\in\ppp{0,1}^n}P_{\bSt}(\bs)\xi\p{\by,\bHtt}\label{exponsaddle},
\end{align}
However, obviously, the previous analyzed term can be regarded as an extended version of \eqref{exponsaddle}, and so we can immediately conclude that\footnote{As mentioned earlier (see, \eqref{convexratemax}), the polynomial term $\tilde{P}_n$ depends only on the exponential behavior of the summations, and thus, common to \eqref{first1sum}.}
\begin{align}
&\sum_{\bst\in\ppp{0,1}^n}P_{\bSt}(\bs)\xi_\epsilon\p{\by,\bHtt}\sim \tilde{P}_n\cdot\exp\ppp{n\p{h\p{\gamma^\circ,m^\circ}+\tilde{t}\p{m^\circ}}}.\label{first1sum}\end{align}
Indeed, recall that what we have analyzed above is
\begin{align}
\sum_{\bst\in\ppp{0,1}^n}\sum_{\brt\in\ppp{0,1}^n}P_{\bSt}(\bs)P_{\bSt}(\br)J\p{\by,\bHtt,\bHtr}\xi\p{\by,\bHtt}\xi\p{\by,\bHtr},\label{whatwedid}
\end{align}
and so, \eqref{exponsaddle} is just a special case of \eqref{whatwedid}, in which the summation is only over $\bs$ and without the leading term $J\p{\by,\bHtt,\bHtr}$. Whence, the asymptotic behavior of \eqref{exponsaddle} is affected only by $\xi\p{\by,\bHtt}$, which after multiplying by $P_{\bSt}(\bs)$ and summing over $\ppp{0,1}^n$, asymptotically behaves as the exponent at the r.h.s. of \eqref{lastssa} (of course, as we sum over $\bs$, only the terms related to $m_s^\circ$ prevail). 

Wrapping up, using \eqref{lastAsympComp} and \eqref{first1sum}, the asymptotic estimate of the inner term of the expectation in \eqref{lastee3} is given by
\begin{align}
g_n &\triangleq\frac{\sigma^2}{n}\sum_{i=1}^nS_i-\bE_{\mu_{s\times r}}\pp{J\p{\bY,\bHtt,\bHtr}\Ind_{\calT_\epsilon^{s,r}}}\\
&\asymp \sigma^2m_a- \tilde q(m^\circ,\ppp{\rho_l^\circ}_{l=1}^3)\\
& = \sigma^2m_a-\beta^2\frac{\alpha\p{m^\circ,\rho_2^\circ}}{g^2\p{m^\circ}}\rho_3^\circ+2\frac{\alpha\p{m^\circ,\rho_2^\circ}b\p{m^\circ}}{g^3\p{m^\circ}}\beta^3\sigma^2\rho_2^\circ\pp{\rho_1^\circ-m^\circ\p{m_a\sigma^2R+\frac{R}{\beta}}}\\
&\triangleq g_\infty.
\end{align}
Thus, we obtained that $g_n\to g_\infty$ a.s., as $n\to\infty$. In order to calculate the MMSE we will apply Lemma \ref{aplem:DCT2}. First, recall that
\begin{align}
g_n = \frac{1}{n}\sum_{i=1}^n\pp{\bE\ppp{X_i^2\vert\by,\bH}-\p{\bE\ppp{X_i\vert\by,\bH}}^2},
\end{align}
and thus, due to Jensens's inequality, $g_n$ is nonnegative for any $n$. Then, for any $\varepsilon>0$, using Cauchy-Schwartz and Chebyshev's inequalities, we get
 \begin{align}
\limsup_{n\to\infty}\bE\ppp{g_n\cdot\Ind_{g_n\geq c(\varepsilon)}}&\leq \limsup_{n\to\infty}\p{\bE g_n^2\cdot\Pr\ppp{g_n\geq c(\varepsilon)}}^{1/2}\\
&\leq \limsup_{n\to\infty}\p{\bE g_n^2}^{1/2}\p{\frac{\bE g_n}{c(\varepsilon)}}^{1/2}
\end{align}
where $c(\varepsilon)$ is a non-negative real. Now, by the definition of the MMSE, we know that $\bE g_n\leq\sigma^2$ and that
\begin{align}
\bE g_n^2&\leq \frac{1}{n}\bE\pp{\sum_{i=1}^n\pp{\bE\ppp{X_i^2\vert\by,\bH}-\p{\bE\ppp{X_i\vert\by,\bH}}^2}^2}\\
&\leq \frac{1}{n}\bE\pp{\sum_{i=1}^n\p{\bE\ppp{X_i^2\vert\by,\bH}}^2+\p{\bE\ppp{X_i\vert\by,\bH}}^4}\\
&\leq \frac{1}{n}\bE\pp{\sum_{i=1}^n\bE\ppp{X_i^4\vert\by,\bH}+\bE\ppp{X_i^4\vert\by,\bH}}\\
& = \frac{2}{n}\sum_{i=1}^n\bE\ppp{X_i^4}\leq 6\sigma^4,
\end{align}
where the first inequality follows from the fact that $(a_1+\ldots+a_n)^2\leq n\cdot(a_1^2+\ldots+a_n^2)$, the third inequality is due to Jensens's inequality, and in the last inequality we have used the fact that $n^{-1}\sum_i S_i\leq 1$ w.p. 1. Therefore,
\begin{align}
\limsup_{n\to\infty}\bE\ppp{g_n\cdot\Ind_{g_n\geq c(\varepsilon)}}&\leq \frac{\sqrt{6}\sigma^3}{c^{1/2}(\varepsilon)}=\varepsilon
\end{align}
where the last inequality follows by taking $c^{1/2}(\varepsilon) = \sqrt{6}\sigma^3/\varepsilon$. Thus, we can apply Lemma \ref{aplem:DCT2}, and obtain
\begin{align}
\lim_{n\to\infty}\frac{\text{mmse}\p{\bX\vert\bY,\bH}}{n}&=\sigma^2m_a-\beta^2\frac{\alpha\p{m^\circ,\rho_2^\circ}}{g^2\p{m^\circ}}\rho_3^\circ\nonumber\\
&\ \ \ +2\frac{\alpha\p{m^\circ,\rho_2^\circ}b\p{m^\circ}}{g^3\p{m^\circ}}\beta^3\sigma^2\rho_2^\circ\pp{\rho_1^\circ-m^\circ\p{m_a\sigma^2R+\frac{R}{\beta}}}.
\label{MMSEbeforeex}
\end{align}

Finally, we show a concentration property of the saddle point equations given in \eqref{saddlepoo}, and obtain ``instead" the saddle point equations given in \eqref{magnetddd}-\eqref{rhoa3}. Accordingly, the expectation in \eqref{MMSEbeforeex} becomes ``superfluous", as all the involved random variables ($m^\circ$ and $\ppp{\rho_i^\circ}_{i=1}^3$) converge to a deterministic quantity. According to \eqref{saddlepoo}, it can be seen that the saddle point equations share the following common term
\begin{align}
\frac{1}{n}\sum_{i=1}^n\phi\p{\abs{\bh_i^T\bY}^2}
\label{sumCLT}
\end{align}
where $\phi\p{\cdot}:\mathbb{R}\to\mathbb{R}$ is some integrable function (in the $L^1$ sense). In the following, we first show that \eqref{sumCLT} admits an SLLN property. To this end, let us define 
\begin{align}
T_n \triangleq \sum_{i=1}^nK_i,
\end{align}
where $K_i\triangleq \phi\p{\abs{\bh_i^T\bY}^2}$, and let $\calG_n = \sigma\p{\bX,\bW}\cap\sigma\p{T_n,T_{n+1},\ldots}$ be the $\sigma$-field (filtration) generated by $T_n$, $\ppp{K_i}_{i>n}$, $\bX$, and $\bW$. We will now show that $M_n\triangleq -\frac{T_{-n}}{n}$ is a backwards martingale sequence w.r.t. $\calF_n\triangleq \calG_{-n}$, $n\leq -1$. Indeed, for $m\leq -1$, we have that
\begin{align}
\bE\ppp{M_{m+1}\Biggm\vert\calF_m} = \bE\ppp{\frac{T_{-m-1}}{-m-1}\Biggm\vert\calG_{-m}}.
\end{align}
Setting $n=-m$, we see that
\begin{align}
\bE\ppp{\frac{T_{n-1}}{n-1}\Biggm\vert\calG_{n}} &= \bE\ppp{\frac{T_{n}-K_n}{n-1}\Biggm\vert\calG_{n}}\\
& = \frac{T_n}{n-1} - \bE\ppp{\frac{K_n}{n-1}\Biggm\vert\calG_{n}}
\label{concbe}
\end{align}
where we have used the fact that $T_n$ is measurable w.r.t. $\calG_{n}$. Now, we have that
\begin{align}
\bE\ppp{K_n\vert\calG_{n}\cap\sigma\p{\bY}} &= \bE\ppp{K_n\vert T_n,\bY,\sigma\p{\bX,\bW}}\\
&= \bE\ppp{K_j\vert T_n,\bY,\sigma\p{\bX,\bW}}
\label{symmetry}
\end{align}
for any $1\leq j\leq n$, where in the first equality we have used the facts that $\calG_n = \sigma\p{\bX,\bW}\cap\sigma\p{T_n,T_{n+1},\ldots} =  \sigma\p{\bX,\bW}\cap\sigma\p{T_n,K_{n+1},K_{n+2},\ldots}$, that $\bY = \sum_{i=1}^n\bh_iX_i+\bW$ and that $\ppp{\bh_i}$ are statistically independent, and the second equality follows due to $\bY = \bH\bX+\bW$, the symmetry of $T_n$ w.r.t. $K_1,\ldots, K_n$, and the fact that $\ppp{\bh_i}$ are statistically independent. Clearly,
\begin{align}
\sum_{i=1}^n\bE\ppp{K_i\vert T_n,\bY,\sigma\p{\bX,\bW}} &= \bE\ppp{\sum_{i=1}^nK_i\Biggm\vert T_n,\bY,\sigma\p{\bX,\bW}}\\
& = T_n,
\end{align} 
and thus, due to \eqref{symmetry}, we obtain that $\bE\ppp{K_n\vert\calG_{n}\cap\sigma\p{\bY}} = T_n/n$ a.s. Whence, using \eqref{concbe} and the last result, we obtain
\begin{align}
\bE\ppp{\frac{T_{n-1}}{n-1}\Biggm\vert\calG_{n}} &= \frac{T_n}{n-1} - \bE\ppp{\frac{K_n}{n-1}\Biggm\vert\calG_{n}}\\
& = \frac{T_n}{n-1} - \bE\ppp{\bE\ppp{\frac{K_n}{n-1}\Biggm\vert\calG_{n}\cap\sigma\p{\bY}}\Biggm\vert\calG_{n}}\\
& = \frac{T_n}{n-1} - \frac{T_n}{n\p{n-1}}  = \frac{T_n}{n},\ \ \ \ \ \text{a.s.}
\end{align}
This concludes the proof that $M_n$ is a backwards martingale sequence w.r.t. $\ppp{\calF_n}_{n\leq -1}$. Now, by the backwards martingale convergence theorem \cite{martingale,martingale2}, we deduce that $T_n/n$ converges as $n\to\infty$, and in $L^1$, to a random variable $K \triangleq \lim_{n\to\infty}T_n/n$. Obviously, for all $m$
\begin{align}
K = \lim_{n\to\infty}\frac{\tilde{K}_{m+1}+\ldots+\tilde{K}_{m+n}}{n},
\end{align}
where (due to the fact that $\ppp{\bh_i}_i$ are i.i.d.)
\begin{align}
\tilde{K}_{m+i} = \phi\p{\abs{\bh_{m+i}^T\p{\sum_{j=m+i}^{n+m+i}\bh_jX_j +\bW}}^2},\ \ \text{for}\ i=1,\ldots,n.
\end{align}
Thus $K$ is $\sigma\p{\bX,\bW}\cap\sigma\p{\bh_{m+1},\ldots}$-measurable, for all $m$, and hence it is also $\sigma\p{\bX,\bW}\cap\bigcap_{m}\sigma\p{\bh_{m+1},\ldots}$-measurable (namely, the tail $\sigma$-field generated by $\ppp{\bh_i}$ intersected with $\sigma\p{\bX,\bW}$). Thus, by the Kolmogorov's 0-1 law \cite{martingale}, we conclude that there exists a constant $C\in\mathbb{R}$ (w.r.t. $\sigma\p{\bX,\bW}$) such that $\pr\ppp{K=C\vert\sigma\p{\bX,\bW}} = 1$. This constant is obviously given by
\begin{align}
C = \bE\ppp{K\vert\sigma\p{\bX,\bW}} = \lim_{n\to\infty}\bE\ppp{\frac{T_n}{n}\Biggm\vert\sigma\p{\bX,\bW}}.
\end{align}
Thus, we have shown that
\begin{align}
\frac{1}{n}\sum_{i=1}^n\phi\p{\abs{\bh_i^T\bY}^2} - \frac{1}{n}\bE\ppp{\sum_{i=1}^n\phi\p{\abs{\bh_i^T\bY}^2}\Biggm\vert\bX,\bW}\to0,
\end{align}
a.s. as $n\to\infty$, namely, we show an SLLN property of \eqref{sumCLT}. Our next step is to infer the asymptotic behavior of each summand. First, we note that
\begin{align}
\bh_i^T\bY &= \bh_i^T\pp{\bH\bX}_i+ X_i\norm{\bh_i}^2 + \bh_i^T\bW
\end{align}  
where $\pp{\bH\bX}_i \triangleq \bH\bX-\bh_iX_i$. Let $\hat{\bX}_i$ be a new $n$-dimensional vector, such that its $i$th component is zero and the other components are identical to that of $\bX$. Similarly, let $\hat{\bH}_i$ denote a new matrix such that its $i$th column contains zeros, and the other columns are identical to those of $\bH$. Accordingly, let $\hat{\bz}_{i,j}$ denote the $j$th row of $\hat{\bH}_i$. With this notations, we have that $\pp{\bH\bX}_i = \hat{\bH}_i\hat{\bX}_i$. Thus,
\begin{align}
\bh_i^T\bY &= \sum_{j=1}^kH_{j,i}\pp{\hat{\bz}_{i,j}^T\hat{\bX}_i+W_j}+X_i\norm{\bh_i}^2\\
&= \frac{1}{\sqrt{n}}\sum_{j=1}^k\tilde{H}_{j,i}\pp{\hat{\bz}_{i,j}^T\hat{\bX}_i+W_j}+X_i\norm{\bh_i}^2.
\label{CLT1}
\end{align}
where $\tilde{H}_{i,j} \triangleq \sqrt{n}H_{i,j}$. Given $\bX$, by using Lyapunov's central limit theorem \cite{probability}, we may infer the following weak convergence
\begin{align}
\frac{1}{\sqrt{n}}\sum_{j=1}^k\tilde{H}_{j,i}\pp{\hat{\bz}_{i,j}^T\hat{\bX}_i+W_j} \stackrel{d}{\longrightarrow} \calN\p{0,Rm_a\sigma^2+\frac{R}{\beta}},
\label{CLT2}
\end{align}
as $n\to\infty$. Accordingly, let $\mathscr{Y}$ be the limit point in \eqref{CLT2}, namely, $\mathscr{Y}$ is distributed $\calN\p{0,m_a\sigma^2R+R/\beta}$. Therefore, based on \eqref{CLT1}, \eqref{CLT2}, and Slutsky's lemma \cite[Lemma 2.8]{contimap}, we may conclude that (conditioned on $\bX$)
\begin{align}
\bh_i^T\bY\stackrel{d}{\longrightarrow} \mathscr{Y}+RX_i.
\end{align}
Using the last results, and Lemmas \ref{continuouslemma} and \ref{aplem:DCT2}, we obtain that\footnote{In our case, the sequence of random variables $\phi\p{\abs{\bh_i^T\bY}^2}$ meet the asymptotic uniform integrability assumption of Lemma \ref{aplem:DCT2}, for the various choices of $\phi$ according to \eqref{magnetddd}-\eqref{rhoa3}.}
\begin{align}
\frac{1}{n}\sum_{i=1}^n\phi\p{\abs{\bh_i^T\bY}^2} - \frac{1}{n}\bE\ppp{\sum_{i=1}^n\phi\p{\abs{\mathscr{Y}+RX_i}^2}\Biggm\vert\bX}\to0.
\label{lastonpAp}
\end{align}
Now, applying the SLLN on \eqref{lastonpAp}, we finally may write that
\begin{align}
\frac{1}{n}\sum_{i=1}^n\phi\p{\abs{\bh_i^T\bY}^2}\to \bE\pp{\phi\p{\abs{\mathscr{Y}+RX}^2}},
\end{align}
a.s. as $n\to\infty$, where the expectation is taken w.r.t. the product measure corresponding to $\mathscr{Y}$, and $X$ which is distributed according to a mixture of two measures: Dirac measure at $0$ with weight $1-m_a$, and a Gaussian measure with zero mean and variance $\sigma^2$ and weight $m_a$. Equivalently, the last result can be rewritten as
\begin{align}
\frac{1}{n}\sum_{i=1}^n\phi\p{\abs{\bh_i^T\bY}^2}\to \bE\pp{\phi\p{\abs{\mathscr{X}}^2}}\label{AsymReSad},
\end{align}
a.s. as $n\to\infty$, where the expectation over $\mathscr{X}$ is now taken w.r.t. a mixture of two measures: Gaussian measure with zero mean and variance $\p{m_a\sigma^2R+R/\beta}$ and weight $1-m_a$, and a Gaussian measure with zero mean and variance $\p{m_a\sigma^2R+R/\beta+R^2\sigma^2}$ and weight $m_a$. 

Next, we wish to apply the last general asymptotic result to the saddle point equations given in \eqref{saddlepoo}, and obtain
\begin{align}
&\gamma^\circ =-\frac{1}{2}\bE\ppp{\pp{1+\tanh\p{\frac{L\p{m^\circ}\abs{\mathscr{X}}^2-\gamma^\circ}{2}}}\left.\frac{\mathrm{d}L\p{m}}{\mathrm{d}m}\right|_{m=m^\circ}\abs{\mathscr{X}}^2}-\left.\frac{\mathrm{d}t\p{m}}{\mathrm{d}m}\right|_{m=m^\circ},\label{f1}\\
&m^\circ = \frac{1}{2}\bE\ppp{1+\tanh\p{\frac{L\p{m^\circ}\abs{\mathscr{X}}^2-\gamma^\circ}{2}}},\label{f2}\\
&\rho_1^\circ = \rho_4^\circ=\frac{1}{2}\bE\ppp{\pp{1+\tanh\p{\frac{L\p{m^\circ}\abs{\mathscr{X}}^2-\gamma^\circ}{2}}}\abs{\mathscr{X}}^2},\label{f3}\\
&\rho_2^\circ = \frac{1}{4}\bE\ppp{\pp{1+\tanh\p{\frac{L\p{m^\circ}\abs{\mathscr{X}}^2-\gamma^\circ}{2}}}^2},\label{f4}\\
&\rho_3^\circ = \frac{1}{4}\bE\ppp{\pp{1+\tanh\p{\frac{L\p{m^\circ}\abs{\mathscr{X}}^2-\gamma^\circ}{2}}}^2\abs{\mathscr{X}}^2},\label{f5}
\end{align}
where for \eqref{f1}-\eqref{f5} the following choices of $\phi$ have been used
\begin{align}
\phi(x) &= \frac{1}{2}\pp{1+\tanh\p{\frac{L(m^\circ)x-\gamma^\circ}{2}}}\frac{\partial L\p{m^\circ}}{\partial m^\circ}x\label{fundd1}\\
\phi(x) &= \frac{1}{2}\pp{1+\tanh\p{\frac{L(m^\circ)x-\gamma^\circ}{2}}}\label{fundd2}\\
\phi(x) &= \frac{1}{2}\pp{1+\tanh\p{\frac{L(m^\circ)x-\gamma^\circ}{2}}}x\label{fundd3}\\
\phi(x) &= \frac{1}{4}\pp{1+\tanh\p{\frac{L(m^\circ)x-\gamma^\circ}{2}}}^2\label{fundd4}
\end{align}
and
\begin{align}
\phi(x) = \frac{1}{4}\pp{1+\tanh\p{\frac{L(m^\circ)x-\gamma^\circ}{2}}}^2x\label{fundd5},
\end{align}
respectively. Indeed, the convergence of $\rho_i^\circ$ for $i=1,2,3$, in \eqref{saddlepoo3}-\eqref{saddlepoo5}, follows directly by considering the choices in \eqref{fundd3}-\eqref{fundd5}, and using \eqref{AsymReSad}, respectively. However, the convergence of \eqref{saddlepoo1} and \eqref{saddlepoo2} is more delicate. Specifically, consider, for example, the convergence of \eqref{saddlepoo2} (the convergence of \eqref{saddlepoo1} is handled in a similar manner), and let $m_n^\circ$ designate the solution of \eqref{saddlepoo2} for a fixed $n$ (now we emphasize the dependency of the saddle point on $n$), that is,
\begin{align}
m_n^\circ = \frac{1}{2n}\sum_{i=1}^n\pp{1+\tanh\p{\frac{L\p{m_n^\circ}\abs{\by^T\bh_i}^2-\tilde{\gamma}}{2}}} \define \phi_n(m_n^\circ),
\end{align}
for any $\tilde{\gamma}$. We already saw that for a fixed $x$, $\phi_n(x)\to \phi_\infty(x)$ a.s. pointwise. Now, we wish to show that the sequence of random variables $\ppp{m_n^\circ}$ converges to the solution of $m^\circ = \phi_\infty(m^\circ)$. To this end, note that the sequence $\ppp{m_n^\circ}$ is bounded\footnote{Letting $\gamma_n^\circ$ designate the solution of \eqref{saddlepoo1} for a fixed $n$, the boundedness is, essentially, guaranteed by definition. Alternatively, it can be shown that the set of vectors $\ppp{\by,\ppp{\bh_i}}$ for which $\ppp{\gamma_n^\circ}$ is bounded, is of probability 1, for large $k$ and $n$.} in a compact set, and thus, by Bolzano-Weierstrass theorem, there must exist a converging subsequence $\ppp{m_{n_l}^\circ}$ along this sequence. Denoting its limit by $m_\infty^\circ$, we get
\begin{align}
m_\infty^\circ = \lim_{n\to\infty}\phi_n(m_{n_l}^\circ).
\end{align}
However, due to the fact that $\phi_n(\cdot)$ is continuous, we have that
\begin{align}
m_\infty^\circ = \phi_\infty(m_\infty^\circ).\label{equation}
\end{align}
Finally, we show the existence of a solution to \eqref{saddlepoo2}. This is equivalent to showing that there exists a solution $x_0\in\pp{0,1}$ to the equation $x = \frac{1}{2}\pp{1+\tanh(f(x))}$. This follows from the fact that $y(x) = \frac{1}{2}\pp{1+\tanh(f(x))}$ is a bounded between zero and one and thus must have an intersection with the linear function $y(x) = x$ within the interval $\pp{0,1}$. 
\end{proof}

\section{Mathematical Tools}\label{sec:mathtoll}
\begin{lemma}[\cite{SilversteinBai}][Matrix Inversion Lemma]\label{aplem:1}
Let $\bU$ be an $N \times N$ invertible matrix and $\bx \in \mathbb{C}^N$, $c \in\mathbb{C}$ for which $\bU + c\bx\bx^H$ is invertible. Then
\begin{align}
\bx^H\p{\bU+c\bx\bx^H}^{-1} = \frac{\bx^H\bU^{-1}}{1+c\bx^H\bU^{-1}\bx}.
\end{align}
\end{lemma}
\begin{lemma}[Matrix Inversion Lemma 2]\label{aplem:2}
Under the assumptions of Lemma \ref{aplem:1},
\begin{align}
\p{\bU+c\bx\bx^H}^{-1} = \bU^{-1}-\frac{\bU^{-1}c\bx\bx^H\bU^{-1}}{1+c\bx^H\bU^{-1}\bx}.
\end{align}
\end{lemma}
\begin{lemma}[Resolvent Identity]\label{aplem:3}
Let $\bU$ and $\bV$ be two invertible complex matrices of size $N\times N$. Then
\begin{align}
\bU^{-1}-\bV^{-1} = -\bU^{-1}\p{\bU-\bV}\bV^{-1}.
\end{align}
\end{lemma}
The following lemma is a powerful tool which is widely used in RMT with many versions and extensions.

\begin{lemma}[\cite{baisilbook,coulbook}]\label{aplem:trace}
Let $\bA_N \in\mathbb{C}^{N\times N}$ be a sequence of deterministic matrices, and let $\bx_N\in\mathbb{C}^N$ have i.i.d. complex entries with zero mean, variance $1/N$, and bounded $l$th order moment $\bE\abs{\sqrt{N}X_i}^l\leq\nu_l$. Then, for any $p \geq 1$
\begin{align}
\bE\abs{\bx_N^H\bA_N\bx_N-\frac{1}{N}\tr\bA_N}^p\leq\frac{C_p}{N^{p/2}}\p{\frac{1}{N}\tr\bA_N\bA_N^H}^{p/2}\pp{\nu_4^{p/2}+\nu_{2p}}
\end{align}
where $C_p$ is a constant depending only on $p$. Also, if $\by_N\in\mathbb{C}^N$ is another random vector with i.i.d. complex entries with zero mean, variance $1/N$, bounded $l$th order moment $\bE\abs{\sqrt{N}Y_i}^l\leq\nu_l$, and independent of $\bx_N$, then:
\begin{align}
\bE\abs{\bx_N^H\p{\bA_N-\frac{1}{N}\tr\bA_N}\by_N}^p\leq\frac{C_p}{N^{p/2}}\p{\frac{1}{N}\tr\bA_N\bA_N^H}^{p/2}\pp{\nu_2^{p}+\nu^2_{p}}.
\end{align}
\end{lemma}

\begin{lemma}[\cite{coulbook,Sebastian}][Trace Lemma]\label{aplem:4}
Let $\p{\bA_N}_{N\geq1}$, $\bA_N\in\mathbb{C}^{N\times N}$, be a sequence of random matrices and $\p{\bx_N}_{N\geq1} = \pp{X_{1,N},\ldots,X_{N,N}}^T\in\mathbb{C}^{N}$, a sequence of random vectors of i.i.d. entries, statistically independent of $\p{\bA_N}_{N\geq1}$. Assume that $\bE\ppp{X_{i,j}}=0$, $\bE\ppp{\abs{X_{i,j}}^2}=1$, $\bE\ppp{\abs{X_{i,j}}^8}<\infty$, and that $\bA$ has bounded spectral norm (in the a.s. sense). Then, a.s.,
\begin{align}
\frac{1}{N}\bx_N^H\bA_N\bx_N-\frac{1}{N}\tr\bA_N\to0.
\end{align}
\end{lemma}

\begin{lemma}[\cite{Peacock}]\label{aplem:5}
Let $\p{a_n}_{n\geq1},\p{b_n}_{n\geq1},\p{\bar{a}_n}_{n\geq1},\p{\bar{b}_n}_{n\geq1}$ be four infinite sequences of complex random variables. Assume that $a_n\asymp\bar{a}_n$ and $b_n\asymp\bar{b}_n$ in the a.s. sense. 
\begin{itemize}
\item If $\abs{a_n}$, $\abs{\bar{b}_n}$ and/or $\abs{\bar{a}_n}$, $\abs{b_n}$ are a.s. bounded, then a.s.,
$$
a_nb_n\asymp \bar{a}_n\bar{b}_n.
$$
\item If $\abs{a_n}$, $\abs{\bar{b}_n}^{-1}$ and/or $\abs{\bar{a}_n}$, $\abs{b_n}^{-1}$ are a.s. bounded, then a.s., 
$$
a_n/b_n\asymp \bar{a}_n/\bar{b}_n.
$$ 
\end{itemize}
\end{lemma}

\begin{lemma}[\cite{coulbook,Sebastian}]\label{aplem:6}
Let $\p{\bA_N}_{N\geq1}$, $\bA_N\in\mathbb{C}^{N\times N}$, be a sequence of matrices with uniformly bounded spectral norm, and $\p{\bB_N}_{N\geq1}$, $\bB_N\in\mathbb{C}^{N\times N}$ be random Hermitian,
with eigenvalues $\lambda_1\leq\ldots\leq\lambda_N$ such that, with probability one, there exist $\epsilon>0$ for which $\lambda_1>\epsilon$ for all large $N$. Then, for $\bv_N\in\mathbb{C}^N$,
\begin{align}
\frac{1}{N}\tr\bA_N\bB_N^{-1}-\frac{1}{N}\tr\bA_N\p{\bB_N+\bv_N\bv_N^H}^{-1}\to0
\end{align}
a.s. as $N\to\infty$, where $\bB_N^{-1}$ and $\p{\bB_N+\bv\bv^H}^{-1}$ are assumed to exist with probability 1. 
\end{lemma}

\begin{lemma}[\cite{Silverstein12}][Rank-1 Perturbation Lemma]\label{aplem:tracedif}
Let $z\in\mathbb{C}\setminus\mathbb{R}^+$, $\bA\in\mathbb{C}^{N\times N}$ and $\bB\in\mathbb{C}^{N\times N}$ where $\bB$ is Hermitian nonnegative definite, and $\bx\in\mathbb{C}^N$. Then,
\begin{align}
\abs{\tr\p{\p{\bB-z\bI_N}^{-1}-\p{\bB+\bx\bx^H-z\bI_N}^{-1}}\bA}\leq\frac{\norm{\bA}}{\text{dist}\p{z,\mathbb{R}^+}}
\end{align}
where $\text{dist}\p{\cdot,\cdot}$ denotes the Euclidean distance.
\end{lemma}
The following result can be found in \cite[Th. 2.3]{contimap}.
\begin{lemma}[The continuous mapping theorem]\label{continuouslemma} 
Let $\Phi:\mathbb{R}\to\mathbb{R}$ be an almost-everywhere continuous mapping, and let $\ppp{J_i}$ be a sequence of real-valued random variables that converges weakly to a real-valued random variable $J$. Then, $\ppp{\Phi\p{J_i}}$ converges weakly to the real-valued random variable $\Phi\p{J}$. 
\end{lemma}
The following result can be found in \cite[Theorem 2.20]{contimap}.
\begin{lemma}[Portmanteau's lemma (extended version)]\label{aplem:DCT2}
Suppose that $\p{X_n}_n$ is a sequence of nonnegative random variables for which $X_n\to X_\infty$ a.s. as $n\to\infty$, where $\bE X_\infty<\infty$. Then, $\bE X_n\to\bE X_\infty$ as $n\to\infty$ if and only if $\p{X_n}_n$ is uniformly integrable, that is, if, for each $\varepsilon>0$, there exists $c = c(\varepsilon)$ such that 
\begin{align}
\limsup_{n\to\infty}\bE\ppp{\abs{X_n}\Ind_{\ppp{\abs{X_n}\geq c}}}<\varepsilon.
\end{align}
\end{lemma}

\begin{lemma}\label{aplem:finite}
Let $\bx_N\in\mathbb{C}^N$ be a random vector with i.i.d. entries each with zero mean and unit variance, and let $\bA_N\in\mathbb{C}^{N\times N}$ such that $\tr\pp{\p{\bA_N^H\bA_N}^{1/2}}$ is uniformly bounded for all $N$. Then, for any finite $p$,
\begin{align}
\bE\abs{\bx_N^H\bA_N\bx_N}^p<\infty
\end{align}
for all $N$.
\end{lemma}
\begin{proof}
By Jensen's inequality we may write that
\begin{align}
\bE\abs{\bx_N^H\bA_N\bx_N}^p&\leq 2^{p-1}\p{\bE\abs{\bx_N^H\bA_N\bx_N-\tr\bA_N}^p+\abs{\tr\bA_N}^p}<\infty\nonumber
\end{align}
where the second inequality follows from the facts that: the first term in the r.h.s. is bounded by Lemma \ref{aplem:trace}, and the second term is bounded by assumption due to the fact that $\abs{\tr\bA_N}\leq\tr\pp{\p{\bA_N^H\bA_N}^{1/2}}$.
\end{proof}

\ifCLASSOPTIONcaptionsoff
  \newpage
\fi
\bibliographystyle{IEEEtran}
\bibliography{strings}
\end{document}